\documentclass{lmcs}


\usepackage{hyperref}

\theoremstyle{plain} %

\usepackage{hyperref}

\usepackage{amssymb} %

\newcommand{\off}[1]{ }
\newcommand{\backup}[1]{\off{1}}

\newcommand{\mentionind}[2]{\index{#1}#2}

\newcommand{\alarm}[1]{} 

\newcommand{\Var}{\textbf{\textit{Var}}}
\newcommand{\Val}{\textbf{\textit{Val}}}     \def\Val{\textbf{\textit{Val}}}
\newcommand{\Ref}{\textbf{\textit{Ref}}}
\newcommand{\Exp}{\textbf{\textit{Exp}}}     \def\Expr{\textbf{\textit{Exp}}}   
\newcommand{\Typ}{\textbf{\textit{Typ}}}
\newcommand{\Dom}{\textbf{\textit{Dom}}}     \newcommand{\kD}[0]{\textbf{\textit{Dom}}}
\newcommand{\Nam}{\textbf{\textit{Nam}}}     
\newcommand{\Pri}{\textbf{\textit{Pri}}}        
\newcommand{\Lev}{\textbf{\textit{Lev}}}
\newcommand{\Pse}{\textbf{\textit{Pse}}}     
     
\newcommand{\Flo}{\textbf{\textit{Flo}}}

\def\SecFPCLTplus{\mathcal{FPC}^{+}_\kw{LT}}
\def\SecFPC{\mathcal{FPC}}
\def\SecDNI{\mathcal{DNI}}

\def\SecNDN{\mathcal{NDN}}
\def\SecDND{\mathcal{DND}}

\def\Proc{\mathcal{P}\kern-1pt\mathit{roc}}

\def\semhigh{\mathcal{H}}
\def\semconf{\mathcal{CTC}}   
\def\semconfLT{\mathcal{CLT}}
\def\semconfTC{\mathcal{CTC}}

\newcommand{\klmeet}[0]{\curlywedge}
\newcommand{\kljoin}[0]{\curlyvee}
\newcommand{\kltop}[0]{\mho}
\newcommand{\klbot}[0]{\Omega}
\newcommand{\klpreceq}{\preccurlyeq}

\renewcommand{\preceq}[0]{\sqsubseteq}
\newcommand{\join}[0]{\sqcup} 
\newcommand{\meet}[0]{\sqcap} 

\newcommand{\Fjoin}[1]{\sqcup^{#1}} 
\newcommand{\Fmeet}[1]{\sqcap^{#1}} 
\newcommand{\Ftop}[1]{\top^{#1}}
\newcommand{\Fbot}[1]{\bot^{#1}}
\newcommand{\Fpreceq}[1]{\sqsubseteq^{#1}}

\def\upclo#1#2{#2\uparrow_{#1}}   
\def\secpolcon#1{\lceil#1\rceil}

\def\bool{\mathsf{bool}}

\def\rf{\mathsf{ref}}
\def\vrai{\mathit{tt}}
\def\faux{\mathit{ff}}
\def\unit{\mathsf{unit}}

\def\nil{(\!)}

\newcommand{\xarr}[2]{\xrightarrow[\vspace{50pt}#1]{\vspace{-50pt}#2}} 

\def\rarr{\twoheadrightarrow}

\def\hbisim{\mathcal{H}}          

\newcommand{\memeqF}[3]{=^{{#1}}_{#2,#3}}
\newcommand{\biseqt}[3]{{\approx}^{#1,#2}_{#3}}
\newcommand{\biseqtold}[3]{{\dot{\approx}}^{#1,#2}_{#3}}
\newcommand{\sbiseqt}[3]{{\sim}^{W,#1,#2}_{#3}}

\def\loctyp{\textsc{Loc}}
\def\vartyp{\textsc{Var}}
\def\abstyp{\textsc{Abs}}
\def\reftyp{\textsc{Ref}}
\def\dereftyp{\textsc{Der}}

\def\assigntyp{\textsc{Assign}}
\def\niltyp{\textsc{Nil}}

\def\boolttyp{\textsc{BT}}
\def\boolftyp{\textsc{BF}}

\def\condtyp{\textsc{Cond}}
\def\apptyp{\textsc{App}}
\def\seqtyp{\textsc{Seq}}
\def\rectyp{\textsc{Rec}}
\def\flowtyp{\textsc{Flow}}

\def\allowtyp{\textsc{Allow}}

\def\migtyp{\textsc{Mig}}

\def\loctypI{\textsc{Loc}$_{\textsc{I}}$}
\def\vartypI{\textsc{Var}$_{\textsc{I}}$}
\def\abstypI{\textsc{Abs}$_{\textsc{I}}$}
\def\reftypI{\textsc{Ref}$_{\textsc{I}}$}
\def\dereftypI{\textsc{Der}$_{\textsc{I}}$}
\def\assigntypI{\textsc{Ass}$_{\textsc{I}}$}
\def\niltypI{\textsc{Nil}$_{\textsc{I}}$}
\def\boolttypI{\textsc{Bt}$_{\textsc{I}}$} 
\def\boolftypI{\textsc{Bf}$_{\textsc{I}}$} 

\def\condtypI{\textsc{Cond}$_{\textsc{I}}$}
\def\apptypI{\textsc{App}$_{\textsc{I}}$}
\def\seqtypI{\textsc{Seq}$_{\textsc{I}}$}
\def\rectypI{\textsc{Rec}$_{\textsc{I}}$}
\def\migtypI{\textsc{Mig}$_{\textsc{I}}$}
\def\allowtypI{\textsc{Allow}$_{\textsc{I}}$}
\def\flowtypI{\textsc{Flow}$_{\textsc{I}}$}

\newcommand{\rfrt}[2]{#1 \rkw{ref}_{#2}}

\newcommand{\tef}[2]{{#1} , {#2}}		

\newcommand{\fra}[2]{\frac{\displaystyle\raisebox{1ex}{$#1$}}
                          {\displaystyle\raisebox{-1ex}{$#2$}}}
\newcommand{\lkw}[1]{{\mathrm{#1}}~}
\newcommand{\rkw}[1]{{~\mathrm{#1}}}
\newcommand{\mkw}[1]{{~\mathrm{#1}}~}
\newcommand{\kw}[1]{{\mathrm{#1}}}

\renewcommand{\ell}[0]{l}
\newcommand{\fA}[0]{\mathcal{A}}
\newcommand{\fB}[0]{\mathcal{B}}

\newcommand{\fH}[0]{\mathcal{H}}

\newcommand{\fL}[0]{\mathcal{L}}

\newcommand{\fS}[0]{\mathcal{S}}

\newcommand{\fR}[0]{\mathcal{R}}
\newcommand{\fT}[0]{\mathcal{T}}

\newcommand{\secmap}[1]{\Sigma_{#1}}
\newcommand{\typenv}[0]{\Gamma}

\newcommand{\low}{\mathit{low}}

\newcommand{\EC}[1]{\mathrm{E}\mathbf{[}#1\mathbf{]}}
\newcommand{\cE}[1]{\mathrm{#1}}
\newcommand{\cC}[2]{\mathrm{#1}\mathbf{[}#2\mathbf{]}}

\newcommand{\dom}[1]{\kw{dom}(#1)}

\newcommand{\rn}[1]{\kw{rn}(#1)}
\newcommand{\dn}[1]{\kw{dn}(#1)}

\newcommand{\fv}[1]{\kw{fv}(#1)}

\renewcommand{\tt}[0]{\mathit{tt}}
\newcommand{\ff}[0]{\mathit{ff}}
\newcommand{\assign}[2]{{(#1 := #2)}}

\newcommand{\seq}[2]{{(#1;#2)}}

\newcommand{\cond}[3]{{(\lkw{if} #1 \mkw{then} #2 \mkw{else} #3)}}
\newcommand{\allowed}[3]{{(\lkw{allowed} #1 \mkw{then} #2 \mkw{else} #3)}}

\newcommand{\lam}[2]{{(\lambda #1.#2)}}
\newcommand{\app}[2]{{(#1 ~#2)}}

\newcommand{\rfr}[2]{(\kw{ref}_{#1}~ #2)}
\newcommand{\rfrl}[3]{(\kw{ref}_{#1,#2}~ #3)}
\newcommand{\deref}[1]{(\lkw{!} #1)}

\newcommand{\fix}[2]{(\varrho #1. #2)}

\newcommand{\threadnat}[3]{(\kw{thread}_{#1}~#2 \mkw{at} #3)}
\newcommand{\threadanot}[4]{(\kw{thread}_{#1}^{#4}~#2 \mkw{at} #3)}

\def\flow#1#2{(\lkw{flow} #1 \mkw{in} #2)}

\newcommand{\iconft}[3]{{\langle #2,#1,#3 \rangle}}

\newcommand{\confd}[2]{{\langle #1,#2 \rangle}}

\newcommand{\equas}[1]{ ~{\buildrel {\text{#1}} \over {\Leftrightarrow} }~ }

\newcommand{\byinov}[2]{ \vdash^{#1}}
\newcommand{\byund}[3]{ \vdash_{#1}^{#2}}

\newcommand{\byint}[3]{ \vdash_{#1,#3}^{#2}}

\def\extrf#1{\lceil#1\rceil}

\newcommand{\figline}[0]{\rule{\linewidth}{0.3mm} }

\newcommand{\substi}[3]{ {\{#1 \mapsto #2\}}#3 }
\newcommand{\update}[3]{ [#1:=#2]#3 }
\newcommand{\trou}[0]{[]}
\newcommand{\reff}[1]{ #1.r }
\newcommand{\weff}[1]{ #1.w }
\newcommand{\teff}[1]{ #1.t }

\newcommand{\hide}[1]{}   %
\newcommand{\cut}[1]{}

\newcommand{\duvida}[1]{\textbf{}} 
\newcommand{\NEW}[0]{}

\newcommand{\eft}[3]{{\langle #1,#2,#3 \rangle}}	%
\newcommand{\var}[1]{#1}	%

\newcommand{\annot}[1]{\kw{annot}^{\secmap{},\typenv}(#1)}

\newcommand{\trans}[0]{\hookrightarrow}
\newcommand{\semvdash}[0]{\vdash^{\secmap{},\Upsilon}} 
\newcommand{\klpreceqtyp}[0]{\klpreceq}
\newcommand{\kleqtyp}[0]{\approx}
\newcommand{\klpseudominus}[0]{\smile}  %

\newcommand{\realbiseq}[1]{{\sim}^{{#1}}}

\newcommand{\myexample}[1]{\[#1\]} %
\newcommand{\numbexample}[1]{\begin{equation}#1\end{equation}} %

\begin{document}

\bibliographystyle{alpha}

\title{Information flow in a distributed security setting} %

\author[A. Almeida Matos]{Ana Almeida Matos}
\address{Instituto de Telecomunica\c c\~ oes (SQIG) and
Instituto Superior T\'ecnico, Lisbon, Portugal}
\email{ana.matos@ist.utl.pt}

\thanks{Partially funded by FCT within Project Elven (POCI-01-0145-FEDER-016844), 
  Project 9471 - Refor\c car a Investiga\c c\~ao, o Desenvolvimento Tecnol\'ogico e a Inova\c c\~ao (Project 9471-RIDTI) and by Fundo Comunit\'ario Europeu FEDER.  The first author is also partially funded by FCT within Project DeDuCe (PTDC/CCI-COM/32166/2017)}

\author[J. Cederquist]{Jan Cederquist}
\address{Instituto de Telecomunica\c c\~ oes (SQIG) and
  Instituto Superior T\'ecnico, Lisbon, Portugal}
\email{jan.cederquist@ist.utl.pt}
%

\date{empty}

\begin{abstract}
\noindent Information flow security is classically formulated in terms of the absence of illegal information flows, with respect to a security setting consisting of a single flow policy that specifies what information flows should be permitted in the system.  However, this security model does not suit scenarios of distributed and mobile computing where security policies are decentralized and distinct policies might simultaneously apply to different parts of a program depending on their dynamically changing location.  The challenge of managing compliance to heterogeneous security policies is further fledged when migrating programs conform to varying policies at different points of their executions, as in when declassification must be performed.
  
We investigate the security issues that emerge in \emph{distributed security settings}, where each computation domain establishes its own local security policy, and where programs may exhibit location-dependent behavior.  In particular, we study the interplay between two distinct flow policy \emph{layers}:  the \emph{declared flow policy}, established by the program itself, and the \emph{allowed flow policy}, established externally to the program by each computation domain.  We refine two security properties that articulate how the behaviors of programs comply to the respective flow policies:
\emph{Distributed Non-disclosure}, for \emph{enabling} programs to declare locally scoped flow policies;
and \emph{Flow Policy Confinement}, for \emph{controlling} the flow policies that are declared by programs.

We present enforcement mechanisms that are based on type and effect systems, ranging from purely static mechanisms to hybrid combinations with dynamic migration control, for enforcing the above properties on an expressive ML-like language with concurrent threads and code migration, and which includes an allowed flow policy construct that dynamically tests the allowed flow policy of the current context.

Finally, we show that the combination of the above two properties guarantees that actual information flows do not violate the relevant allowed flow policies.  To this end we propose and use the \emph{Distributed Non-Interference} property, a natural generalization of Non-Interference to a distributed security setting that ensures that information flows in a program respect the allowed flow policy of the domains where they originate.
\end{abstract}

\maketitle

\section{Introduction}

Crucial services of today's societies, such as health care, energy, communication, transportation, finance, security and defense are increasingly supported by and dependent of the well-functioning of interconnected networks of computing devices that can be seen as a cyber ecosystem~\cite{US11}.  %
At the same time, the impact of cyber attacks is continuously escalating, having become increasingly frequent and widespread, now holding an unprecedented potential for damage.  Indeed, the new possibilities offered by existing distributed computing technologies are exploited by parties with hazardous intentions, by means of programs that are able to automatically propagate throughout computation sites, persistently occupy those sites, and set up communication pathways for conducting cyber attacks in a concerted manner.  In other words, they utilize the global computing paradigm, and must be tackled using mechanisms that act with at least the same level of sophistication~\cite{US11}.  
The present paper proposes to contribute to this goal by setting the distributed security problem in terms of global computing concepts, by means of mechanisms drawn from language based security.

Access control serves a crucial role in cyber-security, by ensuring that access to computer resources is only granted to subjects that hold the required clearances.  However, it does not offer the necessary refinement for controlling how programs manipulate the information that is held by a resource once access to it is granted.  
\emph{Information flow} security regards the control of illegal flows of information across resources of different security levels during program execution, such as the prevention of confidentiality violations that are entailed by flows from private to publicly available resources.
This kind of program behavior is established by means of dependencies that are encoded in programs, and can be detected by employing information flow analyses. %
The language-based approach to security has placed a lot of attention on the study of information flow properties and enforcement mechanisms in a variety of programming paradigms~\cite{SM03}.
Information flow properties range in strictness from pure absence of information leaks, known as Non-interference~\cite{GM82}, to more flexible properties that allow for \emph{declassification} to take place in a controlled manner~\cite{SS05}.  

\paragraph{\emph{Distributed security setting.}}

The specification of \emph{information flow policies}, or what information flows should be permitted during program execution~\cite{BDS15}, is classically formulated in terms of %
a single security setting consisting of a security labeling, that assigns security levels to programming language information resources, and a security lattice, that orders security levels according to their degree of confidentiality and/or integrity~\cite{Den76}.
However, in the context of distributed and mobile computing, security policies are decentralized, and distinct policies might simultaneously apply to different parts of a program depending on their location. Enforcement of confidentiality in networks must therefore deal with distributed security policies, since different \emph{computation domains} (or \emph{sites}) follow different security orientations.  This issue is particularly challenging when migrating programs are given the flexibility to perform declassifying operations, as these might be acceptable or not depending on the policy of the thread's current computation domain.   
For example, migrating programs that were conceived to comply to certain security policies do not necessarily respect those of the computational locations they might end up executing at.  
Furthermore, when programs consist of more than one thread running concurrently, the same program might need to comply to more than one security policy simultaneously.  
This problem seems to be beyond the grasp of single program constructs that can restrict by whom, when, what, or where in the program declassification can be performed \cite{SS05}, since here the question is: \emph{in which context?}  

So far, most studies of declassification have been directed towards local computation scenarios, thus overlooking security issues that result from heterogeneous security requirements.
In order to study this problem, we consider as a starting point a \emph{distributed security setting} where each computation domain establishes its own local flow policy in the form of a security lattice.  More precisely, we consider a refinement of the notion of flow policy, and study the compliance of programs, consisting of one or more migrating threads, to two distinct information flow policy \emph{layers}:  the \emph{declared flow policy}, established by the program itself, representing a specification mechanism for \emph{enabling} declassification; the \emph{allowed flow policy}, established externally to the program by each computation domain, as an actual boundary on information leaks that are performed by programs.  During computation, both of these flow policies can change dynamically.  While the former is changed via declassification declarations, the latter is changed via program migrations.

\paragraph{\emph{Flow policy awareness.}}

At the programming language level, we assume that intention to perform declassification is expressed by means of a purely declarative declassification-enabling construct, while inspection of the relevant allowed flow policies is facilitated by means of a flow policy context testing construct that provides the programmer with flow policy awareness.  Let us take a closer look at the concrete programming constructs that are considered in this paper for these roles.

In order to enable local dynamic changes to the valid flow policy, the programming language may be enriched with a \emph{flow declaration} construct
$\flow F M$
that simply declares the flow policy $F$ as valid in its scope $M$. %
It is then easy to declare more flexible flow policy environments for delimited blocks of code, as for instance the part of a program that is executed by authenticated users:
\numbexample{\cond {\textit{authenticated}} {\flow {F_\textit{permissive}} {M}} {N} }
This program declares that flows in $M$ conform to a policy that is extended by $F_\textit{permissive}$. %
In other words, $M$ may contain declassifications that comply~to~$F_\textit{permissive}$.

At the moment that a program is written, it might be hard to anticipate which flow policies will be imposed at execution time by the domains where the program will run.  In a distributed context with code mobility, the problem becomes more acute, since the computation site might change \emph{during} execution, along with the allowed flow policy with which the program must comply.  In order to provide programs with some awareness regarding the flow policy that is ruling in the current computation domain, we introduce the \emph{allowed-condition}, written $\allowed F M N$, that tests whether the flow policy $F$ is allowed by the current domain and executes branches $M$ or $N$ accordingly.  Programs can then offer alternative behaviors to be taken in case the domains they end up at do not allow declassifications of the kind they intended to perform:
\numbexample{\allowed {F_\textit{disclose\_secret}} {M} {\textit{plan\_B}}}
The allowed-condition brings no guarantees that the $\textit{plan\_B}$ of the above program does not disclose just as much as the $M$ branch. However, misbehaving programs can be rejected by the domains where they would migrate, so its chances of being allowed to execute are increased by covering portions of code containing declassifications with appropriate allowed-conditions.

When moving within a distributed security setting, the allowed-condition makes program behavior depend on location.  As a consequence, new information flows can be encoded using this language construct, and give rise to a form of illegal flow known as \emph{migration leak}, where a program's location is leaked~\cite{AC11}.

\paragraph{\emph{Modular security properties and enforcement mechanisms.}}

Separating the problems of declaring and of controlling declassification enables the decomposition of security requirements into specialized properties that can be defined and enforced independently.  Here we treat the former as an information flow control problem, in the frame of the \emph{non-disclosure} property~\cite{AB09}, while the latter is isolated as the problem of ensuring that declassifications that are performed by mobile code respect the flow policy that is allowed at the computation domain where they are performed.  

The non-disclosure property can be seen as a generalization of Non-interference.  It uses information provided by the program semantics describing which flow policies are valid at different points of the computation, to ensure that, at each step, all information flows comply with the declared flow policy.  The property serves to guarantee consistency between the declared flow policies and the actual flows that programs perform.  It does not impose any limit on the usage of declarations.  The imposition of one or more levels of control can be treated as an independent problem.

Enforcing compliance of declassifications with allowed flow policies of computation domains raises new challenges when code moves in a distributed security setting. For instance, a computation domain $d$ might want to impose a limit to the flexibility of the flow declarations that it executes, and prevent incoming code from containing:
\numbexample{\flow {F_{\textit{all\_is\_allowed}}} M}
In the above example, the flow declaration validates any declassifications that might occur in $M$, regardless of what is considered acceptable by $d$.  However, in order to reason about whether this declaration is acceptable, it is necessary to take into consideration what are the security policies of the domains where this code might run.
We formulate the problem of ensuring that programs' threads can only declare declassifications that comply with the allowed flow policy of the domain where they are performed in terms of confinement with respect to an allowed flow policy.  %
This property is not concerned with the actual leaks that are taking place, but only with the declared intentions of enabling a declassification.

To sum up, two security properties that articulate the compliance of program behavior to the respective flow policies are studied in this paper:
\begin{enumerate}
\item \emph{Distributed Non-disclosure}, which requires information leaks that are encoded by programs to respect the declared flow policies of the context in which they are executing, in a language that includes a notion of locality, migration capabilities and location-dependent behavior.  This property is intended as a tool for \emph{enabling} the programmer to introduce declassification in the program. 

\item \emph{Flow Policy Confinement} requires that declassifications that are used by programs respect the allowed flow policy of the context in which they execute.  This property formalizes a simple form of declassification \emph{control}.
\end{enumerate}
Both properties can be enforced statically by means of a type and effect system.  However, the latter depends on the allowed flow policies of each domain that the program might migrate to.  This motivates the proposal to shift part of the burden of the declassification control to runtime, in order to improve efficiency and precision of the mechanism.
We argue for this shift by presenting mechanisms, centered on migration control, that range from purely static to hybrid combinations with dynamic features. %

\paragraph{\emph{Declassification control as migration control.}}

The problem of controlling the use of declassification across a distributed security setting can be reduced to a migration control problem, by preventing programs from migrating to sites if they would potentially violate that site's allowed flow policy.  To achieve this, domains must be able to check incoming code against their own allowed flow policies, ideally assisted by certificates that are carried by the program, and then decide upon whether those programs should be let in.
We thus address the technical problem of inferring what are the flexible flow policies that a program might set up for its own executions.  %
A certificate could consist of information about all the flow policies that are declared in the program and are \emph{not} validated by the allowed-branch of an allowed-condition. %
We call this flow policy the \emph{declassification effect} of the program.  Then, while the program
\numbexample{\allowed {F_1} M {\flow {F_2} N}}
would have a declassification effect that includes $F_2$ -- meaning that it should only be allowed to run in domains where $F_2$ is allowed --, the program
\numbexample{\allowed {F} {\flow {F} M} N}
(assuming that $M$ and $N$ have no flow declarations) would have an empty  declassification effect -- meaning that it could be safely allowed to run in any domain.

We propose an efficient hybrid mechanism based on statically annotating programs with the declassification effect of migrating code.  This is done by means of an informative type and effect pre-processing of the program, and is used for supporting runtime decisions.

\paragraph{\emph{Semantic Validation.}}

The coherence of the above properties should be supported by showing that combining Distributed Non-disclosure and Flow Policy Confinement has the intended meaning.  Ultimately, it is desirable to establish the absence of illegal information flows for programs that satisfy both properties.  In the presence of a single global allowed flow policy, this amounts to classic Non-interference.  However, in a distributed security setting, it is necessary to accommodate the reality of that distributed and mobile programs must obey different security settings at different points of their computation, depending on their location. This raises the question of what are the relevant allowed flow policies that the actual information flows should comply to, in order to be legal.
More concretely, which site or sites should have a say on the legality of flows performed by a thread that, during computation, reads and writes to resources as it migrates between different sites?  We argue that confidentiality of information that is \emph{read} at a certain site should be preserved according to that site's allowed flow policy.  
This idea leads to what we propose as the natural generalization of Non-interference to a distributed security setting, which we call \emph{Distributed Non-interference} property, that ensures that propagation of information respects the allowed flow policy of the domains where the information originates.

This paper proposes a generalized definition of Non-interference that accommodates the reality of that distributed and mobile programs must obey different security settings at different points of their computation, depending on their location.  It is formalized for a simple and general network model where computation domains are units of abstract allowed information flow policies. %

\paragraph{\emph{Contributions.}}

This paper presents a simple language-based framework for studying information flow security in distributed security settings in the presence of code mobility.  In this model, computation domains are units of \emph{allowed flow policies}, which have a scope that is local to each domain.
While the formulation of the security properties is largely language-independent, a concrete language is defined and considered for the purpose of examples and as a target to the proposed enforcement mechanisms. %
It consists of an expressive distributed higher-order imperative lambda-calculus with remote thread creation.  The latter language feature implies in particular that programs might need to comply to more than one dynamically changing allowed flow policy simultaneously.  
The main technical contributions are:
\begin{enumerate}

\item A new programming construct $\allowed F M N$ that tests the flexibility of the allowed flow policy imposed by the domain where it is currently located and can act accordingly.

\item A refinement of the Non-disclosure for Networks property~\cite{AC11} that is more suitable for settings where migration is subjective (i.e., cannot be induced by external threads), called \emph{Distributed Non-disclosure}\footnote{The new designation for the property highlights the fact that it is intended for distributed security settings, independently of whether the underlying memory and communication model is closer to that of a full-fledged network, or is reduced to simply logically distributed computation domains.}.

\item A new security property, named \emph{Flow Policy Confinement}, that regards the compliance of declassification operations that are performed by programs to the valid allowed flow policies where they take place.  

\item A comparative study of three enforcement mechanisms for flow policy confinement in the form of migration control mechanisms for deciding whether or not certain programs should be allowed to execute at each site.  These are based on a type and effect system, and differ on the emphasis that is placed on static and runtime effort:

\begin{enumerate}
\item A purely static type and effect system for enforcing flow policy confinement.\label{one}  %
\item A type and effect system for checking migrating threads at runtime, that is more precise than the one in point~\ref{one}. \label{two}
\item A static-time informative pre-processing type and effect system for annotating programs with a \emph{declassification effect}, for a more efficient and precise mechanism than the one in point~\ref{two}.  
\label{three}
\end{enumerate}

\item A new information flow property, named Distributed Non-interference, that naturally generalizes classical Non-interference to Distributed security settings.

\item A study of the semantic coherence between the proposed definition of distributed Non-interference and other relevant information flow properties, namely (local) Non-interference, Distributed Non-disclosure and Flow Policy Confinement.

\end{enumerate}
This paper revises, unifies and expands work that is presented in the conference articles~\cite{Alm09},~\cite{AC13}, and part of~\cite{AC14}.
Proofs that make use of techniques that are not novel to this paper are omitted for space reasons, but are available in the Appendix for convenience of the reviewer.

\paragraph{\emph{Outline of the paper.}}

We start by defining the formal security and language setting of the paper (Section~\ref{sec-setting}).  Two main sections follow, dedicated to the security analyses of Distributed Non-disclosure (Section~\ref{sec-iflow}) and Flow Policy Confinement (Section~\ref{sec-confinement}).  In each, the formal properties are proposed (Subsections~\ref{subsec-iflow-property} and~\ref{subsec-confinement-property}), enforcement mechanisms are presented (Subsections~\ref{subsec-iflow-typesystem} and \ref{subsec-confinement-static} to~\ref{subsec-confinement-decleffect}), and their soundness is proved.
In the latter, we study the efficiency and precision of three type and effect-based mechanisms for enforcing confinement by means of migration control, that place different weight over static and run time.
We then propose a definition of Distributed Non-interference against which we verify the semantic coherence between the proposed and other relevant information flow properties, namely (local) Non-interference, Distributed Non-disclosure and Flow Policy Confinement (Section~\ref{sec-distnonint}).
Finally we discuss related work (Section~\ref{sec-related}) and conclude (Section~\ref{sec-concl}).

\section{Distributed Security Setting} \label{sec-setting}

\subsection{Security Setting} \label{subsec-secsetting} 

We adopt the classic representation of a security policy as a \emph{security lattice} of security levels~\cite{Den76}, corresponding to security clearances, such as read-access rights. The lattice defines a \emph{flow relation} between security levels which determines their relative degree of confidentiality.
Security levels can then be associated %
to observable information holders in the programming language by means of a \emph{security labeling}, which maps resources to security levels.
Hence, a security lattice represents a baseline information \emph{flow policy}, establishing that information pertaining to resources labeled with $l_1$ can be legally transferred to resources labeled with $l_2$ only if $l_2$ is at least as confidential as $l_1$ according to the lattice.
Other flow policies, which enable additional legal flow directions between security levels, can be seen as relaxations of the baseline policy.  Such flow policies can be formalized as \emph{downward closure operators} that collapse security levels of the baseline security lattice into lower ones~\cite{AS12}.

\subsubsection{Abstract requirements}

\paragraph{\emph{Baseline lattice of security levels.}}  %
Security levels $l,j,k \in \Lev$ are assumed to be structured according to their confidentiality by means of an abstract lattice $\fL = \langle \Lev, \preceq, \meet, \join, \top, \bot \rangle$, where: the partial-order $\preceq$, when relating two security levels $l_1,l_2$ as $l_1 \preceq l_2$, means that $l_2$ is at least as confidential %
as $l_1$; the meet operation $\meet$ gives, for any two security levels $l_1,l_2$, the most confidential security level that is at least as permissive as both $l_1$ and $l_2$; the join operation $\join$ gives, for any two security levels $l_1,l_2$, the least confidential security level that is at least as restrictive as both $l_1$ and $l_2$; the highest confidential security level $\top$ is the most restrictive one; and the lowest confidential level $\bot$ is the most permissive one.

The information flow policy that is represented by the baseline security lattice $\fL$ asserts as legal those flows that respect the baseline flow relation $\preceq$, since when information flows from $l_1$ to $l_2$, its original confidentiality requirements are preserved.

\paragraph{\emph{Relaxed lattice of security levels.}}
A downward closure operator on a baseline lattice $\fL = \langle \Lev, \preceq, \meet, \join, \top, \bot \rangle$, is an operator $F:\Lev\rightarrow\Lev$ that is monotone, idempotent and restrictive (for every level $l\in\Lev, F(l)\preceq l$). The image of a downward closure operator $F$, equipped with the same order relation as $\fL$, is a sub-lattice of $\fL$, denoted by $F(\fL) = \langle F(\Lev), \Fpreceq{F}, \Fmeet{F}, \Fjoin{F}, \Ftop{F}, \Fbot{F} \rangle$, where, for every two security levels $l_1,l_2\in F(\Lev)$, and set of security levels $I\subseteq F(\Lev)$, we have that (i)~ $l_1\Fpreceq{F}l_2$ if $l_1\preceq l_2$; (ii)~ $\Fmeet{F}I=F(\meet I)$; (iii)~ $\Fjoin{F}I=\join I$; (iv)~ $\Ftop{F}=F(\top)$; and (v)~ $\Fbot{F}=\bot$.

The new more general flow relation $\Fpreceq{F}$ that is determined by the flow policy $F$ now enables the information flows that are allowed by $F$.

\paragraph{\emph{Lattice of flow policies.}}
Flow policies $A,F \in \Flo$ can be ordered according to their permissiveness by means of a \emph{permissiveness relation} $\klpreceq$, where $F_1 \klpreceq F_2$ means that $F_1$ is at least as permissive as $F_2$.
We assume an abstract lattice of flow policies that supports a pseudo-subtraction operation $\langle \Flo, \klpreceq, \klmeet, \kljoin, \kltop, \klbot, \klpseudominus \rangle$, where: the meet operation $\klmeet$ gives, for any two flow policies $F_1,F_2$, the strictest policy that allows for both $F_1$ and $F_2$; the join operation~$\kljoin$ gives, for any two flow policies $F_1,F_2$, the most permissive policy that only allows what both $F_1$ and $F_2$ allow; the most restrictive flow policy~$\kltop$ does not allow any information flows; and the most permissive flow policy~$\klbot$ that allows all information flows.  
Finally, the pseudo-subtraction %
operation $\klpseudominus$ between two flow policies $F_1$ and $F_2$ \footnote{This operation is used for refining the static analysis of the policy-testing construct, and is not a requirement of the security properties that are studied here.} represents the most permissive policy that allows everything that is allowed by the first ($F_1$), while excluding all that is allowed by the second ($F_2$); it is defined as the relative pseudo-complement of $F_2$ with respect to $F_1$, i.e. the greatest $F$ such that $F \klmeet F_2 \klpreceq F_1$. %

The lattice of flow policies can be achieved using downward closure operators.  Given the lattice $\fL$, the set of all downward closure operators on $\Lev$ form a lattice $\langle \Flo, \klpreceq, \klmeet, \kljoin, \kltop, \klbot, \klpseudominus \rangle$, where: (i) $F_1\klpreceq F_2 \Leftrightarrow \forall l\in\Lev.F_1(l)\preceq F_2(l)$; (ii)~ $(\klmeet K)(l) = l$ if $\forall F\in K.F(l)=l$; (iii)~ $(\kljoin K)(l) = \join\{ F(l) \| F\in K \}$; (iv)~ $\kltop(l) = \bot$; and (v)~ $\klbot(l) =l$,

This lattice can be interpreted as a lattice of relaxations of the original security setting.

\subsubsection{Concrete example} \label{subsubsec-concrete}

The following concrete security setting  meets the abstract requirements defined above and provides helpful intuitions. %

\paragraph{\emph{Baseline lattice of security levels.}}
The security levels are the subsets of the principals, $l \subseteq \Pri$, similar to read-access lists.    In this setting, security levels are ordered by means of the flow relation $\supseteq$.  
\paragraph{\emph{Lattice of flow policies.}}
Flow policies then consist of binary relations on $\Pri$, which can be understood as representing additional directions in which information is allowed to flow between principals:  a pair $(p,q)\in F$, most often written $p\prec q$, is to be understood as ``information may flow from $p$ to $q$''.  New more permissive security lattices are obtained by collapsing security levels into possibly lower ones, by closing them with respect to the valid flow policy.
Writing $F_1 \klpreceq F_2$ means that $F_1$ allows flows between at least as many pairs of principals as $F_2$.  The relation is here defined as ${F_1 \klpreceq F_2}$ iff ${F_2 \subseteq F_1^*}$ (where $F^*$ denotes the reflexive and transitive closure of $F$):
The meet operation is then defined as $\klmeet = \cup$, the join operation is defined as $F_1 \kljoin F_2 = F_1^* \cap F_2^*$, the top flow policy is given by ${\kltop = \emptyset}$, the bottom flow policy is given by $\klbot = \Pri \times \Pri$, and the pseudo-subtraction operation is given by~$\klpseudominus~=~-$ (set subtraction).

\paragraph{\emph{Relaxed lattice of security levels.}}
In order to define $\Fpreceq{F}$ we use the notion of $F$-\textit{upward closure} of a security level $\ell$, defined as $\upclo{F}{\ell} = \{q ~|~ \exists p\in\ell.\ p~F^*~ q\}$.
The $F$-upward closure of $\ell$ contains all the principals that are allowed by the policy $F$ to read information labeled $\ell$. A more permissive flow relation can now be derived as follows~\cite{ML98,AB09}:
\myexample{l_1 \Fpreceq{F} l_2 ~~\equas{def}~~ \forall q \in l_2~.~\exists p \in l_1 ~:~ p ~F^*~ q   ~\equas{}~  (\upclo F {l_1}) \supseteq (\upclo F {l_2})}
Furthermore, ${l}_1 \Fmeet{F} {l_2} = {l}_1 \cup {l}_2$ and ${l}_1 \Fjoin{F} {l}_2 = (\upclo {F}{{l}_1}) \cap (\upclo {F}{{l}_2})$, $\Ftop{F} = \emptyset$ and $\Fbot{F} = \Pri$.

Notice that $\Fpreceq{F}$ extends $\supseteq$ in the sense that $\Fpreceq{F}$ is larger than $\supseteq$ and that $\Fpreceq{\emptyset} \ =\ \supseteq$.  In other words, for the base security lattice (where the flow policy parameter is $\kltop$), the flow relation coincides with reverse inclusion of security levels, while the join operator is simply given by $\Fjoin{\emptyset} = \cap$.

\subsection{Language Setting}   \label{subsec-language}

We now present the basic language requirements to which the technical developments of this paper apply.  %
We then define a concrete instance of the language that suits these requirements. It will be used for providing illustrative examples, and as a target for the enforcement mechanisms. %

\subsubsection{Abstract requirements} \label{subsubsec-setting-labelings}

\begin{figure}
\centering
\[
\begin{array}{c}
\begin{array}{lrcllrcllcrl}
\textit{Security Levels}  &l,j  &\in&\Lev  \qquad  
  &\textit{Reference Names} &a,b        &\in&\Ref \\ %
\textit{Flow Policies}    &A,F      &\in&\Flo
  &\textit{Thread Names}    &m,n      &\in&\Nam \\
\textit{Types}            &\tau,\sigma,\theta&\in&\Typ
  &\textit{Domain Names}    &d        &\in&\Dom
\end{array}\\[7mm]
\begin{array}{lrcllrcllcrl}
\textit{Values}          &V      &\in&\Val\\
\textit{Expressions}     &M, N &\in&\Exp
\end{array}
\end{array}
\]
\caption{Syntax of basic elements of the language}  \label{fig-syntaxbasiceles}
\figline
\end{figure}

\emph{Networks}\footnote{We adopt, from the global computing community, the term \emph{networks} to designate an interconnected structure of \emph{computation domains}, which consist of hosts where a number of processes compute over resources.} %
are flat juxtapositions of \emph{domains}, each hosting the execution of a
pool of threads under the governance of a local allowed flow policy.  
Information is associated to globally accessible references, which are information holders for values of a designated type. 
As threads can move between computation domains during their execution, their location, or current domain, also carries information.

The basic elements of the language are summarized in Figure~\ref{fig-syntaxbasiceles}.  The names of references, threads and domains are drawn from disjoint countable sets $a, b \in \Ref$, $m,n \in \Nam$, and $d \in \kD \neq \emptyset$, respectively.
As mentioned in Subsection~\ref{subsec-secsetting}, security levels are associated to information holders by means of security labelings.  We define two labelings:
\begin{itemize}
\item A \emph{reference labeling} $\secmap{} : \Ref \rightarrow \Lev \times \Typ$, whose left projection $\secmap{1}: \Ref \rightarrow \Lev$ corresponds to the usual \emph{reference security labeling} that assigns security levels to references, and whose right projection $\secmap{2} : \Ref \rightarrow \Typ$ corresponds to the \emph{type labeling} that determines the type of values that can be assigned to each reference.
\item A \emph{thread security labeling} $\Upsilon{} : \Nam \rightarrow \Lev$, that assigns security levels to thread names.  This mapping represents the security level of the knowledge of the position of a thread in the network.
\end{itemize}
The security labelings are used in the security analysis that is performed in Section~\ref{sec-iflow}.
In the context of examples, the mapping between reference or thread names and their corresponding security annotations and types may be informally denoted as subscript of names. %

\emph{Threads} consist of named expressions (drawn from $\Exp$) and run concurrently in \emph{pools} $P : \Nam \rightarrow \Exp$, which are mappings from thread names to expressions.  Threads %
are denoted by $M^m \in \Exp\times\Nam$, and pools are denoted as sets of threads.
\emph{Stores} $S : \Ref \rightarrow \Val$, or \emph{memories} map reference names to values. %
\emph{Position-trackers} $T : \Nam \rightarrow \Dom$, map thread names to domain names, and are used to keep track of the locations of threads in the network.  The pool $P$ containing all the threads in the network, the position tracker $T$ that keeps track of their positions, and the store $S$ containing all the references in the network, form \emph{configurations}  $\iconft T P S$, over which the evaluation relation is defined in the next subsection. We refer to the pairs $\confd S T$ as \emph{states}, and pairs $\confd P T$ as \emph{thread configurations}. 
The flow policies that are allowed by each domain are kept by the \emph{allowed-policy mapping} $W : \Dom \rightarrow \Flo$ from domain names to flow policies.

The following basic notations are useful for defining properties and modifications to the elements of configurations.
For a mapping $Z$, we define $\dom{Z}$ as the domain of a given mapping~$Z$. %
We say a name is fresh in $Z$ %
if it does not occur %
in $\dom{Z}$.
Given an expression $M$, we denote by $\rn{M}$ and 
$\dn{M}$ the set of %
reference and domain names, respectively, 
that occur in $M$.  This notation is extended in the obvious way to pools of threads.  %
We let $\fv M$ be the set of variables occurring free in $M$.
We restrict our attention to \emph{well formed configurations} $\iconft T P S$ satisfying the conditions that
$\rn{P} \subseteq \dom{S}$, %
that $\dn{P} \subseteq \dom{W}$,
that $\dom{P} \subseteq \dom{T}$,
and that, for every $a \in \dom{S}$, $\rn{S(a)} \subseteq \dom{S}$ and $\dn{S(a)} \subseteq \dom{W}$.
We denote by $\substi{x}{X}M$ the capture-avoiding substitution of the pseudo-value $X$ for the free occurrences of $x$ in $M$.
The operation of adding or updating the image of an object $z$ to $z'$ in a mapping $Z$ is denoted~$\update z {z'} Z$.

The `$W \semvdash$' turnstile gives a security context to the definition of the semantics, making explicit the allowed flow policy of each domain in the network, and the valid reference and thread labelings.  These parameters are fixed, and are not central do this study, so they are omitted in the rest of the paper, written simply `$W \vdash$'.
\renewcommand{\semvdash}[0]{\vdash}
The reduction relation is a transition relation between (well-formed) configurations $\xarr{F}{d}$, which are decorated with %
the name of the domain $d$ where each step is taking place and the flow policy $F$ declared by the evaluation context where they are performed.  
The semantics needs not depend on this information, which is made available for the security analysis.
The relation~$\rarr$ denotes the reflexive closure of the transition relation $\xarr{F}{d}$. %

\subsubsection{Concrete object language}

The distributed language that we use is an imperative higher-order $\lambda$-calculus with reference creation, where we include a flow policy declaration construct %
(for directly manipulating flow policies~\cite{AB09}) and the new flow policy tester construct that branches according to whether a certain flow policy is allowed in the program's computing context, obtained by adding a notion of computing domain, to which we associate an allowed flow policy, and a code migration primitive.  Threads are also named in order to keep track of their position in the network.  Programs executing in different domains are subjected to different allowed flow policies -- this is what distinguishes local computations from global computations, and is the main novelty in this language.  
We opt for the simplest memory model, assuming memory to be shared by all programs and every computation domain, in a transparent form.  As we will see in Section~\ref{sec-iflow}, this allows us to focus on the effects of considering a distributed security setting, while avoiding synchronization issues that are not central to this work.  %

{
\begin{figure}   %
\centering
\[
\begin{array}{lrcllllllllllrcll}
  \textit{Values} 
  &V &::= &{\nil} ~|~ x ~|~ \lam x M ~|~ \mathit{tt} ~|~ \mathit{ff} ~|~ {a} \\[2mm] %
\textit{Pseudo-values} 
  &X &::= &V ~|~  \fix x X \\[2mm]

  \textit{Expressions}  	
  &M, N &::= &X ~|~ \app M N ~|~ \seq M N ~|~ \deref{N} ~|~ \assign M N ~|~ \rfrl{l}{\theta}M ~|\\
  &      &     & {\cond M {N_t} {N_f}} ~|~ { \flow F M } ~|\\
  &&&       {\allowed F {N_t} {N_f} } ~|~  {\threadnat l M d }
\end{array}
\]
\caption{Syntax of expressions}  \label{fig-syntaxexpressions}\label{fig-syntaxconfigurations}
\figline
\end{figure}
}

\paragraph{\emph{Syntax.}} %

The syntax of expressions, ranged over by $M,N \in \Expr$, is defined in Figure~\ref{fig-syntaxexpressions}.  It is based on an imperative higher order $\lambda$-calculus that includes declassification, a context-policy testing construct and remote thread creation.  The values of the language are the nil command $\nil$, variables ranged over by $x,y \in \Var$, function abstraction $\lam x M$, boolean values $\mathit{tt}$ and $\mathit{ff}$, and reference names.  Pseudo-values ranged over by $X$, extend values to include recursion, provided by the $\fix x X$ construct.  Other standard expressions are formed by function application $\app M N$, sequential composition $\seq M N$, the dereferencing operation $\deref M$, assignment $\assign M N$, reference creation $\rfrl{l}{\theta}M$ and conditional branching ${\cond M {N_t} {N_f}}$. 

Reference names are not associated to any security levels or types at the language level (in this aspect we depart from~\cite{Alm09}), and make use of the reference labeling $\secmap{}$ that is defined in Subsection~\ref{subsubsec-setting-labelings} only during the security analysis.  Nevertheless, reference names can be created at runtime, by a construct that is annotated with a type and security level that should be associated with the new reference.

The new features of the language are the flow declaration and the allowed-condition.  %
The flow declaration construct is written $\flow{F}{M}$, %
where $M$ is executed in the context of the current flow policy \textit{extended with} $F$; after termination the current flow policy is restored, that is, the scope of $F$ is $M$. The allowed-condition is similar to a standard boolean condition, with the difference that in $\allowed F {N_t} {N_f}$ the branches $N_t$ or $N_f$ are executed according to whether or not $F$ is allowed by the site's allowed flow policy. %

The remote thread creator $\threadnat l M d$ spawns a new thread with security level $l$ and expression $M$ in domain $d$, to be executed concurrently with other threads at that domain. 
It functions as a migration construct when the new domain of the created thread is different from that of the parent thread.  %
The security level $l$ is the confidentiality level that is associated to the %
position of the thread in the network.

  The following example illustrates the usage of the non-standard constructs of the language:
\numbexample{ \label{exallowed}
  (\threadnat H {(\lkw{allowed} {F_{H \prec L}}~\kw{then}~{\flow {F_{H \prec L}} {\assign {x_L} {\deref {y_H}}}}
  ~\kw{else}~{\textit{plan\_B}})} d)
}
The program creates a remote thread at domain~$d$, tests whether its allowed flow policy allows for the flows in $F_{H \prec L}$, which informally represents a flow policy that allows information to flow from level $H$ to level~$L$, and executes the first or second branch accordingly

\begin{figure}[t!]
\begin{equation*}
\begin{array}{c}
\begin{array}{rcll}

W \semvdash {\iconft{T}{\{\EC{\app {\lam {x} M} V}^m\}}S} &\xarr{\extrf{\cE{E}}}{T(m)}& {\iconft{T}{\{\EC{\substi x V M}^m\}} S} \\[1mm]

W \semvdash {\iconft{T}{\{\EC{\cond {\mathit{tt}} {N_t} {N_f}}^m\}}S} &\xarr{\extrf{\cE{E}}}{T(m)}& {\iconft{T} {\{\EC{N_t}^m\}}S}\\[1mm]

W \semvdash {\iconft{T} {\EC{\cond {\mathit{ff}} {N_t} {N_f}}^m\}} S} &\xarr{\extrf{\cE{E}}}{T(m)}& {\iconft{T} {\{\EC{N_f}^m\}} S}\\[1mm]

W \semvdash {\iconft{T} {\{\EC{\seq V N}^m\}} S} &\xarr{\extrf{\cE{E}}}{T(m)}& {\iconft{T} {\{\EC N^m\}} S}\\[1mm]

W \semvdash {\iconft{T}{\{\EC{\fix x X}^m\}}S}  &\xarr{\extrf{\cE{E}}}{T(m)}&  {\iconft{T}{\{\EC{\app {\substi x {\fix x X}} X}^m\}} S} \\[1mm]

W \semvdash {\iconft{T}{\{\EC{{{{\flow F V}}}}^m\}}{S}} &\xarr{\extrf{\cE{E}}}{T(m)}& {\iconft{T}{\{\EC{{{V}}}^m\}}{S}} \\[1mm]

W \semvdash {\iconft{T}{\{\EC{{{\deref {a}}}}^m\}}{S}} &\xarr{\extrf{\cE{E}}}{T(m)}& {\iconft{T}{\{\EC{{{S(a)}}}^m\}}{S}} \\[1mm]

W \semvdash {\iconft{T}{\{\EC{{{{\assign {a} V}}}}^m\}}{S}} &\xarr{\extrf{\cE{E}}}{T(m)}& {\iconft{T}{\{\EC{{{\nil}}}^m\}}{\update{a}V S}} \\[1mm]

W \semvdash {\iconft{T}{\{\EC{{{\rfrl{l}{\theta} V}}}^m\}}{S}} &\xarr{\extrf{\cE{E}}}{T(m)}& {\iconft{T}{\{\EC{{a}}^m\}}{\update {a} V S}}, \textit{ where}\\[-2mm]
&&\hspace{2pt}~a \textit{ fresh in } S\textit{ and }{\secmap{}(a)=(l,\theta)} %
\end{array}\\[3mm]
\begin{array}{c}
\fra{{W(T(m))\!\klpreceq\!F}}
{{W \semvdash {\iconft{T} {\{\EC{\allowed {F}{N_t}{N_f}}^m\}} S}} ~{\xarr{\extrf{\cE{E}}}{T(m)}}~ {{\iconft{T} {\{\EC{N_t}^m\}} S}}}\\
\fra{{W(T(m))\!\not\klpreceq\!F}}
{{W \semvdash {\iconft{T} {\{\EC{\allowed {F}{N_t}{N_f}}^m\}} S}} ~{\xarr{\extrf{\cE{E}}}{T(m)}}~ {{\iconft{T} {\{\EC{N_f}^m\}} S}}}\\[3mm]
\begin{array}{c}
{W \semvdash {\iconft{T}{\{{\EC{\threadnat{l}{N}{d}}}^{m}\}}{S}}\!\!\xarr{\extrf{\cE{E}}}{T(m)}\!\!{\iconft{\update {n} d T}{\{{\EC{\nil}}^{m},N^{n}\}}{S}}}, {\textit{ where}}\\[-2mm]
\hspace{210pt}{~n {\textit{ fresh in }}T\textit{ and }{\Upsilon{}(n)=l}}
\end{array}\\[3mm]
\fra {W \semvdash \iconft{T}{P}{S} \xarr{F}{d} {\iconft{T'}{P'}{S'}} ~~~ \iconft{T}{P \cup Q}{S} \textit{ is well formed}}
{W \semvdash \iconft{T}{P \cup Q}{S} \xarr{F}{d} {\iconft{T'}{P' \cup Q}{S'}}}
\end{array}
\end{array}
\end{equation*}
\caption{Operational semantics} \label{fig-semantics}
\figline
\end{figure}

\paragraph{\emph{Operational semantics.}}   %

In order to define the operational semantics, expressions are represented using \emph{evaluation contexts}, which specify a call-by-value evaluation order: %
\myexample{
\begin{array}{rl}
\cE{E} ::= &\trou ~|~ {\app {\cE{E}} N} ~|~ {\app V {\cE{E}}} ~|~\seq{\cE{E}}N ~|~ \rfrl{l}{\theta}{\cE{E}} ~|~ \deref{\cE{E}}
~|~   \assign{\cE{E}}N ~|~ \assign V {\cE{E}} ~|\\
& {\cond {\cE{E}} {N_t} {N_f}} ~|~ \boldsymbol{\flow F {\cE{E}}}
\end{array}
}
We write $\cC{E}{M}$ to denote an expression where the sub-expression $M$ is placed in the evaluation context $\cE{E}$, obtained by replacing the occurrence of $\trou$ in $\cE{E}$ by $M$.
The flow policy that is permitted by the evaluation context ${\cE{E}}$ is denoted by $\extrf{\cE{E}}$. It is a lower bound %
to all the flow policies that are declared by the context:
\myexample{
\extrf{\trou} = \kltop \qquad \extrf{\flow F{\cE{E}}} = F \klmeet \extrf{\cE{E}} \qquad
\extrf{\cE{E'}[\cE{E}]} = \extrf{\cE{E}} ~\textit{when $\cE{E'}$ has no flow declarations}
}

The small step operational semantics of the language is defined in Figure~\ref{fig-semantics}.  Most of the rules are standard, except for the allowed-policy mapping $W$, that parameterizes all rules, and the domain name and flow policy that decorate the transitions.  The $W$ parameter is used only by the allowed-condition, to retrieve the allowed flow policy of the domain where the condition is executed.  %
Notice that $W$ is never changed. Furthermore, the semantics does not depend on the flow policy that decorates the arrows.
Thread names are used in two situations:  When a new thread is created, its new fresh name is added to the position-tracker, associated to the parameter domain.  As with reference creation, the security level that is associated to the new thread does not influence the semantics, but is used later by the security analysis.  %
Thread names are also used when an allowed-condition is performed: the tested flow policy is compared to the allowed flow policy of the site where that particular thread is executing. %
The last rule establishes that the execution of a pool of threads is compositional (up to the choice of new names).  

According to the chosen semantics, %
dereferencing and assigning to remote references can be done transparently. In spite of this, since the allowed flow policies are distributed, the behavior of a program fragment may differ on different machines.
As an example, consider the thread
\numbexample{\label{ex-distributed}
{\allowed {F} {\assign {y_L} {1}}{\assign {y_L} {2}}}^{m}
}
running in a network with domains $d_1$ and $d_2$, where $W(d_1) \klpreceq F$ but $W(d_2) \not\klpreceq F$.  The thread will perform different assignments depending on whether $T(m)=d_1$ or $T(m)=d_2$.  
In Section~\ref{sec-iflow} we show that these possible behaviors are distinguishable by an information flow bisimulation relation.
This distributed behavior occurs regardless of the chosen memory model.

One can prove that the semantics preserves well-formedness of configurations, and that a configuration with a single thread has at most one transition, up to the choice of new names.

\section{Controlling Information Flow}   \label{sec-iflow}

As we have seen, the inclusion of the new allowed-condition construct in our programming language has introduced subtle security issues.  Indeed, even though the memory model that we consider is non-distributed, the allowed construct projects the distributed nature of the security policies into location-dependent program behavior.  In other words, the location of a thread in the network becomes an information flow channel.  It is then necessary to clarify what information flow dependencies are introduced by the construct, and to sort out what information flows should be allowed or not.

In this section we consider the problem of controlling information flow in a language that includes the new allowed-condition construct, in the context of a distributed security setting with concurrent threads and code mobility.  Information flows that violate the baseline security policy should conform to declared intentions to perform declassification.  To this end, we define the \emph{Distributed Non-disclosure} property, and analyze the security behavior of the new allowed-condition construct in light of the proposed property.  We show that the new formalization is more permissive and precise than \emph{Non-disclosure for Networks}~\cite{AC11}, and argue for the suitability of the proposed formalization for settings where migration is subjective (i.e., auto-induced), such as the present one.
We present a type and effect system for enforcing the property over the concrete language of Subsection~\ref{subsec-language} and establish soundness of the enforcement mechanism. %

\subsection{Distributed Non-disclosure} \label{subsec-iflow-property}

Non-disclosure states that, at each computation step performed by a program, information flows respect the flow policy that is declared by the evaluation context where the step is performed.  The property is naturally defined in terms of an information flow bisimulation~\cite{SS00,Smi01,BC02,FG95} on concurrent threads.  At each execution point, the bisimulation relates the outcomes of each possible step that is performed over states (stores) that are indistinguishable at a certain observation level, where the notion of indistinguishability is locally tuned to the flow policy that is currently valid for that particular step~\cite{AB09}.  Low-equality, the formalization of indistinguishability, is thus a crucial aspect of the definition of information flow bisimulations, and should reflect what are the observable resources at a given security level.  

As we will see towards the end of this section, in our setting the position of a thread in the network can reveal information about the values in the memory.  In other words, states include other information holders which are also observable (besides the store), namely the position of threads in the network.  Low-equality must then be generalized accordingly in order to relate states that include position trackers.

\paragraph{\emph{Low equality.}} 
We define a notion of indistinguishability between states that are composed of a store~$S$ and a position tracker~$T$, and is parameterized by the corresponding security labelings -- a reference security labeling~$\secmap{1}$ and a thread security labeling~$\Upsilon{}$ (for retrieving the security levels of references and threads).  Low-equality is defined relative to a given flow policy~$F$ and observation label~$l$.  Recall, from Sub-subsection~\ref{subsubsec-setting-labelings}, that the flow policy $F$ can be used to determine the permissiveness of the flow relation.
Intuitively, two stores are said to be low-equal at level~$l$, with respect to a flow policy $F$, if they coincide in the values of all references whose security levels are lower or equal than~$l$ (according to $F$). 
Similarly, two position trackers are said to be low-equal at level~$l$, with respect to a flow policy $F$, %
if they coincide in the location of all thread names whose security levels are lower or equal than~$l$ (according to $F$).  Low-equality between states is then defined point-wise on stores and position trackers:
\begin{defi}[Low-Equality] Given a reference security labeling~$\secmap{1}$ and a thread security labeling~$\Upsilon{}$, low-equality with respect to a flow policy $F$ and a security level~$l$ is defined between well-labelled stores, position-trackers and states as follows:
\begin{itemize}
\item  %
  ${S_1} \memeqF{\secmap{1},\Upsilon}{F}{l} {S_2}$, if
$\{ (a,V) ~|~ (a,V)\!\in\!{S_1} ~\&~ \secmap{1}(a) \Fpreceq{F} l \} = \{ (a,V) ~|~ (a,V)\!\in\!{S_2} ~\&~ \secmap{1}(a) \Fpreceq{F} l \}$
\item  %
  ${T_1} \memeqF{\secmap{1},\Upsilon}{F}{l} {T_2}$, if
$\{ (n,d) ~|~ (n,d)\!\in\!{T_1} ~\&~ \Upsilon{}(n) \Fpreceq{F} l \} = \{ (n,d) ~|~ (n,d)\!\in\!{T_2} ~\&~ \Upsilon{}(n) \Fpreceq{F} l \}$
\item  %
  $\confd{T_1}{S_1} \memeqF{\secmap{1},\Upsilon}{F}{l} \confd{T_2}{S_2}$, if
${S_1} \memeqF{\secmap{1},\Upsilon}{F}{l} {S_2}$ and ${T_1} \memeqF{\secmap{1},\Upsilon}{F}{l} {T_2}$.
\end{itemize}
\end{defi}
\renewcommand{\memeqF}[3]{=_{#2,#3}}
\noindent This relation is transitive, reflexive and symmetric.  In order to lighten the notation, the parameters $\secmap{}$ and $\Upsilon$ are omitted in the rest of the paper, written simply `$\memeqF{\secmap{1},\Upsilon}{F}{l}$'.

\paragraph{\emph{Store compatibility.}} The language defined in Section~\ref{subsec-language} is a higher-order language, where programs can build and execute expressions using values stored in memory. For example, the expression
\numbexample{
\app {\deref{a}} {\nil}
}
can evolve into an insecure program when running on a memory that maps a reference $a$ to a lambda-abstraction whose body consists of an insecure expression.  In order to avoid considering all such programs insecure, it is necessary to make assumptions concerning the contents of the memory.  Here, memories are assumed to be compatible to the given security setting and typing environment, requiring typability of their contents with respect to the type system that is defined in the next subsection (see Definition~\ref{def-iflow-compatibility}).  The assumption is used in the definition of bisimulation that follows (as well as in subsequent security definitions).

\paragraph{\emph{Information flow Bisimulation.}}
Intuitively, if a program is shown to be related to itself by means of an information flow bisimulation, one can conclude that it has the same behavior regardless of changes in the high part of the state.  In other words, the high part of the state has not interfered with the low part, i.e., no security leak has occurred.  A secure program can then be defined as one that is related to itself by an information flow bisimulation.  Resetting the state arbitrarily at each step of the bisimulation game accounts for changes that might be induced by threads that are external to the pools under consideration, thus enabling compositionality of the property.

Given that we are considering a setting with subjective migration (i.e. only the thread itself can trigger its own migration), changes in the position of threads in a given pool cannot be induced externally to that pool.  We can then focus on the behavior of threads when coupled with their possible locations in the network, and look at how the behavior of thread configurations is affected by changes in the high part of the stores.  Accordingly, the following information flow bisimulation fixes the position tracker across the bisimulation steps:
\begin{defi}[$\biseqt{W,\secmap{},\Upsilon}{\typenv}{\ell}$] \label{def-bisimdefNDN2}  %
Given a security level $\ell$, 
a $(W,{\secmap{}},{\Upsilon},{\typenv},{\ell})$-bisimulation is a symmetric relation $\fR$ on thread configurations that satisfies, for all $P_1,T_1,P_2,T_2$, and $({\secmap{}},{\typenv})$-compatible stores $S_1, S_2$
\myexample{
\confd{P_1}{T_1}~\fR~\confd{P_2}{T_2} ~\textit{and}~ W \vdash \iconft{T_1}{P_1}{S_1} \xarr{\textbf{F}}{d} \iconft{T_1'}{P_1'}{S_1'}
~\textit{and}~ \confd{T_1}{S_1} \memeqF{\secmap{1},\Upsilon}{\textbf{F}}{l} \confd{T_2}{S_2},
}
with $(\dom{{S_1}'}-\dom{S_1})\cap\dom{S_2}=\emptyset$ and $(\dom{{T_1}'}-$ $\dom{T_1})\cap\dom{T_2}=\emptyset$ implies that there exist ${P_2'}, {T_2'}, {S_2'}$~such that:
\myexample{
W \vdash \iconft{T_2}{P_2}{S_2} \rarr \iconft{T_2'}{P_2'}{S_2'} ~\textit{and}~ \confd{T_1'}{S_1'} \memeqF{\secmap{1},\Upsilon}{\kltop}{l} \confd{T_2'}{S_2'}
~\textit{and}~ \confd{P_1'}{T_1'}~\fR~\confd{P_2'}{T_2'}
}
Furthermore, $S_1',S_2'$ are still $(\secmap{},\typenv)$-compatible.
The largest $(W,{\secmap{}},\Upsilon,{\typenv},\ell)$-bisimulation, the union of all such bisimulations, is denoted~$\biseqt{W,\secmap{},\Upsilon}{\typenv}{\ell}$.
 \end{defi}
\renewcommand{\biseqt}[3]{{\approx}^{W,#2}_{#3}}
In order to lighten the notation, the parameters $\secmap{}$ and $\Upsilon$ are omitted in the rest of the paper, written simply `$\biseqt{W,\secmap{},\Upsilon}{\typenv}{\ell}$'.
For any %
${\typenv}$ and $\ell$, the set of pairs of thread configurations where threads are values is an %
$(W,{\typenv},{\ell})$-bisimulation.  Furthermore, the union of a family of %
$(W,{\typenv},\ell)$-bisimulations is a %
$(W,{\typenv},\ell)$-bisimulation. Consequently, $\biseqt{W,\secmap{},\Upsilon}{\typenv}{\ell}$ exists. %

We briefly recall intuitions that explain the above definition here, and refer the reader to~\cite{AC11} for more explanations. 
The reason why the above bisimulation potentially relates more programs than one for Non-interference is the stronger premise $\confd{T_1}{S_1}\memeqF{\secmap{1},\Upsilon}{F}{l} \confd{T_2}{S_2}$, which assumes pairs of states that coincide to a greater extent, thus facilitating the reproduction of the behavior by the opposite thread configuration.  The absence of a condition on the flow policy of the matching move for $P_2$ enables all expressions without side-effects (such as) to be bisimilar, independently of the flow policy that is declared by their evaluation contexts.
Clearly, the relation $\biseqt{W,\secmap{},\Upsilon}{\typenv}{\ell}$ is not reflexive since, as only ``secure'' programs, as defined next, are bisimilar to themselves.  For instance, the insecure expression $\assign{b_B}{\deref {a_A}}$ is not bisimilar to itself if $A\not\Fpreceq{F} B$.

We now present a new formalization of the non-disclosure property, defined over thread configurations, that is suitable for distributed settings.
\begin{defi}[Distributed Non-disclosure] \label{def-propertyDND}
A pool of threads $P$ satisfies the Distributed Non-disclosure property with respect to an allowed-policy mapping $W$, reference labeling $\secmap{}$, thread labeling $\Upsilon$ and typing environment $\typenv$, if it satisfies $\confd{P}{T_1}~\biseqt{W,\secmap{},\Upsilon}{\typenv}{\ell}~\confd{P}{T_2}$ for all security levels $\ell$ and position trackers $T_1,T_2$ such that $\dom{P}=\dom{T_1}=\dom{T_2}$ and $T_1 \memeqF{\secmap{1},\Upsilon}{\kltop}{l} T_2$.  We then write $P\in\SecDND(W,\secmap{},\Upsilon,\typenv)$.
\end{defi}
\noindent Distributed Non-disclosure is compositional with respect to set union of pools of threads, up to disjoint naming of threads, and a subjective migration primitive. %

We are considering a simplistic memory model where all of the network's memory is accessible at all times by every process in the network.  With this assumption we avoid \emph{migration leaks} that derive from synchronization behaviors \cite{AC11},
but migration leaks can be encoded nonetheless.  The idea is that now a program can reveal information about the position of a thread in a network by performing tests on the flow policy that is allowed by that site.
\begin{exa}[Migration Leak] \label{example-updatemigr}
  In the following program, the new thread will be created at (migrate to) domains $d_1$ or $d_2$ depending on the tested high value.
  \numbexample{
\begin{array}{ll}
\lkw{if} {\deref {a_H}} &\lkw{then} {\threadnat l {\allowed F {\assign {b_L} {1}}{\assign {b_L} {2}}} {d_1}}\\
&\hspace{-7pt}\lkw{else} {\threadnat l {\allowed F {\assign {b_L} {1}}{\assign {b_L} {2}}} {d_2}}
\end{array} \label{exallowedmigr}
}
If the allowed flow policies of these domains differ on whether they allow $F$, different low-assignments are performed, thus revealing high level information.  Therefore, the program is insecure with respect to Distributed Non-disclosure.
\end{exa}

\NEW{}
It should be clear that this property concerns only the matching between flow declarations and the leaks that are encoded in the program.  It does not restrict the usage of flow declarations. For example, %
\numbexample{ \label{ex-notDNI}
\threadnat l {\flow F {\assign {b} {\deref {a}}}} d
} always satisfies Distributed Non-disclosure when $F = \klbot$, even though it violates the allowed flow policy of domain $d$ if $\secmap{1}(a) \not\Fpreceq{W(d)} \secmap{1}(b)$.

\paragraph{\emph{Comparison with Non-Disclosure for Networks.}}

The definition of Distributed Non-disclosure consists of a technical refinement of the Non-Disclosure for Networks property~\cite{AC11}.  The main difference is that, while Non-Disclosure for Networks was defined over pools of threads, we now define Distributed Non-disclosure over thread configurations.  This means that, when imposing restrictions on the behaviors of bisimilar pools of threads,
the former definition resets the state arbitrarily at each step of the bisimulation game.  However, as we have pointed out, in a context where migration is subjective, resetting the position tracker arbitrarily is unnecessary.  In fact, it leads to a property that is overly restrictive.  In the following example, program ${M_\textit{insec}}$ can be a direct leak that is not placed within a flow declaration:
\numbexample{
\threadnat l {\allowed F {\nil} {M_\textit{insec}}} d   \label{ex-restrictive-NDN1}
}
The above program is intuitively secure if $W(d) \klpreceq F$ and insecure otherwise, as the body of the thread is known to be executed at domain $d$.   This is reflected by Distributed Non-disclosure as proposed in Definition~\ref{def-propertyDND}.  However, it is considered insecure by the former definition, %
as if the allowed-condition is executed over ``fresh'' thread configurations such that the thread is located at a domain where $F$ is \emph{not} allowed, then the branch with the illegal expression ${M_\textit{insec}}$ would be executed.  Another example of a secure program that is considered insecure according to the former definition,
but not by Definition~\ref{def-propertyDND}, is a variant of the one in Equation~(\ref{exallowedmigr}), where $d_1$=$d_2$, i.e. where the two branches are syntactically equal.

Distributed Non-disclosure of Definition~\ref{def-propertyDND} is strictly weaker than the thread pool-based Definition of Non-disclosure for Networks, %
in the sense that it considers more programs as secure.  %
In order to formalize this result, we denote by $\SecNDN(W,\secmap{},\Upsilon,\typenv)$ the set of pools of threads that satisfy the former definition: %
\begin{prop}\text{} \label{prop-DNDcomparison}
$\SecNDN(W,\secmap{},\Upsilon,\typenv) \subset \SecDND(W,\secmap{},\Upsilon,\typenv)$.
\end{prop}

\begin{figure}[t!]
\begin{equation*}
\begin{array}{c}

\text{[\niltyp]} ~
{\typenv \byint{j}{\secmap{}}{F} \nil : s, {{\unit}}}  \qquad  

\text{[\boolttyp]} ~
{\typenv \byint{j}{\secmap{}}{F} \vrai : s, {{\bool}}} \qquad

\text{[\boolftyp]} ~
{\typenv \byint{j}{\secmap{}}{F} \faux : s, {{\bool}}} \\[3mm]

\text{[\loctyp]} ~
{\typenv \byint{j}{\secmap{}}{F} {a} : s, {\rfrt {\secmap{2}(a)} {\secmap{1}(a)}}} \qquad

\text{[\vartyp]} ~
{\typenv,x:\tau \byint{j}{\secmap{}}{F} x : s, {\tau}} \\[3mm] %

\text{[\abstyp]} ~
\fra{\typenv, x:{\tau} \byint{j}{\secmap{}}{F} M : s,{\sigma}}
{\typenv \byint{j'}{\secmap{}}{F'} {\lam {x} {M}} : s', {\tau \xarr {j,F}{s} \sigma}} \qquad

\text{[\rectyp]} ~
\fra{\typenv,x:\tau \byint{j}{\secmap{}}{F} X : s,\tau}
{\typenv \byint{j}{\secmap{}}{F} \fix{x}X : s,\tau} \\[5mm]

\text{[\apptyp]} ~

\fra{
  \begin{array}{c}
    \typenv \byint{j}{\secmap{}}{F} {M} : {\tef {s} {\tau \xarr{j,F}{s'} \sigma}}\\
    \typenv \byint{j}{\secmap{}}{F} N : {{\tef{s''}{\tau}}}
      \end{array}
    ~~~
\begin{array}{rcl}  
s.t &\Fpreceq{F}& s''.w\\
 s.r, s''.r &\Fpreceq{F}&s'.w
\end{array}}
    {\typenv \byint{j}{\secmap{}}{F} {\app {M} {N}} : {\tef{s \join s' \join s'' \join \eft{\bot}{\top}{s.r \join s''.r}} \sigma} } \\[5mm]

    \text{[\seqtyp]} ~
\fra{
  \begin{array}{c}
    \typenv \byint{j}{\secmap{}}{F} {M} : {\tef s {\tau}}\\
    \typenv \byint{j}{\secmap{}}{F} N : {{\tef{s'}{\sigma}}}
  \end{array}
  ~~~  s.t \Fpreceq{F} s'.w}
{\typenv \byint{j}{\secmap{}}{F} {\seq {M} {N}} : {\tef{s \join s'} \sigma} } \qquad

\text{[\reftyp]} ~
\fra{\typenv \byint{j}{\secmap{}}{F} {M} : {\tef s {\theta}}  ~~~ s.r,s.t \Fpreceq{F} l} 
{\typenv \byint{j}{\secmap{}}{F} \rfr{l,\theta}{M} : s \join \eft {\bot}{l}{\bot},\rfrt \theta {l}} \\[5mm]

\text{[\assigntyp]} ~ 
\fra{
  \begin{array}{c}
    \typenv \byint{j}{\secmap{}}{F} {M} : {\tef s {\rfrt \theta {l}}} \\
    \typenv \byint{j}{\secmap{}}{F} N : {{\tef{s'}{\theta}}}
  \end{array}
  ~~~  
\begin{array}{rcl}
s.t &\Fpreceq{F}&s'.w \\
s.r, s'.r &\Fpreceq{F}&{l}
\end{array}}
{\typenv \byint{j}{\secmap{}}{F} {\assign {M} {N}} : {\tef{s \join s' \join {\eft \bot {l} {\bot}}} \unit} } \quad %

\text{[\dereftyp]} ~
\fra{\typenv \byint{j}{\secmap{}}{F} {M} : {\tef s {\rfrt \theta {l}}} }
{\typenv \byint{j}{\secmap{}}{F} \deref{M} : {\tef {s \join {\eft {l} {\top} {\bot}}} {\theta}}} \\[5mm]

\text{[\condtyp]} ~
  \fra{
    \begin{array}{c} 
    \typenv \byint{j}{\secmap{}}{F} M: {s,{\bool}} \\
    \typenv \byint{j}{\secmap{}}{F} N_t : {{s_t,{\tau}}} \\ 
    \typenv \byint{j}{\secmap{}}{F} N_f : {{s_f,{\tau}}}
\end{array}  ~~~  {s.r \Fpreceq{F} s_t.w, s_f.w}}
{\typenv \byint{j}{\secmap{}}{F} {\cond{M}{N_t}{N_f}} : {s \join s_t \join s_f \join \eft{\bot}{\top}{s.r},{\tau}}} \\[5mm]

\boldsymbol{\textbf{[\allowtyp]}} ~
\fra{
\begin{array}{c} 
\typenv \byint{j}{\secmap{}}{F} N_t : {{s_t,{\tau}}} \\ 
\typenv \byint{j}{\secmap{}}{F} N_f : {{s_f,{\tau}}}
\end{array} 
 ~~~  {j \Fpreceq{F} s_t.w, s_f.w}}   %
{\typenv \byint{j}{\secmap{}}{F} \allowed {F'} {N_t} {N_f} : {s_t \join s_f \join \eft{j}{\top}{j} },{\tau}} \\[5mm]

\boldsymbol{\textbf{[\flowtyp]}} ~
\fra{\typenv \byint{j}{\secmap{}}{{F} \klmeet F'} N : s,\tau}   %
{\typenv \byint{j}{\secmap{}}{F} \flow {F'} N : s, \tau}  \quad %

\boldsymbol{\textbf{[\migtyp]}} ~
\fra{
\typenv \byint{l}{\secmap{}}{\kltop} {{M}} : {\tef s {\unit}}}
{\typenv \byint{j}{\secmap{}}{F} \threadnat l M {d'} : \tef {\eft{\bot}{l \meet s.w}{\bot}} {\unit}} \qquad
	\end{array}
\end{equation*}
\caption{Type and effect system for checking Distributed Non-disclosure}	\label{fig-iflow-typesystem}
\figline
\end{figure}

\subsection{Type and Effect System}  \label{subsec-iflow-typesystem}

We now present a type and effect system~\cite{LG88} that accepts programs that satisfy Distributed Non-disclosure, as defined in Sub-subsection~\ref{subsec-iflow-property}.
The typing judgments used in Figure~\ref{fig-iflow-typesystem} have the form
\myexample{
\typenv\byint{j}{\secmap{}}{F} M:s,\tau
}
meaning that the expression $M$ is typable with type $\tau$ and security effect $s$ in the typing context $\typenv :\Var \rightarrow \Typ$, which assigns types to variables.  The turnstile has three parameters: 
(1) the reference labeling $\secmap{}$; 
(2) the flow policy \emph{declared by the context} $F$, represents the one that is valid in the evaluation context in which the expression $M$ is typed, and contributes to the meaning of operations and relations on security levels.
(3) the security level $j$ represents the confidentiality level associated to the thread that the expression $M$ is part of, i.e. the confidentiality level of the location of that thread in the network.

The security effect $s$ is composed of three security levels that are referred to by $s.r$, $s.w$ and $s.t$ (in this order), and can be understood as follows: $s.r$ is the \emph{reading effect}, an upper-bound on the security levels of the references that are read by $M$; $s.w$ is the \emph{writing effect}, a lower bound on the references that are written by $M$; $s.t$ is the \emph{termination effect}, an upper bound on the level of the references on which the termination of expression $M$ might depend.
According to these intuitions, in the type system the reading and termination levels are composed in a covariant way, whereas the writing level is contravariant.

Types have the following syntax ($t$ is a type variable):
\myexample{
  \begin{array}{rcl}
s \in \Lev \times \Lev \times \Lev ~ &::=& \eft{\reff{s}} {~\weff{s}} {~\teff{s}}\\
\tau, \sigma, \theta \in \Typ ~ &::=& ~ t ~|~ \unit ~|~  \bool ~|~ \rfrt \theta l ~|~ \tau \xarr{j,F}{s} \sigma
\end{array}
}
Typable expressions that reduce to $\nil$ have type $\unit$, and those that reduce to booleans have type $\bool$.
Typable expressions that reduce to a reference which points to values of type $\theta$ and has security level $l$ have the reference type $\rfrt {\theta}{\ell}$.  The security level $l$ is used to determine the effects of expressions that handle references.  
Expressions that reduce to a function that takes a parameter of type $\tau$, that returns an expression of type $\sigma$, and with a \textit{latent effect} $s$ \cite{LG88} have the function type $\tau \xarr{j,F}{s} \sigma$.  The latent effect is the security effect of the body of the function, while the latent flow policies are those that are assumed to hold when the function is applied to an argument, and the latent security level $j$ of the thread containing the expression that appears in the type of expressions that reduce to functions.

We use a lattice on security effects, that is obtained from the point-wise composition of three lattices of the security levels.  More precisely:
\myexample{
\begin{array}{c}
s\Fpreceq{F} s' ~~\equas{def}~~ s.r \Fpreceq{F} s'.r ~~\&~~ s'.w \Fpreceq{F} s.w ~~\&~~ s.t \Fpreceq{F} s'.t \\
s\Fjoin{} s' ~~\equas{def}~~ \eft{\reff{s}\Fjoin{}\reff{s'}} {~\weff{s}\Fmeet{}\weff{s'}} {~\teff{s}\Fjoin{}\teff{s'}}\qquad
\top = \eft{\top}{\bot}{\top} \qquad \bot =  \eft{\bot}{\top}{\bot}
\end{array}
}

Our type and effect system applies restrictions to programs in order to enforce compliance of all information flows to the flow policy that is declared by each context.  This is achieved by parameterizing the flow relation $\Fpreceq{F}$ with the flow policy $F$ that parameterizes the typing judgments.  The flow policy is updated by the rule $\flowtyp$, which allows to type the body of the flow declaration in a more permissive way.  Apart from the parameterization of the flow relation with the current flow policy, the conditions imposed in the premises, and the update of the security effects in the conclusion, are fairly standard in information flow type systems and enforce syntactic rules of the kind ``no low writes should depend on high reads'', both with respect to the values that are read, and to termination behaviors that might be derived.  We refer the reader to~\cite{AC11} for explanations on all of these conditions and updates, and focus next on the new aspects that are introduced in order to prevent new forms of information leaks that appear in our distributed security setting (such as Example~\ref{example-updatemigr}), and that deserve further attention.

 In the allowed-condition rule~$\allowtyp$, the security level~$j$ that is associated to each thread, and represents the confidentiality level of the position of the thread in the network, is used to update
the reading and termination effect.  This is because the choice of the branch (which depends on the threads location) can determine the final value and the termination behavior of the condition.  The security level of the thread is also constrained not to be higher in confidentiality than the level at which the state can be changed (as it could potentially leak information about the thread's position) in rule $\allowtyp$.  Finally, in rule $\migtyp$, the confidentiality level of the new thread is used to update the writing effect of thread creation, which causes a change in the position-tracker.

We can now define the compatibility predicate that applies to this particular information flow analysis:
\begin{defi}[$({\secmap{}},{\typenv})$-Compatibility] \label{def-iflow-compatibility}
A memory $S$ is said to be $({\secmap{}},{\typenv})$-\emph{com\-pa\-tib\-le} if, for every reference $a \in \dom{S}$, its value $S(a)$ satisfies ${\typenv \byint{j}{\secmap{}}{F} S(a) : s,\secmap{2}(a)}$, for every security level $j$, flow policy $F$, and security effect $s$.
\end{defi}

\subsubsection{Soundness.} %

The main result of this section, soundness, states that the type system only accepts expressions that are secure in the sense of Definition~\ref{def-propertyDND}. In the remainder of this section we sketch the main definitions and results that can be used to reconstruct a direct proof of this result. A similar proof is given in detail for a similar language (without the allowed-condition or remote thread creation) in~\cite{AC11}.  The main difference, besides the treatment of new language constructs, is the assumption of memory compatibility that is introduced in the present work.

\paragraph{\emph{Subject Reduction.}}

In order to establish the soundness of the type system of Figure~\ref{fig-iflow-typesystem} we need a Subject Reduction result, stating that the type of a thread is preserved by reduction. When a thread performs a computation step, some of its effects may be performed by reading, updating or creating a reference, and some may be discarded when a branch in a conditional expression is taken. Then the effects of an expression ``weaken'' along the computations.  This result is established by Proposition~\ref{prop-iflow-subjectreduction} below.  %
\begin{prop}[Subject Reduction]\label{prop-iflow-subjectreduction}
Consider a thread $M^{m}$ such that ${\typenv \byint{\Upsilon(m)}{\secmap{}}{F} M : s, \tau}$, and suppose that $W \vdash \iconft{T}{\{M^{m}\}} S$ $\xarr{\var {F'}}{d}$ $\iconft{T'} {\{M'^{m}\} \cup P} {S'}$, for a memory $S$ that is $({\secmap{}},{\typenv})$-compatible. Then, there is an effect $s'$ such that $s' \preceq s$ and $\typenv\byint{\Upsilon(m)}{\secmap{}}{F} M':s',\tau$, and $S'$ is also $({\secmap{}},{\typenv})$-compatible.
Furthermore, if~$P = \{N^n\}$, for some expression $N$ and thread name $n$, then there exists %
$s''$ such that $s'' \preceq s$ %
and $\typenv\byint{\Upsilon(n)}{\secmap{}}{\kltop} N:s'',\unit$.
\end{prop}
\begin{proof}
We follow the usual steps \cite{WF94}, where the main proof is a case analysis on the transition $W \vdash \iconft{T}{\{M^{m}\}} S$ $\xarr{\var {F}}{d}$ $\iconft{T'} {\{M'^{m}\} \cup P} {S'}$.
\end{proof}

\paragraph{\emph{Properties of the Semantics.}}

\hide{
One can show that if a thread is able to perform a step in one memory, when executing on a low-equal memory it can also perform a step and produce a low-equal result.  %
\begin{lem}[Guaranteed Transitions] \label{prop-iflow-guarantee}
Consider an allowed-policy mapping $W$, a thread $M^m$ and two states $\confd {T_1}{S_1}, \confd {T_2}{S_2}$ such that $W \vdash \iconft {T_1} {\{{M}^m\}} {S_1}$ $\xarr{F}{d}$ $\iconft {T_1'}{P_1'}{S_1'}$, and for some $F'$ we have $\confd {T_1}{S_1} \memeqF{\secmap{1},\Upsilon}{F \klmeet F'}{\low} \confd {T_2}{S_2}$.  Then $P_1' = \{M_1'^m\} \cup P'$ for some expression $M_1'$, and if $(\dom{{S_1}'}-\dom{S_1})\cap\dom{S_2}=\emptyset$ and $(\dom{{T_1}'}-$ $\dom{T_1})\cap\dom{T_2}=\emptyset$, we have that there exist $M_2'$, $T_2'$ and $S_2'$ such that $W \vdash \iconft {T_2}{\{{M}^m\}}{S_2} \xarr{F}{d} \iconft{T_2'}{\{M_2'^m\} \cup P}{S_2'}$ and $\confd {T_1'}{S_1'} \memeqF{\secmap{1},\Upsilon}{F \klmeet F'}{\low} \confd {T_2'}{S_2'}$.  Furthermore, if $P' = \{N^n\}$~for some expression $N$, then $M_1' = M_2'$.
\end{lem}
\begin{proof}
By case analysis on the proof of $W \vdash \iconft {T_1} {\{{M}^m\}} {S_1}$ $\xarr{F}{d}$ $\iconft {T_1'}{P_1'}{S_1'}$.
\end{proof}
}
If the evaluation of a thread $M^m$ differs in the context of two distinct states while not creating two distinct reference or thread names, this is because either $M^m$ is performing a dereferencing operation, which yields different results depending on the memory, or because $M^m$ is testing the allowed policy. %
\begin{lem}[Splitting Computations]\label{prop-iflow-split}\label{app-prop-iflow-split}\text{}\\ 
  Consider a thread $M^{m}$ such that $W \vdash \iconft{T_1}{\{M^m\}}{S_1}$ $\xarr{F}{d}$ $\iconft{T_1'}{P_1'}{S_1'}$ and $W \vdash \iconft{T_2}{\{M^m\}}{S_2}$ $\xarr{F'}{d}$ $\iconft{T_2'}{P_2'}{S_2'}$ with ${P_1}' \neq {P_2}'$. Then $P_1' = {\{{M_1}'^m\}}$, $P_2' = {\{{M_2}'^m\}}$ for some $M_1'$ and $M_2'$, and there exists $\cE{E}$ such that $F=\secpolcon{\cE{E}}=F'$, where either:
\begin{itemize}
\item $M=\cC{E}{\deref {a}}$ for some reference $a$, and $M_1'=\cC{E}{S_1(a)}$, $M_2'=\cC{E}{S_2(a)}$ with $\confd{T_1'}{S_1'}=\confd{T_1}{S_1}$ and $\confd{T_2'}{S_2'}=\confd{T_2}{S_2}$, or
\item $M=\cC{E}{\allowed {\bar F}{N_t}{N_f}}$ for some ${\bar F}, {N_t}, {N_f}$, and $T_1(m) \neq T_2(m)$ with $\confd{T_1'}{S_1'}=\confd{T_1}{S_1}$ and $\confd{T_2'}{S_2'}=\confd{T_2}{S_2}$.
\end{itemize}
\end{lem}
\begin{proof}
By case analysis on the transition $W \vdash \iconft{T_1}{\{M^m\}}{S_1}$ $\xarr{F}{d}$ $\iconft{T_1'}{P_1'}{S_1'}$.
Note that the only rules that depend on the state are those for the reduction of $\cC{E}{\deref {a}}$ and of $\cC{E}{\allowed {F'}{N_t}{N_f}}$.  
\end{proof}

We can identify a class of threads that have the property of never performing any change in the ``low'' part of the memory.  These are classified as being ``high'' according to their behavior:
\begin{defi}[Operationally High Threads] \label{def-operationallyhigh}\label{app-def-operationallyhigh}
Given a flow policy $F$ and a security level $l$, a set of threads $\fH$ is a set of \emph{operationally $(W,\secmap{},\Upsilon,\typenv,F,\ell)$-high threads} if the following holds for all $M^m \in \fH$, and for all states $\confd T S$ where $S$ is $({\secmap{}},{\typenv})$-compatible:
\myexample{W \vdash \iconft{T}{\{M^m\}}{S} \xarr{F'}{d} \iconft{T'}{P'}{S'} \text{ implies } \confd T S \memeqF{\secmap{1},\Upsilon}{F}{l} \confd {T'} {S'} \text{ and } P' \subseteq \fH
}
Furthermore, $S'$ is still $(\secmap{},\typenv)$-compatible.
The largest set of operationally $(W,\secmap{},\Upsilon,\typenv,F,\ell)$-high threads is denoted by $\semhigh^{W,\secmap{},\Upsilon,\typenv}_{F,\ell}$.  We then say that a thread $M^m$ is operationally $(W,\secmap{},\Upsilon,\typenv,F,\ell)$-high, if $M^m \in \semhigh^{W,\secmap{},\Upsilon,\typenv}_{F,\ell}$.
\end{defi}
In order to lighten the notation, the parameters $\secmap{}$ and $\Upsilon$ are omitted in the rest of the paper, written simply `operationally $(W,\typenv,F,\ell)$-high' and `$\semhigh^{W,\typenv}_{F,\ell}$'.
Remark that for any $W$, %
$\typenv$, $F$ and $\ell$, the set of threads with values as expressions is a set of operationally %
$(W,\typenv,F,\ell)$-high threads.  Furthermore, the union of a family of sets of operationally %
$(W,\typenv,F,\ell)$-high threads is a set of operationally %
$(W,\typenv,F,\ell)$-high threads. Consequently, %
$\semhigh^{W,\typenv}_{F,\ell}$ exists. %
Notice also that if $F'\subseteq F$, then $\semhigh^{W,\typenv}_{F,\ell} \subseteq \semhigh^{W,\typenv}_{F',\ell}$.

\paragraph{\emph{Soundness.}}

The following result shows that the behavior of typable high threads (i.e. those with a high security level) that is location sensitive (i.e. depend on their location) is operationally high.
\begin{lem}[Location Sensitive Typable High Threads] \label{prop-iflow-potentially}
For a given flow policy~$F$ and security levels $j$ and $\low$, consider a thread $M^m$ such that ${\typenv \byint{j}{\secmap{}}{F} M : s, \tau}$ and $M = \EC{\allowed {F'} {N_t}{N_f}}$ with $j \not\Fpreceq{F} \low$.  Then, for all allowed-policy mappings $W$ and thread labelings $\Upsilon$ such that $\Upsilon(m)=j$, we have that $M^m \in \semhigh^{W,\typenv}_{F,\low}$.
\end{lem}
\begin{proof} By induction on the structure of evaluation context $\cE{E}$, making use of the flow restrictions that are introduced by the corresponding typing rule in each case.  Notice that expressions that dereference or test the allowed flow policy at a high security level have high reading effect (and termination effect for the latter), while those that assign and create a reference or thread at a low security level have a low writing effect.
  We show how operationally high expressions can be syntactically composed from other operationally high expressions. %
\end{proof}

\begin{thm}[Soundness of Typing Distributed Non-disclosure] \label{prop-iflow-soundness}\label{app-prop-iflow-soundness}
  Consider a pool of threads $P$, an allowed-policy mapping $W$, a reference labeling $\secmap{}$, a thread labeling $\Upsilon$ and a typing environment $\typenv$.  If for all $M^{m} \in P$ there exist $s$, and $\tau$ such that ${\typenv \byint{\Upsilon(m)}{\secmap{}}{\kltop} M : s, \tau}$, then $P$ satisfies the Distributed Non-disclosure property, i.e. $P\in\SecDND(W,\secmap{},\Upsilon,\typenv)$.
\end{thm}
\begin{proof}  The proof follows the structure presented in~\cite{AC11}.  The differences in the proof stem mainly from the treatment of the new language constructs (remote thread creation and allowed-condition, instead of suspensive local memory access), and the assumption of memory compatibility, that is introduced in the present work.  Informally, the steps are the following:
  \begin{enumerate}
  \item We build a syntactic symmetric binary relation between expressions that are typable with a low termination effect, and whose terminating behaviors do not depend on high references.  This includes expressions that are typable with a low termination effect, and that have just performed a high dereference.\label{step-one}
    The binary relation should be a ``kind of'' strong bisimulation with respect to the transition relation $\xarr{F}{d}$ and the parameterized low-equality $\memeqF{\secmap{1},\Upsilon}{F}{l}$.  In particular, it should be such that if the evaluation of two related expressions, in the context of two low-equal stores should split (see Lemma~\ref{prop-iflow-split}), then the resulting expressions are still in the relation.
  \item We define a larger symmetric binary relation on all typable expressions that relates operationally high threads, and, similarly to the previous one, relates the results of the computations of two related expressions that are not operationally high in the context of two low-equal memories.  This includes all typable expressions that have just performed a high dereference.
    For non operationally high threads, the binary relation should again be a ``kind of'' strong bisimulation with respect to the transition relation $\xarr{F}{d}$ and the parameterized low-equality $\memeqF{\secmap{1},\Upsilon}{F}{l}$.  The result Location Sensitive Typable High Threads (Lemma~\ref{prop-iflow-potentially}) helps to show that if the evaluation of two related expressions, in the context of two low-equal stores should split (see Splitting Computations, Lemma~\ref{prop-iflow-split}), then the resulting expressions are still in the relation.
    \label{step-two}
  \item We exhibit a $(W,{\typenv},{\ell})$-bisimulation on thread configurations that extends the previous one to relate operationally high threads with terminated computations. \qedhere \label{step-three}
    \end{enumerate}
\end{proof}
\noindent Soundness is compositional, in the sense that it is enough to verify the typability of each thread separately in order to ensure Distributed Non-disclosure for the whole network.

\section{Controlling Declassification}  \label{sec-confinement}

In this section we study the formalization and enforcement of a security property according to which declassification policies performed by migrating threads must comply to the allowed flow policy of the domain in which they are performed.  We formalize the property, named Flow Policy Confinement, with respect to a distributed security setting and justify the chosen formalization. 
We study its enforcement by means of three migration control mechanisms, whose aim is to prevent migration of threads when the declassifications that they would potentially perform are not allowed by that domain.

The first enforcement mechanism consists of a purely static type and effect system (Subsection~\ref{subsec-confinement-static}) for checking confinement. This type system is inherently restrictive, as the domains where each part of the code will actually compute cannot in general be known statically.  Furthermore, it requires information about the allowed flow policies of all the domains in the networks to which analyzed programs might migrate.
We therefore present a more relaxed type and effect system to be used at runtime by the semantics of the language for checking migrating threads for confinement to the allowed flow policy of the destination domain (Subsection~\ref{subsec-confinement-runtime}).  Information about the allowed flow policies of the domains is not required at static time.  We show that this simple hybrid mechanism is more precise than the purely static one.
Finally, we propose a yet more precise type and effect system that statically computes information about the declassification behaviors of programs (Subsection~\ref{subsec-confinement-decleffect}).  This information will be used at runtime, by the semantics of the language, to support more efficient runtime checks that control migration of programs. We conclude by proving that the third annotation-based mechanism is strictly more permissive than the previous ones, while preserving the meaning of non-annotated programs.

\subsection{Flow Policy Confinement} \label{subsec-confinement-property}

The property of Flow Policy Confinement states that the declassifications that are declared by a program at each computation step comply to the allowed policy of the domain where the step is performed.
In a distributed setting with concurrent mobile code, programs might need to comply simultaneously to different allowed flow policies that change dynamically.  We deal with this difficulty by placing individual restrictions on each step that might be performed by a part of the program, taking into account the possible location where it might take place.

\paragraph{\emph{Memory compatibility.}}
Similarly to Subsection~\ref{subsec-iflow-property}, 
memories are assumed to be compatible to the given security setting and typing environment, requiring typability of their contents according to the relevant enforcement mechanism.
The memory compatibility predicate will be defined for each security analysis that is performed over the next three subsections (see Definitions~\ref{def-confinementI-compatibility}, \ref{def-confinementII-compatibility} and~\ref{def-confinementIII-compatibility}).\footnote{For the sake of simplifying the exposition, the compatibility predicate uses three parameters $W$, ${\secmap{}}$ and ${\typenv}$, as they are parameters of the type system used by Definition~\ref{def-confinementI-compatibility}, although in Definitions~\ref{def-confinementII-compatibility} and~\ref{def-confinementIII-compatibility}, the parameter $W$ is not needed.}

\paragraph{\emph{Confined thread configurations.}}
We define the property by means of a co-inductive relation on thread configurations~\cite{AC14}.  The location of each thread determines which allowed flow policy it should obey at that point, and is used to place a restriction on the flow policies that decorate the transitions. %
By using thread configurations, the formalization of the property is simplified with respect to~\cite{AC13}, which used the notion of \emph{located threads}, and a bisimulation-based definition.
\begin{defi}[$(W,\secmap{},\typenv)$-Confined Thread Configurations] \label{def-operationallyconf}
Given an allowed-policy mapping~$W$, a set $\semconf$ of thread configurations is a set of $(W,\secmap{},\typenv$)-\emph{confined thread configurations} if it satisfies, for all $P,T$, and $(W,\secmap{},\typenv)$-compatible stores~$S$:
\myexample{
\begin{array}{l}
\confd{P}{T} \in \semconf ~\textit{and}~ W \vdash \iconft{T}{P}{S} \xarr{F}{d} \iconft{T'}{P'}{S'} \quad \textit{implies} \quad
W(d) \klpreceq F ~\textit{and}~ \confd{P'}{T'} \in \semconf
\end{array}
}
Furthermore, $S'$ is still $(W,\secmap{},\typenv)$-compatible.  
The largest set of $(W,\secmap{},\typenv)$-confined thread configurations %
is denoted~$\semconf^{W,\secmap{},\typenv}$.
\end{defi}
\noindent For any $W$, ${\secmap{}}$ and $\typenv$, the set of thread configurations where threads are values is a set of $(W,\secmap{},\typenv)$-confined thread configurations.  Furthermore, the union of a family of $(W,\secmap{},\typenv)$-confined thread configurations is a $(W,\secmap{},\typenv)$-confined thread configurations. Consequently, $\semconf^{W,\secmap{},\typenv}$ exists. %

The property is now formulated for pools of threads.
\begin{defi}[Flow Policy Confinement]  \label{def-propertyFPC}
A pool of threads $P$ satisfies Flow Policy Confinement with respect to an allowed-policy mapping $W$, reference labeling $\secmap{}$ and typing environment $\typenv$, if all thread configurations satisfy $\confd P T \in \semconf^{W,\secmap{},\typenv}$.  We then write $P \in \SecFPC(W,\secmap{},\typenv)$.
\end{defi}
\noindent Notice that Flow Policy Confinement is parameterized by a particular mapping $W$ from domains to allowed flow policies.  This means that security is defined relative to $W$.  An absolute notion of security holds when $W$ is universally quantified.

\paragraph{\emph{Properties.}}

It should be clear that Flow Policy Confinement speaks strictly about what \emph{flow declarations} a thread can do \emph{while} it is at a specific domain.  In particular, it does not restrict threads from migrating to more permissive domains. %
It does not deal with information flows, and offers no assurance that information leaks that are encoded at each point of the program do obey the declared flow policies for that point.  For example, the program
in Equation~(\ref{ex-notDNI})
always satisfies flow policy confinement when $F = \kltop$, regardless of the levels of references $a$ and $b$. %
But, it violates Distributed Non-disclosure if $\secmap{1}(a) \not\Fpreceq{F} \secmap{1}(b)$, as well as the allowed flow policy of the domain $d$ if $\secmap{1}(a) \not\Fpreceq{W(d)} \secmap{1}(b)$.

Flow Policy Confinement, formulated for pools of threads according to Definition~\ref{def-propertyFPC}, is similar to the one in \cite{Alm09,AC13}, which was formulated for thread configurations.  The present one does not fix the initial position of threads that satisfy the property.  Furthermore, it is more clearly defined in terms of a co-inductive relation on thread configurations~\cite{AC14}, while the former version used a bisimulation on \emph{located threads}. 
In order to compare the two versions of the property, we denote by %
$\SecFPCLTplus(W,\secmap{},\typenv)$ %
the set of pools of thread that are secure, according to the former definition, when coupled with all possible initial position trackers. %
Then, the two versions of the property are equivalent, up to quantification over all possible initial position trackers.
\begin{prop} \label{prop-FPCcomparison}
  $\SecFPC(W,\secmap{},\typenv) = \SecFPCLTplus(W,\secmap{},\typenv).$
\end{prop}
\begin{figure}[t!]
\begin{equation*}
\hspace{-3mm}	\begin{array}{c}

\text{[\niltyp]} ~
{W;\typenv \byund{A}{\secmap{}}{} \nil : {{\unit}}}  \qquad  

\text{[\boolttyp]} ~
{W;\typenv \byund{A}{\secmap{}}{} \vrai : {{\bool}}} \qquad

\text{[\boolftyp]} ~
{W;\typenv \byund{A}{\secmap{}}{} \faux : {{\bool}}} \\[3mm]

\text{[\loctyp]} ~
{W;\typenv \byund{A}{\secmap{}}{} {a} : {\rfrt {\secmap{2}(a)} {}}} \qquad

\text{[\vartyp]} ~
{W;\typenv,x:\tau \byund{A}{\secmap{}}{} x : {\tau}} \\[3mm]

\text{[\abstyp]} ~
\fra{W;\typenv, x:{\tau} \byund{A}{\secmap{}}{} M : {\sigma}}
{W;\typenv \byund{A'}{\secmap{}}{} {\lam {x} {M}} : {\tau \xarr {A}{} \sigma}} \qquad

\text{[\rectyp]} ~
\fra{W;\typenv,x:\tau \byund{A}{\secmap{}}{} X : \tau}
{W;\typenv \byund{A}{\secmap{}}{} \fix{x}X : \tau} \\[5mm]

\text{[\apptyp]} ~
\fra{W;\typenv \byund{A}{\secmap{}}{} {M} : { {\tau \xarr{A}{} \sigma}} ~~~ W;\typenv \byund{A}{\secmap{}}{} N : {{{\tau}}}}
    {W;\typenv \byund{A}{\secmap{}}{} {\app {M} {N}} : { \sigma} } \\[5mm]
    
\text{[\seqtyp]} ~
\fra{W;\typenv \byund{A}{\secmap{}}{} {M} : {{\tau}} ~~~  W;\typenv \byund{A}{\secmap{}}{} N : {{{\sigma}}}}
{W;\typenv \byund{A}{\secmap{}}{} {\seq {M} {N}} : { \sigma} } \qquad %

\text{[\reftyp]} ~
\fra{W;\typenv \byund{A}{\secmap{}}{} {M} : { {\theta}}  } 
{W;\typenv \byund{A}{\secmap{}}{} \rfr{l,\theta}{M} : \rfrt \theta {}} \\[5mm]

\text{[\assigntyp]} ~ 
\fra{W;\typenv \byund{A}{\secmap{}}{} {M} : { {\rfrt \theta {}}} ~~~  W;\typenv \byund{A}{\secmap{}}{} N : {{{\theta}}}} 
{W;\typenv \byund{A}{\secmap{}}{} {\assign {M} {N}} : { \unit} } \qquad %

\text{[\dereftyp]} ~
\fra{W;\typenv \byund{A}{\secmap{}}{} {M} : {\rfrt \theta {}} }
{W;\typenv \byund{A}{\secmap{}}{} \deref{M} : {  {\theta}}} \\[5mm]

\text{[\condtyp]} ~
\fra{W;\typenv \byund{A}{\secmap{}}{} M: {{\bool}}  ~~~  
\begin{array}{c} 
W;\typenv \byund{A}{\secmap{}}{} N_t : {{{\tau}}} \\ 
W;\typenv \byund{A}{\secmap{}}{} N_f : {{{\tau}}}
\end{array} }
{W;\typenv \byund{A}{\secmap{}}{} {\cond{M}{N_t}{N_f}} : {{\tau}}}  \\[5mm] %

\boldsymbol{\textbf{[\allowtyp]}} ~
\fra{
\begin{array}{c} 
W;\typenv \byund{A \klmeet F}{\secmap{}}{} N_t : {{{\tau}}} \\ 
W;\typenv \byund{A}{\secmap{}}{} N_f : {{{\tau}}}
\end{array}}
{W;\typenv \byund{A}{\secmap{}}{} \allowed {F} {N_t} {N_f} : {\tau}}\\[5mm]

\boldsymbol{\textbf{[\flowtyp]}} ~
\fra{W;\typenv \byund{A}{\secmap{}}{} N : \tau  ~~~ A \klpreceq F}
{W;\typenv \byund{A}{\secmap{}}{} \flow {F} N : \tau} \qquad

\boldsymbol{\textbf{[\migtyp]}} ~
\fra{
W;\typenv \byund{W(d)}{\secmap{}}{} {{M}} : {\unit}}  %
{W;\typenv \byund{A}{\secmap{}}{} \threadnat l M {d}:  {\unit}}

	\end{array}
\end{equation*}
\caption{Type and effect system for checking Confinement} \label{fig-confinementI-typesystem}
\figline
\end{figure}

\subsection{Static Type and Effect System} \label{subsec-confinement-static}

We have seen that in a setting where code can migrate between domains with different allowed security policies, the computation domain might change \emph{during} computation, along with the allowed flow policy that the program must comply to.  This can happen in particular within the branch of an allowed-condition:
\numbexample{  
{\allowed F {\threadnat l {\flow F {M_1}} {d}} {M_2}} \label{ex-tricky}
}
In this program, the flow declaration of the policy $F$ is executed only if $F$ has been tested as being allowed by the domain where the program was started.  It might then seem that the flow declaration is guarded by an appropriate allowed construct.  However, by the time the flow declaration is performed, the thread is already located at another domain, where that flow policy might not be allowed.  
It is clear that a static enforcement of a confinement property requires tracking the possible locations where threads might be executing at each point.

Figure~\ref{fig-confinementI-typesystem} presents a new type and effect system~\cite{LG88} for statically enforcing confinement over a migrating program.  The type system guarantees that when operations are executed by a thread within the scope of a flow declaration, the declared flow complies to the allowed flow policy of the current domain.  
The typing judgments have the form
\myexample{
W; \typenv \byund{A}{\secmap{}}{} M:\tau
}
meaning that the expression $M$ is typable with type $\tau$ in the typing context $\typenv :\Var \rightarrow \Typ$, which assigns types to variables, in a context where $W$ is the mapping of domain names to allowed flow policies.  
The turnstile has two parameters: 
(1) the reference labeling $\secmap{}$, of which only the type labeling $\secmap{2}$, carrying the type of the references, is used; 
(2) the flow policy \emph{allowed by the context} $A$, which includes all flow policies that have been positively tested by the program as being allowed at the computation domain where the expression $M$ is running.

Types have the following syntax ($t$ is a type variable):
\myexample{
\tau, \sigma, \theta ~ \in ~  \Typ ~ ::= ~ t ~|~ \unit ~|~  \bool ~|~ \rfrt \theta {} ~|~ \tau \xarr{A}{} \sigma 
}
The syntax is similar to the one used in Sub-subsection~\ref{subsec-iflow-typesystem}, but is simpler:  The security level of references does not appear in the reference types $\rfrt{\theta}{}$, while the type function types $\tau \xarr{A}{} \sigma$ includes only the latent allowed flow policy, the one that is assumed to hold when the function is applied to an argument.

Our type and effect system applies restrictions to programs in order to enforce confinement of all flow declarations of a policy $F$ to be performed only once $F$ has been tested to be positively allowed by the domain's allowed flow policy.  This is achieved by means of the tested allowed flow policy $A$ that parameterizes the typing judgments, and by the condition $A \klpreceq F$ in the $\flowtyp$ rule.  
Flow declarations can only be performed if the declared flow policy is allowed by the flow policy that is tested by the context (\flowtyp).  Conversely, allowed-conditions relax the typing of the allowed-branch by extending the flow policy that is tested by the context with the policy that guards the condition (\allowtyp).  

Note that if an expression is typable with respect to an allowed flow policy $A$, then it is also so for %
any more permissive allowed policy $A'$.

We refer to the enforcement mechanism that consists of statically type checking all threads in a network according to the type and effect system of Figure~\ref{fig-confinementI-typesystem}, with respect to the allowed flow policies of each thread's initial domain, using the semantics represented in Figure~\ref{fig-semantics}, as \emph{Enforcement Mechanism~I}. 

To illustrate the restrictions that are imposed by the enforcement mechanism, we may consider program
\numbexample{
\allowed {F} {\flow {F_{H \prec L}} {\textit{plan\_A}}} {\textit{plan\_B}} \label{exallowed2}
}
where ${\textit{plan\_A}}$ and ${\textit{plan\_B}}$ have no declassifications, and that is running in domain $d$. The program is typable if  $W(d) \klmeet F \klpreceq {F_{H \prec L}}$.  If $F \klpreceq {F_{H \prec L}}$, then the program is always secure. Otherwise, the program is $W$-secure if $W(d) \klpreceq {F_{H \prec L}}$, or if $W(d) \not\klpreceq F$ (thanks to the semantics of the allowed-condition).

We are now in position to define the compatibility predicate that applies to Enforcement Mechanism~I.
\begin{defi}[$(W,{\secmap{}},{\typenv})$-Compatibility] \label{def-confinementI-compatibility}
A memory $S$ is said to be $(W,{\secmap{}},{\typenv})$-\emph{compatible} if, for every reference $a \in \dom{S}$, its value $S(a)$ satisfies the condition $W;{\typenv \byund{\kltop}{\secmap{}}{} S(a) : \secmap{2}(a)}$.
\end{defi}

\subsubsection{Soundness}

In order to establish the soundness of the type system of Figure~\ref{fig-confinementI-typesystem} we need a Subject Reduction result, stating that types that are given to expressions, along with compatibility of memories, are preserved by computation.

\begin{prop}[Subject Reduction]\label{prop-subjectreduction-confinementI}
Consider an allowed-policy mapping $W$ and a thread $M^{m}$ such that ${W;\typenv \byund{A}{\secmap{}}{} M : \tau}$, and suppose that $W \vdash \iconft{T}{\{M^{m}\}} S$ $\xarr{\var {F}}{d}$ $\iconft{T'} {\{M'^{m}\} \cup P} {S'}$, for a memory $S$ that is $(W,{\secmap{}},{\typenv})$-compatible. Then, $W;\typenv\byund{A \klmeet W(T(m))}{\secmap{}}{} M':\tau$, and $S'$ is also $(W,{\secmap{}},{\typenv})$-compatible.
Furthermore, if~$P = \{N^n\}$, for some expression $N$ and thread name $n$, then $W;\typenv\byund{W(T'(n))}{\secmap{}}{} N:\unit$.
\end{prop}
\begin{proof}
We follow the usual steps \cite{WF94}, where the main proof is a case analysis on the transition $W \vdash \iconft{T}{\{M^{m}\}} S$ $\xarr{\var {F}}{d}$ $\iconft{T'} {\{M'^{m}\} \cup P} {S'}$.
\end{proof}

Enforcement Mechanism~I guarantees security of networks with respect to confinement, as is formalized by the following result.  
\begin{thm}[Soundness of Enforcement Mechanism~I]  \label{prop-confinementI-soundness}        
Consider an allowed-policy mapping $W$, reference labeling $\secmap{}$, typing environment $\typenv$, and a thread configuration $\confd P T$ such that for all $M^{m} \in P$ there exists $\tau$ such that ${W;\typenv \byund{W(T(m))}{\secmap{}}{} M : \tau}$.  Then $\confd P T$ is a $(W,\secmap{},\typenv)$-confined thread configuration.
\end{thm}
\begin{proof} 
We show that the following is a set of $(W,\secmap{},\typenv)$-confined thread configurations:
\myexample{
C = \{ \confd{P}{T} ~|~ \forall M^m \in P, \exists \tau ~.~ {W;\typenv \byund{W(T(m))}{\secmap{}}{} M : \tau}\} %
}
By induction on the inference of ${W;\typenv \byund{W(T(m))}{\secmap{}}{} M : \tau}$. %
We use Subject Reduction (Proposition~\ref{prop-subjectreduction-confinementI}) and a case analysis on the last rule of the corresponding typing proof. %
\end{proof}

\paragraph{\emph{Safety and precision.}} %

The following result guarantees that typable threads ca execute until termination.
\begin{prop}[Safety] \label{prop-staticmigrcontrol-safety}
Given an allowed-policy mapping $W$, consider a closed thread $M^{m}$ such that ${W;\emptyset \byund{A}{\secmap{}}{} M : \tau}$.  Then, for any memory $S$ that is $(W,{\secmap{}},{\emptyset})$-compatible and position-tracker $T$, either the program $M$ is a value, or
$W\!\vdash\!\iconft{T}{\{M^{m}\}} S$ $\xarr{\var {F'}}{T(m)}$ $\iconft{T'} {\{M'^{m}\}\!\cup\!P} {S'}$, for some $F'$, $M'$, $P$, $T'$ and~$S'$.
\end{prop}
\begin{proof} By induction on the derivation of ${W;\typenv \byund{A}{\secmap{}}{} M : \tau}$.  \end{proof}

In face of a purely static migration control analysis, some secure programs are bound to be rejected.  
There are different ways to increase the precision of a type system, though they are all intrinsically limited to what can conservatively be predicted before runtime.  For example, for the program
\numbexample{ \label{ex-staticimprecision}
{\cond {\deref {a}} {\threadnat l {\flow {F} {M}} {d_1}} {\threadnat l {\flow {F} {M}} {d_2}}}
} %
it is in general not possible to predict which branch will be executed (or, in practice, to which domain the thread will migrate), for it depends on the contents of the memory.  %
It will then be rejected if $W(d_2) \not\klpreceq F$ or $W(d_1) \not\klpreceq F$.

\subsection{Runtime Type Checking}  \label{subsec-confinement-runtime}

{
\begin{figure}
\begin{equation*}
  \hspace{-3mm}
\begin{array}{c}
\vdots\\[2mm]
  \text{[\migtyp]} ~
  \fra{
\typenv \byund{\boldsymbol{\klbot}}{\secmap{}}{} {{M}} : {\unit} }  %
    {\typenv\byund{A}{\secmap{}}{} \threadnat l M {d}:{\unit}}\\[7mm]
    \vdots\\[2mm]
\fra{\boldsymbol{\typenv \byund{W(d)}{\secmap{}}{} {N} : {\unit}}}
    {W\semvdash{\iconft{T}{\{{\EC{\threadnat{l}{N}{d}}}^{m}\}}{S}} \xarr{\extrf{\cE{E}}}{T(m)}{{\iconft{\update {n}{d} {T}}{\{{\EC{\nil}}^{m},N^{n}\}}{S}}}}
\end{array}
\end{equation*}
\caption{Top: Relaxed type and effect system for checking Confinement (omitted rules are as in Figure~\ref{fig-confinementI-typesystem}).  Bottom:  Operational semantics with runtime type checking for migration control (omitted rules are as in Figure~\ref{fig-semantics}).} \label{fig-confinementII-typesystem}\label{fig-confinementII-semantics}
\figline
\end{figure}
}

In this subsection we study a hybrid mechanism for enforcing confinement, that makes use of a relaxation of the type system of Figure~\ref{fig-confinementI-typesystem}, at runtime.  
Migration is now conditioned by means of a runtime check for typability of migrating threads with respect to the allowed flow policy of the destination domain. The condition represents the standard theoretical requirement of checking incoming code before allowing it to execute in a given machine.

The relaxation is achieved by replacing rule~\migtyp~by the one in Figure~\ref{fig-confinementII-typesystem}.
The new type system no longer imposes \emph{future} migrating threads to conform to the policy of their destination domain, but only to the most permissive allowed flow policy $\klbot$.
The rationale is that it is only concerned about confinement of the non-migrating parts of the program.
This is sufficient, as all threads that are to be spawned by the program will be re-checked at migration time.

The proposed modification to the semantics defined in Figure~\ref{fig-semantics} appears in Figure~\ref{fig-confinementII-semantics}.  
The migration rule introduces the runtime check that controls migration ($n$ fresh in $T$).  The idea is that a thread can only migrate to a domain if it respects its allowed flow policy.
The new remote thread creation rule (our migration primitive), now depends on typability of the migrating thread.  The typing environment~$\typenv$ (which is constant) %
is now an implicit parameter of the operational semantics.  If only closed threads are considered, then also migrating threads are closed.
The allowed flow policy of the destination site now determines whether or not a migration instruction may be consummated, or otherwise block execution.  {{\em E.g.}}, the configuration
\numbexample{ \label{ex-migrate}
{\iconft{T}{\{{\EC{\threadnat l {\flow F M} d}}^{m}\}}{S}}
}
can only proceed if $W(d)$ allows for $F$; otherwise it gets stuck.

Notice that, thanks to postponing the migration control to runtime, the type system no longer needs to be parameterized with information about the allowed flow policies of all domains in the network, which is unrealistic.  The only relevant ones are those of the destination domain of migrating threads.  %

We refer to the enforcement mechanism that consists of statically type checking all threads in a network according to the type and effect system of %
Figure~\ref{fig-confinementII-typesystem} (left), with respect to the allowed flow policies of each thread's initial domain, using the semantics of %
Figure~\ref{fig-confinementII-semantics} (right), as \emph{Enforcement Mechanism~II}. 

Notice that Enforcement Mechanism~II restricts, on one hand, which programs are accepted to run, but also trims their possible executions, with respect to a given allowed-policy mapping $W$.  The program in Equation~(\ref{ex-staticimprecision}) illustrates this mechanism, as it is typable according to the relaxed type system (with respect to any $W$), but will block at the choice of the second branch if $W(d_2) \not\klpreceq F$.

We can now define the compatibility predicate that applies to Enforcement Mechanism~II.\footnote{As mentioned in Subsection~\ref{subsec-confinement-property}, for the sake of simplifying the exposition, the compatibility predicate used in Definition~\ref{def-operationallyconf} includes $W$ as a parameter, although it is not used in Definitions~\ref{def-confinementII-compatibility} and~\ref{def-confinementIII-compatibility} that follow.}
\begin{defi}[$(W,{\secmap{}},{\typenv})$-Compatibility] \label{def-confinementII-compatibility}
A memory $S$ is said to be $(W,{\secmap{}},{\typenv})$-\emph{compatible} if, for every reference $a \in \dom{S}$, its value $S(a)$ satisfies the typing condition ${\typenv \byund{\kltop}{\secmap{}}{} S(a) : \secmap{2}(a)}$.
\end{defi}

\subsubsection{Soundness}

Similarly to what we did in Subsection~\ref{subsec-confinement-static}, in order to establish the soundness of the Enforcement Mechanism~II we need a Subject Reduction result, stating that types that are given to expressions, along with compatibility of memories, are preserved by computation.

\begin{prop}[Subject Reduction]\label{prop-subjectreduction-confinementII}
Consider a thread $M^{m}$ such that ${\typenv \byund{A}{\secmap{}}{} M : \tau}$ and suppose that $W \vdash \iconft{T}{\{M^{m}\}} S$ $\xarr{\var {F}}{d}$ $\iconft{T'} {\{M'^{m}\} \cup P} {S'}$, for a memory $S$ that is $(W,{\secmap{}},{\typenv})$-compatible. Then, $\typenv\byund{A \klmeet W(T(m))}{\secmap{}}{} M':\tau$, and $S'$ is also $(W,{\secmap{}},{\typenv})$-compatible.
Furthermore, if~$P = \{N^n\}$, for some expression $N$ and thread name $n$, then $\typenv\byund{W(T'(n))}{\secmap{}}{} N:\unit$.
\end{prop}
\begin{proof}
We follow the usual steps \cite{WF94}, where the main proof is a case analysis on the transition $W \vdash \iconft{T}{\{M^{m}\}} S$ $\xarr{\var {F'}}{d}$ $\iconft{T'} {\{M'^{m}\} \cup P} {S'}$.
\end{proof}

Enforcement Mechanism~II guarantees security of networks with respect to confinement, as is formalized by the following result.
\begin{thm}[Soundness of Enforcement Mechanism~II]  \label{prop-soundness-confinementII}        
Consider an allowed-policy mapping $W$, reference labeling $\secmap{}$, typing environment $\typenv$, and a thread configuration $\confd P T$ such that for all $M^{m} \in P$ there exists $\tau$ such that ${\typenv \byund{W(T(m))}{\secmap{}}{} M : \tau}$.  Then ${\confd P T}$ is $(W,\secmap{},\typenv)$-confined.    
\end{thm}
\begin{proof}
We show that the following is a set of $(W,\secmap{},\typenv)$-confined thread configurations:
\myexample{
C = \{ \confd{P}{T} ~|~ \forall M^m \in P, \exists \tau ~.~ {\typenv \byund{W(T(m))}{\secmap{}}{} M : \tau}\} %
}
is a set of $(W,\secmap{},\typenv)$-confined thread configurations. 
By induction on the inference of ${\typenv \byund{W(T(m))}{\secmap{}}{} M : \tau}$. %
We use Subject Reduction (Proposition~\ref{prop-subjectreduction-confinementII}) and a case analysis on the last rule of the corresponding typing proof. %
\end{proof}

\subsubsection{Safety, precision and efficiency}
The proposed mechanism does not offer a safety result, guaranteeing that program execution never gets stuck. %
Indeed, the side condition of the thread creation rule introduces the possibility for the execution of a thread to block, since no alternative is given.  
This can happen in the example in Equation~(\ref{ex-tricky}), %
if the flow policy $F$ is not permitted by the allowed policy of the domain of the branch that is actually executed, then the migration will not occur, and execution will not proceed. %
We can, however, prove that the only way for a program to get stuck is if it triggers an unsafe migration:
\begin{prop}[Safety (weakened)] \label{prop-migrcontrol-safety}
Consider a closed thread $M^{m}$ such that ${\emptyset \byund{A}{\secmap{}}{} M : \tau}$.  Then, for any allowed-policy mapping $W$, memory $S$ that is $(W,{\secmap{}},{\emptyset})$-compatible and position-tracker $T$, either the program $M$ is a value, or:
\begin{itemize}
\item $W \vdash \iconft{T}{\{M^{m}\}} S$ $\xarr{\var {F'}}{T(m)}$ $\iconft{T'} {\{M'^{m}\} \cup P} {S'}$, for some $F'$, $M'$, $P$, $S'$ and $T'$, or
\item $M=\cC{{E}}{\threadnat {l} {N} {d}}$, for some ${E}$, $l$ and $d$ such that ${\emptyset \byund{\klbot}{\secmap{}}{} {N} : {\unit}}$ but ${\emptyset \not\byund{W({d})}{\secmap{}}{} {N} : {\unit}}$.
\end{itemize}
\end{prop}
\begin{proof} By induction on the derivation of ${\typenv \byund{A}{\secmap{}}{} M : \tau}$.  \end{proof} 

\NEW{}

In order to have safety, we could design a thread creation instruction that predicts an alternative branch for execution in case the side condition fails. We choose not to pursue this option, thus minimizing changes to the basic language under consideration.  Nevertheless, the condition that was introduced in the semantics of the thread creation sufficed to originate new forms of information flow leaks, as the event of blockage of a computation can reveal information about the control flow that led to it.  These information leaks can be treated in a similar manner as termination leaks, as in~\cite{Alm09}, for blockage can be seen as a form of non-termination.  Such leaks are in fact rejected by the type and effect system of Figure~\ref{fig-iflow-typesystem}, as is discussed further ahead at the end of Section~\ref{sec-confinement}.

It is worth considering that the example in Equation~(\ref{ex-tricky}) could be reformulated as
\numbexample{
{\threadnat l {\allowed F {\flow F {M_1}} {M_2}} d}
}
in effect using the allowed-condition for encoding such alternative behaviors.
Indeed, the programmer can increase the chances that the program will terminate successfully by guarding flow declarations with the new allowed-condition construct.  Our type system does take that effort into account, by subtracting the tested flow policy when typing the allowed-branch.  As a result, programs containing flow declarations that are too permissive according to a certain domain might still be authorized to execute in it, as long as they occur in the not-allowed-branch of our new construct which will not be chosen.

Returning to the example in Equation~(\ref{ex-staticimprecision}), %
thanks to the relaxed~\migtyp~rule, this program is now \emph{always} accepted statically by the type system. Depending on the result of the test, the migration might also be allowed to occur if a safe branch is chosen.  This means that Enforcement Mechanism~II accepts more secure programs, as well as more secure executions of (statically) insecure programs.

A drawback with this enforcement mechanism lies in the computation weight of the runtime type checks.  This is particularly acute for an expressive language such as the one we are considering.  Indeed, recognizing typability of ML expressions has exponential (worst case) complexity~\cite{Mai90}.  

\begin{figure}[t!]
\begin{equation*}
\begin{array}{c}

\text{[\niltypI]} ~
{\typenv \byinov{\secmap{}}{} \nil \hookrightarrow \nil : \kltop, {{\unit}}}  \quad  

\text{[\boolttypI]} ~
{\typenv \byinov{\secmap{}}{} \vrai \hookrightarrow \vrai : \kltop, {{\bool}}} \quad

\text{[\boolftypI]} ~
{\typenv \byinov{\secmap{}}{} \faux \hookrightarrow \faux : \kltop, {{\bool}}} \\[2.5mm]

\text{[\loctypI]} ~
{\typenv \byinov{\secmap{}}{} {a} \hookrightarrow a : \kltop, {\rfrt {\secmap{2}(a)} {}}} \qquad

\text{[\vartypI]} ~
{\typenv,x:\tau \byinov{\secmap{}}{} x \hookrightarrow x : \kltop, {\tau}} \\[2.5mm]

\text{[\abstypI]} ~
\fra{\typenv, x:{\tau} \byinov{\secmap{}}{} M \hookrightarrow {\hat M} : s,{\sigma}}
{\typenv \byinov{\secmap{}}{} {\lam {x} {M}} \hookrightarrow {\lam {x} {\hat M}} : \kltop, {\tau \xarr {}{s} \sigma}} \qquad

\text{[\rectypI]} ~
\fra{\typenv,x:\tau \byinov{\secmap{}}{} X \hookrightarrow {\hat X}: s,\tau}
{\typenv \byinov{\secmap{}}{} {\fix{x}X} \hookrightarrow {\fix{x}{\hat X}}: s,\tau} \\[5mm]

\text{[\apptypI]} ~
\fra{\typenv \byinov{\secmap{}}{} {M} \hookrightarrow {\hat{M}}: {\tef {s} {\tau \xarr{}{s'} \sigma}} ~~~ 
\typenv \byinov{\secmap{}}{} N \hookrightarrow {\hat{N}}: {{\tef{s''}{\tau''}}} ~~~
\tau \klpreceqtyp \tau''}
{\typenv \byinov{\secmap{}}{} {\app {M} {N}} \hookrightarrow {\app {\hat M} {\hat N}}: {\tef{s \klmeet s' \klmeet s''} \sigma} } \\[5mm]

\text{[\seqtypI]} ~
\fra{\typenv \byinov{\secmap{}}{} {M} \hookrightarrow {\hat M}: {\tef s {\tau}} ~~~  \typenv \byinov{\secmap{}}{} N \hookrightarrow {\hat N}: {{\tef{s'}{\sigma}}}}
{\typenv \byinov{\secmap{}}{} {\seq {M} {N}} \hookrightarrow {\seq {\hat M} {\hat N}}: {\tef{s \klmeet s'} \sigma} } \\[5mm]

\text{[\reftypI]} ~
\fra{\typenv \byinov{\secmap{}}{} {M} \hookrightarrow {\hat M}: {\tef s {\theta'}} ~~~
\theta \klpreceqtyp \theta'}
{\typenv \byinov{\secmap{}}{} {\rfr{\theta}{M}} \hookrightarrow {\rfr{\theta}{\hat M}}: s, {\rfrt \theta {}} } \qquad %

\text{[\dereftypI]} ~
\fra{\typenv \byinov{\secmap{}}{} {M} \hookrightarrow {\hat M}: {\tef s {\rfrt \theta {}}} }
{\typenv \byinov{\secmap{}}{} {\deref{M}} \hookrightarrow {\deref{\hat M}}: {\tef {s} {\theta}}} \\[5mm] %

\text{[\assigntypI]} ~ 
\fra{\typenv \byinov{\secmap{}}{} {M} \hookrightarrow {\hat{M}}: {\tef s {\rfrt {\theta} {}}} ~~~
\typenv \byinov{\secmap{}}{} N \hookrightarrow {\hat{N}}: {{\tef{s'}{\theta'}}} ~~~
\theta \klpreceqtyp \theta' }
{\typenv \byinov{\secmap{}}{} {\assign {M} {N}} \hookrightarrow {\assign {\hat M} {\hat N}}: {\tef{s \klmeet s'} \unit} }\\[5mm]

\text{[\condtypI]} ~
\fra{\typenv \byinov{\secmap{}}{} M \hookrightarrow {\hat{M}}: {s,{\bool}} ~~~ 
\begin{array}{c} 
\typenv \byinov{\secmap{}}{} N_t \hookrightarrow {{\hat N}_t}: {{s_t,{\tau_t}}} \\ 
\typenv \byinov{\secmap{}}{} N_f \hookrightarrow {{\hat N}_f}: {{s_f,{\tau_f}}}
\end{array} ~~~
\tau_t \kleqtyp \tau_f}
{\typenv \byinov{\secmap{}}{} {\cond{M}{N_t}{N_f}} \hookrightarrow {\cond{\hat M}{{\hat N}_t}{{\hat N}_f}}: {s \klmeet s_t \klmeet s_f,{\tau_t \klmeet \tau_f}}} \\[5mm]

\textbf{[\allowtypI]} ~
{\fra{ 
\begin{array}{c} 
\typenv \byinov{\secmap{}}{} {N_t} \hookrightarrow {{\hat N}_t}: {{s_t,{\tau_t}}} \\ 
\typenv \byinov{\secmap{}}{} N_f \hookrightarrow {{\hat N}_f}: {{s_f,{\tau_f}}}
\end{array} ~~~
\tau_t \kleqtyp \tau_f}
{\begin{array}{r}
\typenv \byinov{\secmap{}}{} {\allowed {F} {N_t} {N_f}} \hookrightarrow {\allowed {F} {{\hat N}_t} {{\hat N}_f}}: \\
{s_t \klpseudominus F \klmeet s_f},{\tau_t \klmeet \tau_f}\end{array}}} \\[5mm]

{\textbf{[\flowtypI]} ~
\fra{\typenv \byinov{\secmap{}}{} N \hookrightarrow {\hat N}: s,\tau}   %
{\typenv \byinov{\secmap{}}{} {\flow {F} N} \hookrightarrow {\flow {F} {\hat N}} : s \klmeet{} F, \tau}}  \\[7mm] %

{\textbf{[\migtypI]} ~
\fra{
\typenv \byinov{\secmap{}}{} {{M}} \hookrightarrow {\hat{M}}: {\tef s {\unit}}}
{\typenv \byinov{\secmap{}}{} {\threadnat l M {d}} \hookrightarrow {\threadanot l {\hat M} {d} s}: \kltop,{\unit}}}   
	\end{array}
\end{equation*}
\caption{Informative type and effect system for obtaining the Declassification Effect}	\label{fig-typesystem-declassif}
\figline
\end{figure}
\subsection{Static Informative Typing for Runtime Effect Checking}   \label{subsec-confinement-decleffect}

We have seen that bringing the type-based migration control of programs to runtime allows to increase the precision of the confinement analysis.  This is, however, at the cost of performance.  It is possible to separate the program analysis as to what are the declassification operations that are performed by migrating threads, from the safety problem of determining whether those declassification operations should be allowed at a given domain.  To achieve this, we now present an \emph{informative} type system~\cite{AS12} that statically calculates a summary of all the declassification operations that might be performed by a program, in the form of a \emph{declassification effect}.  Furthermore, this type system statically annotates the program with information about the declassifying behavior of migrating threads in order to support the runtime decision of whether they can be considered safe by the destination domain.  The aim is to bring the overhead of the runtime check to static time.

The typing judgments of the type system in Figure~\ref{fig-typesystem-declassif} have the form:
\myexample{
\typenv \byinov{\secmap{}}{} {M} \hookrightarrow \hat{M}: {\tef s {\tau}}}
Comparing with the typing judgments of Subsection~\ref{subsec-confinement-runtime}, while the flow policy allowed by the context parameter is omitted from the turnstile `$\vdash$', the security effect $s$ represents a flow policy which corresponds to the \emph{declassification effect}:  a lower bound to the flow policies that are declared in the typed expression.  The second expression $\hat M$ is the result of annotating $M$.  We thus consider an annotated version of the language of Subsection~\ref{subsec-language}, where the syntax of values and expressions (the sets are denoted $\Val'$ and $\Expr'$, respectively) differs only in the remote thread creation construct, that has an additional policy $F$ as parameter, written $\threadanot l M d F$. Also here, only the type labeling component $\secmap{2}$ of the reference labeling is used.

Types have the following syntax ($t$ is a type variable):
\myexample{
\tau, \sigma, \theta ~ \in ~  \Typ ~ ::= ~ t ~|~ \unit ~|~  \bool ~|~ \rfrt \theta {} ~|~ \tau \xarr{}{s} \sigma 
}
The syntax of types is the same as the one used in Subsections~\ref{subsec-confinement-static} and~\ref{subsec-confinement-runtime}, except that the flow policy that represents the latent declassification effect is written over the arrow of the function type $\tau \xarr{}{s} \sigma$.

It is possible to relax the type system by matching types that have the same structure, even if they differ in flow policies pertaining to them.  
We achieve this by overloading $\klpreceqtyp$ to relate types where certain latent effects in the first are at least as permissive as the corresponding ones in the second.
The more general relation $\kleqtyp$ matches types where certain latent effects differ:
Finally, we define an operation $\klmeet$ between two types $\tau$ and $\tau'$ such that $\tau \kleqtyp \tau'$:%
\myexample{
\begin{array}{c}
\tau\klpreceqtyp\tau' \text{ iff } \tau = \tau', \text{ or }\tau=\theta\xarr{}{F}\sigma \text{ and } \tau'=\theta\xarr{}{F'}\sigma' \text{ with } F\klpreceqtyp F' \text{ and } \sigma\klpreceq\sigma'\\

\tau \kleqtyp \tau' \text{ iff } \tau = \tau', \text{ or }\tau=\theta \xarr{}{F} \sigma \text{ and } \tau'=\theta \xarr{}{F'} \sigma' \text{ with } \sigma \kleqtyp \sigma'\\

\tau\klmeet\tau' = \tau, \text{ if }\tau = \tau', \text{ or } \theta\xarr{}{F\klmeet F'}\sigma \klmeet \sigma', \text{ if } \tau=\theta\xarr{}{F}\sigma \text{ and } \tau'=\theta\xarr{}{F'}\sigma'
\end{array}
}
The $\klpreceqtyp$ relation is used in rules \reftypI, \assigntypI~and~\apptypI, in practice enabling to associate to references and variables (by reference creation, assignment and application) expressions with types that contain stricter policies than required by the declared types.  The relation $\kleqtyp$  is used in rules \condtypI~and~\allowtypI~in order to accept that two branches of the same test construct can differ regarding some of their policies.  Then, the type of the test construct combines both using $\klmeet$, thus reflecting the flow policies in both branches.

The declassification effect is constructed by aggregating (using the meet operation) all relevant flow policies that are declared within the program.  The effect is updated in rule \flowtypI, each time a flow declaration is performed, and becomes more permissive as the declassification effects of sub-expressions are met in order to form that of the parent command.  However, when a part of the program is guarded by an allowed-condition, some of the information in the declassification effect can be discarded.  This happens in rule \allowtypI, where the declassification effect of the first branch is not used entirely:  the part that will be tested during execution by the allowed-condition is omitted.  In rule~\migtypI, the declassification effect of migrating threads is also not recorded in the effect of the parent program, as they will be executed (and tested) elsewhere.  That information is however used to annotate the migration instruction. 

As an example, the thread creation of the program in Equation~(\ref{exallowed2}), still assuming that ${\textit{plan\_A}}$ and ${\textit{plan\_B}}$ have no declassifications, would be annotated with the declassification effect ${F_{H \prec L}} \klpseudominus F$.  In particular, the effect would be $\kltop$ if $F \klpreceq {F_{H \prec L}}$.

One can show that the type system is deterministic, in the sense that it assigns to a non-annotated expression a single annotated version of it, a single declassification effect, and a single type.

\subsubsection{Modified operational semantics, revisited.} \label{subsubsec-deceffect-semantics}
\begin{figure}
\begin{equation*}
  \hspace{-3mm}
{\fra{ \boldsymbol{W(d) \klpreceq s} }  %
  {W \semvdash {\iconft{T}{\{{\EC{\threadanot{l}{N}{d}{s}}}^{m}\}}{S}}\xarr{\extrf{\cE{E}}}{T(m)}{\iconft{\update {n} d T}{\{{\EC{\nil}}^{m},N^{n}\}}{S}}}} %
\end{equation*}
\caption{Operational semantics with statically annotated programs for runtime migration control (omitted rules are as in Figure~\ref{fig-semantics})} \label{fig-confinementIII-semantics}
\figline
\end{figure}
By executing annotated programs, the type check that conditions the migration instruction can be replaced by a simple declassification effect inspection. The new migration rule, presented in Figure~\ref{fig-confinementIII-semantics}, is similar to the one in Subsection~\ref{subsec-confinement-runtime}, but now makes use of the declassification effect ($n$ fresh in $T$).
In the remaining rules of the operational semantics the annotations are ignored.  
Note that the values contained in memories are also assumed to use annotated syntax. %

We refer to the mechanism that consists of statically annotating all threads~in a network according to the type and effect system of Figure~\ref{fig-typesystem-declassif}, assuming that~each thread's declassification effect is allowed by its initial domain, while using the semantics of Figure~\ref{fig-confinementIII-semantics} modified, as \emph{Enforcement Mechanism~III}. 

\NEW
It is useful to define a family of functions $\annot{}{} : (\Ref \rightarrow \Val) \rightarrow (\Ref \rightarrow \Val')$ which, given a typing environment $\typenv$ and a reference labeling $\secmap{}$, extends the annotation process to stores:
\myexample{
  \begin{array}{rrl}
    \annot{S} = \hat S &\textit{ such that }& \dom{S}=\dom{\hat S}\\
    & \textit{ and } & \forall a\in\dom{S}, \typenv \byinov{\secmap{}}{} {S(a)} \hookrightarrow \hat{S}(a): {\tef \kltop {\secmap{2}(a)}}
  \end{array}
}   \label{def-deceffect-annot}
We can now define the compatibility predicate that applies to Enforcement Mechanism~III.
\begin{defi}[$(W,{\secmap{}},{\typenv})$-Compatibility] \label{def-confinementIII-compatibility}
A memory $\hat S$ is said to be $(W,{\secmap{}},{\typenv})$-\emph{compatible} if there exists a non-annotated memory $S$ such that $\annot{S}=\hat S$.
\end{defi}

\subsubsection{Soundness}  \label{subsub-decleffect-subjreduction}

The following proposition ensures that the annotation processing is preserved by the annotated semantics.  This is formulated by stating that after reduction, programs are still well annotated.  Since the type system integrates both the annotating process and the calculation of the declassification effect, the formulation of this result is slightly non-standard.
More precisely, the following result states that if a program is the result of an annotation process, with a certain declassification effect and type, then after one computation step it is still the result of annotating a program, and is given an at least as strict declassification effect and type.
\begin{prop}[Subject Reduction, or Preservation of Annotations] \label{prop-preservanot}
Consider a thread $M^{m}$ such that ${\typenv \byinov{\secmap{}}{} {M} \hookrightarrow {N}: s, {\tau}}$ and suppose that
$W \semvdash \iconft{T}{\{{N}^{m}\}} S$ $\xarr{\var {F}}{d}$ $\iconft{T'} {\{{N'}^{m}\} \cup P} {S'}$, for a memory $S$ that is $(W,{\secmap{}},{\typenv})$-compatible. Then there exist $M'$, $s'$, $\tau'$ such that $s \klmeet {W(T(m))} \klpreceq s'$, and $\tau \klpreceqtyp \tau'$, and $\typenv\byinov{\secmap{}}{} M' \hookrightarrow {N'}:s',\tau'$, and $S'$ is also $(W,{\secmap{}},{\typenv})$-compatible.
Furthermore, if $P = \{N''^n\}$ for some expression $N''$ and thread name $n$, then there exist $M''$, $s''$ such that $W(T'(n)) \klpreceq s''$ and $\typenv\byinov{\secmap{}}{} M'' \hookrightarrow N'':s'',\unit$.
\end{prop} 
\begin{proof}
We follow the usual steps \cite{WF94}, where the main proof is a case analysis on the transition $W \vdash \iconft{T}{\{N^{m}\}} S$ $\xarr{\var {F'}}{d}$ $\iconft{T'} {\{N'^{m}\} \cup P} {S'}$.
\end{proof}

We will now see that the declassification effect can be used for enforcing confinement.
\begin{thm}[Soundness of Enforcement Mechanism~III] \label{prop-soundness-declassif}
Consider an allowed-policy mapping $W$, reference labeling $\secmap{}$, typing environment $\typenv$, and a thread configuration $\confd P T$ such that for all $M^{m} \in P$ there exist $\hat M$, $s$ and $\tau$ such that ${\typenv \byinov{\secmap{}}{} M \hookrightarrow \hat{M}: s,\tau}$ and ${W(T(m))} \klpreceq s$.  Then ${\confd {\hat P} T}$, formed by annotating the threads in $\confd P T$, is $(W,\secmap{},\typenv)$-confined.
\end{thm}
\begin{proof} 
We show that the following is a set of $(W,\secmap{},\typenv)$-confined thread configurations:
\myexample{
C = \{ \confd{P}{T} ~|~ \forall {\hat M}^m \in P,~ \exists M,s,\tau ~.~ {\typenv \byinov{\secmap{}}{} M \hookrightarrow \hat M: s,\tau} \text{ and } W(T(m)) \klpreceq s\}   %
}
By induction on the inference of ${\typenv \byinov{\secmap{}}{} M \hookrightarrow \hat M: s,\tau}$. %
We use Subject Reduction (Proposition~\ref{prop-preservanot}) and a case analysis on the last rule of the corresponding typing proof. %
\end{proof}

\subsubsection{Safety, precision and efficiency}
The weak safety result is formulated similarly to Proposition~\ref{prop-migrcontrol-safety}:
\begin{prop}[Safety (weakened)] \label{prop-decleffect-safety}
Consider a closed thread $M^{m}$ for which we have that ${\emptyset \byinov{\secmap{}}{} M \hookrightarrow {\hat M}: s,\tau}$.  Then, for any allowed-policy mapping $W$, memory $S$ that is $(W,{\secmap{}},{\emptyset})$-compatible and position-tracker $T$, either the program $\hat M$ is a value, or:
\begin{enumerate}
\item $W \vdash \iconft{T}{\{{\hat M}^{m}\}} S$ $\xarr{\var {F'}}{T(m)}$ $\iconft{T'} {\{{\hat M}'^{m}\} \cup P} {S'}$, for some $F'$, ${\hat M}'$, $P$, $S'$ and $T'$, or
\item $M=\cC{{E}}{\threadnat {l} {N} {d}}$, for some ${E}$, $l$, $d$ and $\hat s$ such that $\hat M=\cC{{E}}{\threadanot {l} {\hat N} {d} {\hat s}}$ and ${\emptyset \byinov{\secmap{}}{} N \hookrightarrow {\hat N}: {\hat s},\tau}$ but ${ W(d) \not\klpreceq \hat{s}}$.
\end{enumerate}
\end{prop}
\begin{proof} By induction on the derivation of ${\typenv \byinov{\secmap{}}{} M : s,\tau}$.  \end{proof}
The relaxed type system of Subsection~\ref{subsec-confinement-runtime} for checking confinement, and its informative counterpart of Figure~\ref{fig-typesystem-declassif}, are strongly related.  
The following result states that typability according to latter type system is at least as precise as the former.  
\begin{prop} \label{prop-checkinform-comparison} \label{prop-equivalencetypsys}
If for an allowed policy $A$ and type $\tau$ we have that ${\typenv \byund{A}{\secmap{}}{} M : \tau}$, then there exist $N$, $s$ and $\tau'$ such that ${\typenv \byinov{\secmap{}}{} M \hookrightarrow N: s, \tau'}$ with $A \klpreceq s$ and $\tau \klpreceq \tau'$.
\end{prop} 
\begin{proof}
By induction on the inference of $\typenv\byund{A}{\secmap{}}{} M : {\tau}$, and by case analysis on the last rule used in this typing proof.
\end{proof}

The converse direction is not true, i.e. Enforcement Mechanism~III accepts strictly more programs than Enforcement Mechanism~II.  This can be seen by considering the secure program
\numbexample{ \cond {\deref a} {\deref {\rfr {\theta_1} {M_1}}} {\deref {\rfr {\theta_2} {M_2}}}  \label{ex-strictlymorepermissive}
}
\noindent where $\theta_1 = {\tau \xarr {}{F_1} \sigma}$ and $\theta_2 = {\tau \xarr {}{F_2} \sigma}$.
This program is not accepted by the type system of Section~\ref{subsec-confinement-runtime} because it cannot give the same type to both branches of the conditional (the type of the dereference of a reference of type $\theta$ is precisely $\theta$).  However, since the two types satisfy $\theta_1 \kleqtyp \theta_2$, the informative type system can accept it and give it the type $\theta_1 \klmeet \theta_2$.

A more fundamental difference between the two enforcement mechanisms lays in the timing of the computation overhead that is required by each mechanism.  While mechanism~II requires heavy runtime type checks to occur each time a thread migrates, in III the typability analysis is anticipated to static time, leaving only a comparison between two flow policies to be performed at migration time.  The complexity of this comparison depends on the concrete representation of flow policies.  In the worst case, that of flow policies as general downward closure operators (see Section~\ref{sec-setting}), it is linear on the number of security levels that are considered.  When flow policies are flow relations, then it consists on a subset relation check, which is polynomial on the size of the flow policies.

\subsubsection{Preservation of the semantics}  \label{subsubsec-deceffect-preservation}

We now prove that the program transformation that is encoded in the type system of Figure~\ref{fig-typesystem-declassif} preserves the semantics of the input expressions.  To this end, we define a simulation between pools of threads written in the original language of Subsection~\ref{subsec-language}, and pools of threads written in the language with annotations.  The behavior of pools of threads of the former should be able to simulate step-by-step those of the latter, when operating on memories that are the same, up to the annotations that distinguish the two languages (captured by means of the $\annot{}{}$ function. %

The following simulation on pools of threads relates programs in the annotated language whose behavior can be entirely performed by the corresponding program in the original (non-restricted) language.  
\begin{defi}[$\realbiseq{W,\secmap{},{\typenv}}$] \label{def-realbisim}  %
A $(W,\secmap{},{\typenv})$-simulation is a binary relation $\fS$ on pools of threads, drawn from the original language and the annotated language, respectively, that satisfies, for all thread configurations $T$ and for all memories $S,\hat S$
\myexample{
  {P_1} ~\fS~ {P_2} ~\textit{and}~ W \vdash \iconft{T}{P_2} {\hat S} \xarr{F}{d} \iconft{T'}{P_2'}{\hat S'} ~\textit{and}~ \annot{S} = \hat S}
implies that there exist ${P_1'},{S}'$ such that:
\myexample{
  W \vdash \iconft{T}{P_1}{S} \xarr{F}{d} \iconft{T'}{P_1'}{{S'}} ~\textit{and}~ \annot{S'}={\hat{S}'}~\textit{and}~ {P_1'} ~\fS~ {P_2'}
}
The largest $(W,\secmap{},{\typenv})$-simulation is denoted by $\realbiseq{W,\secmap{},{\typenv}}$.  
\end{defi}
For any $W,\secmap{},\typenv$, the set of pairs of thread configurations where threads are values is a $(W,\secmap{},{\typenv})$-simulation.  Furthermore, the union of a family of $(W,\secmap{},{\typenv})$-simulations is a $(W,\secmap{},{\typenv})$-simulation. Consequently, $\realbiseq{W,\secmap{},{\typenv}}$ exists. %

We can now check that the annotation process produces expressions that can be simulated in the above sense.  In other words this means that Enforcement Mechanism~III only enables behavior of programs that is already present in the original language.
\begin{prop} \label{prop-preservsem}
If %
  ${\typenv \byinov{\secmap{}}{} {M} \hookrightarrow {N}: s, \tau}$, then for all allowed-policy mappings $W$ and thread names $m \in \Nam$ we have that $\{M^m\} \realbiseq{W,\secmap{},{\typenv}} \{N^m\}$.
\end{prop}
\begin{proof}
We prove that the set
\myexample{
B = \{\confd{\{M^m\}}{\{N^m\}} ~|~ m \in \Nam \textit{ and } \exists s,\tau ~.~ {\typenv \byinov{\secmap{}}{} {M} \hookrightarrow {N}: s, \tau}\}
}
is a $(W,\secmap{},{\typenv})$-simulation according to Definition~\ref{def-realbisim}.  By case analysis on the proof of $W \vdash \iconft {T} {\{{M}^m\}} {S}$ $\xarr{F}{d}$ $\iconft {T'}{P_2'}{S'}$ and of $W \vdash \iconft {T} {\{{N}^m\}} {\hat S}$ $\xarr{F}{d}$ $\iconft {T'}{P_2'}{\hat S'}$.
\end{proof}
\noindent To see that the converse is not true, i.e. that the range of behaviors of migrating programs is restricted with respect to the original language, it is enough to consider the program
\numbexample{
\threadnat l {\flow {\klbot} {\assign {a} {1}}} d
}
which is transformed into $\threadanot l {\flow {\klbot} {\assign {a} {1}}} d {\klbot}$ by the informative type system.  This program blocks at the first execution step when $W(d) \neq \klbot$.

\section{Information Flow Compliance} \label{sec-distnonint}

In this section we clarify the meaning of secure information flow in a distributed security setting, by interpreting the allowed flow policy of a domain as an imposition on the information flows that a mobile program is allowed to set up while computing at a given domain.  Using this notion, we propose a natural generalization of Non-interference that establishes the absence of illegal information flows in a distributed security setting.  This new property, named Distributed Non-interference, states that information flows should respect the allowed flow policy of the domains where they originate.  Finally, Distributed Non-interference is used to validate the semantic coherence between Distributed Non-disclosure and Flow Policy Confinement, the two properties that are studied in the previous two sections.  As a consequence, programs that are accepted by enforcement mechanisms for both properties, are known to respect Distributed Non-interference.

\paragraph*{Distributed Non-interference.}  

Informally, a program is said to set up an information flow from security level $l_1$ to security level $l_2$ if, during computation, information that has confidentiality level $l_1$ can interfere with what is observable at security level $l_2$.  From the point when the information is read by the program, to the point where it's interference is reflected in the observable state, the program might change location.  Which allowed flow policy should it comply to?  We propose that it should comply to the allowed flow policy $A$ of the domain where the information was read.  More precisely, using the notation from Subsection~\ref{subsec-secsetting}, it must hold that $l_2 \Fpreceq{A} l_1$.

The location of a thread thus determines which allowed flow policy it should obey at that point, and is used to place a restriction on the information flows that are initiated at that step. %
In order to formalize this idea by means of a bisimulation, the portion of the state that is observable, or read, at a domain $d$ is adjusted according to the current allowed flow policy that is given by $W(d)$, in the first low-equality of the definition.  The bisimulation condition then determines how that information can affect future computation steps, as the resulting thread configurations must still be bisimilar.  In particular, the computation cannot produce pairs of programs that produce distinguishable states.
\begin{defi}[$\sbiseqt{\secmap{},\Upsilon}{\typenv}{\ell}$] \label{def-bisimdef-DNI}  %
Given an security level $\ell$,
a $(W,{\secmap{}},{\Upsilon},{\typenv},{\ell})$-bisimulation is a symmetric relation $\fR$ on thread configurations that satisfies, for all $P_1,T_1,P_2,T_2$, and $(W,{\secmap{}},{\typenv})$-compatible stores~$S_1, S_2$:
\myexample{
\confd{P_1}{T_1}~\fR~\confd{P_2}{T_2} ~\text{and}~ W \vdash \iconft{T_1}{P_1}{S_1} \xarr{F}{\boldsymbol{d}} \iconft{T_1'}{P_1'}{S_1'}
~\text{and}~ \confd{T_1}{S_1} \memeqF{\secmap{1},\Upsilon}{\boldsymbol{W(d)}}{l} \confd{T_2}{S_2}
}
with $(\dom{{S_1}'}\setminus\dom{S_1})\cap\dom{S_2}=\emptyset$ and $\dom{{T_1}'}\setminus\dom{T_1})\cap\dom{T_2}=\emptyset$ implies that there exist ${P_2'}, {T_2'}, {S_2'}$~such that:
\myexample{
W \vdash \iconft{T_2}{P_2}{S_2}~\rarr~\iconft{T_2'}{P_2'}{S_2'} ~\text{and}~ 
\confd{T_1'}{S_1'}~\memeqF{\secmap{1},\Upsilon}{\kltop}{l}~\confd{T_2'}{S_2'} ~\text{and}~ \confd{P_1'}{T_1'}~\fR~\confd{P_2'}{T_2'}
}
Furthermore, $S_1',S_2'$ are still $(W,\secmap{},\typenv)$-compatible. %
The largest $(W,{\secmap{}},\Upsilon,{\typenv},\ell)$-bisimulation %
is denoted~$\sbiseqt{\secmap{},\Upsilon}{\typenv}{\ell}$. 
\end{defi}
For any ${\secmap{}}$, $\Upsilon$, $\ell$, the set of pairs of thread configurations where threads are values is a $(W,{\secmap{}},\Upsilon,{\typenv},{\ell})$-bisimulation.  Furthermore, the union of a family of $(W,{\secmap{}},\Upsilon,{\typenv},\ell)$-bisimulations is a $(W,{\secmap{}},\Upsilon,{\typenv},\ell)$-bisimulation. Consequently, $\sbiseqt{\secmap{},\Upsilon}{\typenv}{\ell}$ exists. %
\renewcommand{\sbiseqt}[3]{{\sim}^{W,#2}_{#3}}
In order to lighten the notation, the parameters $\secmap{}$ and $\Upsilon$ are omitted in the rest of the paper, written simply `$\sbiseqt{\secmap{},\Upsilon}{\typenv}{\ell}$'.

Information flows that are initiated at a given step are captured by considering all computations that initiate on every possible pair of stores that coincide in the observable (reading) level.  This enables compositionality of the property, accounting for the possible changes to the store that are induced by external threads.  Position trackers of configurations are fixed across steps within the bisimulation game, reflecting the assumption that changes in the position of a thread can only be induced by the thread itself (\emph{subjective} migration).  
Note that the above relation is not reflexive. In fact, only secure programs are related to themselves:
\begin{defi}[Distributed Non-interference]\label{def-propertyDNI}
A pool of threads~$P$ sa\-tis\-fies Distributed Non-interference with respect to an allowed-policy mapping~$W$, reference and thread labelings $\secmap{},\Upsilon$ and typing environment $\typenv$, if it satisfies $\confd{P}{T_1}~\sbiseqt{\secmap{},\Upsilon}{\typenv}{\ell}~\confd{P}{T_2}$ for all security levels~$\ell$ and position trackers $T_1,T_2$ such that $\dom{P}=\dom{T_1}=\dom{T_2}$ and $T_1 \memeqF{\secmap{1},\Upsilon}{\kltop}{l} T_2$.  We then write $P\in\SecDNI(W,\secmap{},\Upsilon,\typenv)$.
\end{defi}
Distributed Non-interference is compositional by set union of pools of threads, assuming disjoint naming of threads and a subjective migration primitive.  %

\paragraph*{Non-interference.}

As expected, local Non-interference follows from Definition~\ref{def-propertyDNI} when networks are collapsed into a single domain $d^*$, i.e. when $\Dom = \{d^*\}$.  The resulting property is parameterized by $W(d^*)$, %
and coincides with the view of Non-interference as the \emph{absence of information leaks} (in the sense defined in Section~\ref{sec-setting}) when $W(d^*) = \kltop$.  However, here we adopt the view of Non-interference as the \emph{absence of illegal flows}, which is relative to the particular allowed flow policy.  This point is further discussed in Section~\ref{sec-related}.

Programs that violate Non-interference, such as the direct leak ${\assign {b} {\deref {a}}}$ when running at $d$ such that $\secmap{1}(a)~\not\Fpreceq{W(d)}~\secmap{1}(b)$
are also insecure with respect to Distributed Non-interference.  The same holds for indirect leaks via control and termination.

In the following program, information flows from reference $a$ to $b$ via observation of the end domain:
\numbexample{\label{ex-distr-allowedmigr1}
\begin{array}{rl}
(\lkw{if} {\deref {a}}&\lkw{then} {\threadnat l {\allowed F {\assign {b} {0}}{\nil}} {d_1}}\\
&\lkw{else} {\threadnat l {\allowed F {\assign {b} {1}}{\nil}} {d_2}})
\end{array}
}
In fact, since the position of threads in the network is part of the observable state, a \emph{migration leak}~\cite{AC11} from level $\secmap{1}(a)$ to level $l$ occurs as soon as the new thread is created. %
The information is read before migration, so the leak is only secure if it is allowed by the policy of the thread's initial domain.

Let us now consider the simpler program
\numbexample{\label{ex-distr-allowedmigr2}
\begin{array}{l}
(\lkw{if} {\deref {a}} \mkw{then}{ \threadnat l {\assign {b} {0}} {d_1}} \mkw{else} {\threadnat l {\assign {b} {1}} {d_1}})\end{array}
}
to be executed at $d$ such that $\secmap{1}(a)\Fpreceq{W(d)}l$ (where the migration leak is allowed). Since the two threads produce different changes at the level of $b$, then the program is again secure only if the first domain $d$ allows it, i.e. $\secmap{1}(a)\Fpreceq{W(d)}\secmap{1}(b)$.
The policy of the new domain $d_1$ only rules over what happens from the point where the thread enters it.  Since the behavior of the code that actually migrates does not depend on the location of the corresponding threads, no leak is taking place at this point of the program. 

Notice that the declassification operations are transparent to the property.  In particular, the program in Equation~(\ref{ex-notDNI})
violates Distributed Non-interference, regardless of $F$, if~$\secmap{1}(a) \not\Fpreceq{W(d)} \secmap{1}(b)$.

\subsection{Semantic Coherence} \label{subsec-semcoherence}

We now show that Distributed Non-interference is coherent with %
the properties of Non-interference, Distributed Non-disclosure and Flow Policy Confinement.

While Distributed Non-disclosure establishes a match between the leaks that are performed by a program and the declassifications that are declared in its code, Flow Policy Confinement requires that the declassifications %
comply to the allowed flow policy of the domain where they occur.  It is thus expected that a notion of Distributed Non-interference, which ensures that leaks respect the relevant allowed flow policies, should follow from the combination of the other two.  %

\begin{thm} \label{theo-local-implic}
$\begin{array}{l}\SecDND(W,\secmap{},\Upsilon,\typenv)~\cap~\SecFPC(W,\secmap{},\typenv)~\subseteq~\SecDNI(W,\secmap{},\Upsilon,\typenv).\end{array}$
\end{thm}
\begin{proof}
We use Definitions~\ref{def-propertyDND}, \ref{def-propertyFPC} and~\ref{def-propertyDNI} to show that the set:
\myexample{C = \biseqt{W,\secmap{},\Upsilon}{\typenv}{\ell} ~\cap~ \semconf^{W,\secmap{},\typenv} \times \semconf^{W,\secmap{},\typenv} \subseteq \sbiseqt{\secmap{},\Upsilon}{\typenv}{\ell}}
Consider $(\confd{P_1}{T_1},\confd{P_2}{T_2}) \in C$, and suppose that $W \vdash \iconft{T_1}{P_1} {S_1} \xarr{F}{d} \iconft{T_1'}{P_1'}{S_1'}$ and that $T_1 \memeqF{\secmap{1},\Upsilon}{W(d)}{l} T_2$.  Now consider $S_2$ such that ${S_1} \memeqF{\secmap{1},\Upsilon}{W(d)}{l} {S_2}$ and $\dom{{S_1}'} - \dom{S_1} \cap \dom{S_2} = \emptyset$.  Since $\confd{P_1}{T_1} \in \semconf^{W,\secmap{},\typenv}$, we know that $W(d) \klpreceq F$ and $\confd{P_1'}{T_1'} \in \semconf^{W,\secmap{},\typenv}$.  Because of the former, we therefore have that $\confd{P_1}{T_1} \memeqF{\secmap{1},\Upsilon}{F}{l} \confd{P_2}{T_2}$.   Since $\confd{P_1}{T_1}~\biseqt{W,\secmap{},\Upsilon}{\typenv}{\ell}~\confd{P_2}{T_2}$, this implies that $\exists {P_2'}, {T_2'}, {S_2'}$ such that $W \vdash \iconft{T_2}{P_2}{S_2} \rarr \iconft{T_2'}{P_2'}{S_2'}$ and $\confd{T_1'}{S_1'} \memeqF{\secmap{1},\Upsilon}{\kltop}{l} \confd{T_2'}{S_2'}$ and $\confd{P_1'}{T_1'}~\biseqt{W,\secmap{},\Upsilon}{\typenv}{\ell}~\confd{P_2'}{T_2'}$.  Since $\confd{P_2}{T_2} \in \semconf^{W,\secmap{},\typenv}$, whether $\confd{P_2'}{T_2'}$ results from zero or one step, we have that $\confd{P_2'}{T_2'} \in \semconf^{W,\secmap{},\typenv}$.  We can then conclude that $(\confd{P_1'}{T_1'},\confd{P_2'}{T_2'})~\in~C$.
\end{proof}
To see that the sets are not equal, consider again the program in Equation~(\ref{ex-notDNI}).  While it is secure regarding Distributed Non-interference if $H \Fpreceq{W(d)} L$,
if $F=\kltop$ it violates Distributed Non-disclosure, but respects Flow Policy Confinement; and
if $F=\klbot$ it respects Distributed Non-disclosure, but violates Flow Policy Confinement if $W(d)\not\klpreceq~F$.

\subsection{Enforcement}

While it is possible to design direct enforcement mechanisms for Distributed Non-interference (such as the proof of concept presented in~\cite{AC14}, the point that is relevant to this work is that enforcement of Distributed Non-interference can be obtained by simultaneously enforcing Distributed Non-Disclosure and Flow Policy Confinement, as follows from Theorem~\ref{theo-local-implic}.

It is worth highlighting that, in order to apply the above consideration, the semantics of programs that are in effect must be the same when enforcing the two properties.  This means in particular that if a hybrid mechanism such as Enforcement Mechanisms~II and III of the previous section is to be used, Distributed Non-disclosure must be ensured for the modified semantics, in order to exclude the emergence of side-channels.  It turns out that the type system of Section~\ref{sec-iflow} is still sound with respect to the modified semantics of Figures~\ref{fig-confinementII-semantics} and~\ref{fig-confinementIII-semantics}.  Indeed, when considering the alternative semantics,
the results Subject Reduction, Splitting Computations %
and Location Sensitive Typable High Threads ({prop-iflow-split} and (Proposition~\ref{prop-iflow-subjectreduction} and Lemmas~\ref{prop-iflow-split} and~\ref{prop-iflow-potentially}, respectively), hold unchanged.
Also the results outlined in the proof of Soundness (Theorem~\ref{prop-iflow-soundness}) still hold.  Intuitively, regardless of the state, the same remote thread creation always either succeeds or suspends, so the outcome is the same.

The programs of Equation~(\ref{ex-distr-allowedmigr1}) and~(\ref{ex-distr-allowedmigr2}) would be rejected by the type system of Figure~\ref{fig-iflow-typesystem} (Section~\ref{sec-iflow}) if $\secmap{1}(a)~\not\Fpreceq{\kltop}~\secmap{1}(b)$, because the illegal flows would not be matched by an appropriate flow declaration.  If written within a flow declaration that allows information to flow from $\secmap{1}(a)$ to $\secmap{1}(b)$, and when starting at a domain $d$, then the program would be rejected by the type systems of Figures~\ref{fig-confinementI-typesystem} and~\ref{fig-confinementII-typesystem} (Section~\ref{sec-confinement}) if $F \not\klpreceq W(d)$ because the program would be declaring intent of performing a declassification that is not allowed by the domain where it would take place.  Returning to the program of Equation~(\ref{ex-notDNI}), it would be rejected by the type system of Figure~\ref{fig-iflow-typesystem} if $\secmap{1}(a)~\not\Fpreceq{F}~\secmap{1}(b)$, and by that of Figure~\ref{fig-confinementI-typesystem} if $F \not\klpreceq W(d)$.  When running under Enforcement Mechanisms~II and~III, it would suspend at the point of remote thread creation if $F \not\klpreceq W(d)$.

\section{Related work and Discussion}  \label{sec-related}

\paragraph{\emph{Distributed security setting.}}

This paper is founded on a distributed computation model in which domains are units of security policies~\cite{BCGL02}, in the sense that behaviors that take place at the same domain are subjected to the same security policy: the allowed flow policy.  Other dimensions that have been considered for the semantics of domains include the unification of communication (the behavior of communication between entities is determined by their location)~\cite{CG98, HR02a, SWP99}, of mobility (entire domains move as a block)~\cite{FGLMR96, CG98, VC99}, of failures (failure in a domain affect the entities that are located there in a similar way)~\cite{Ama97,FGLMR96} or of memory access (completion of memory accesses depends on location)~\cite{SY97}.
Formal distributed security settings have been studied over the generic notion of programmable domain~\cite{SS03}, which offers expressivity to equip each domain with means to encode and enforce its own policies.  Besides security, these distributed policies may regard other goals such as failure resilience, resource management and correctness of mobile code.  Particular attention has been given to various forms of programmable membranes~\cite{Bou05a}, where computing power and policy specification is explicitly associated to the boundaries of computation domains, akin to firewalls (discussed further ahead under the topic \emph{Controlling code mobility}).  Policies that are associated to domains can more specifically be expressed as mathematical objects of different degree of expressiveness, such as types or finite automata~\cite{MV05,HMR03,GHS05,HRY05}.

Domains' security assurances can be differentiated as security levels, in the context of controlling information flows.  Myers and Liskov's Decentralized Label Model (DLM)~\cite{ML98,ML00} provides a way of expressing decentralized labels (referred to as policies, in a different use of the word) that extend the concrete setting of security levels as sets of principals, by distributing them according to ownership.
Zdan\-ce\-wic et. al~\cite{ZZNM02} propose in Jif/Split a technique for automatically partitioning %
programs by placing code and data onto hosts in accordance with DLM labels in the source code. Jif/Split assigns to each host arbitrary trust levels that are granted by the same principals of the system.  It 
ensures that if a host is subverted, the only data whose confidentiality or integrity is threatened during execution of a part of the program, is data owned by principals that trust that host.  %
Chong et. al~\cite{CLMQVZZ07} present Swift as a specialization of this idea for Web applications.
Zheng and Myers~\cite{ZM06} address the issue of how availability of domains (hosts) might affect information flows in a distributed computation.  Here the label on a host represents the strength of the policies that the host is trusted to enforce. %
Fournet et. al~\cite{FLR09,FP11} present a compiler that produces distributed code where communications are implemented using cryptographic mechanisms, and ensures that all confidentiality and integrity properties are preserved, despite the presence of active adversaries. %
While associating security levels to domains is a step towards considering heterogeneous security status, in the above works these levels are understood with respect to a single flow policy.

The Fabric programming system~\cite{LAGM17}, incorporates information flow methods into a platform for distributed computation over shared persistent data. %
Host nodes are represented by principals, who can declare trust in other principals, from which a global trust lattice is derived.
The goal is to ensure that %
the security of a principal does not depend on any part of the system that it considers to be an adversary.
Nothing is assumed about the behavior of nodes, which can be malicious.  This is compatible with a distributed security setting in which the policies of each node are unknown.  However, enforcement of information flow control for mobile code is done with respect to a global policy, bounded by the provider label of that code.
Most recently, FlowFence~\cite{FPRSCP16} addresses information flow security in the Internet of Things (IoT).  Based on an evaluation of representative examples of existing IoT frameworks, it advocates and supports differentiation of policies defined by both producers and consumers of apps, specifying the permissible flows that govern their data, against which the declared app's flow policies are checked.  Though not formally defined, policies are described as sets of flow rules between taint labels (data tags) and sinks (sensitive data outlets), stating which data can flow to which sinks.  The paper presents the Opacified Computation model, for enforcing source to sink information flow control as declared by consumers' policies.
To our knowledge, the work presented in this paper (which includes~\cite{Alm09,AC13,AC14}) is the only one that studies formally information flow security in the presence of distributed allowed flow policies.

\paragraph{\emph{Information flow and locality.}} %

Information flow security for location-aware programs has been studied in the absence of a distributed security setting.  Mantel and Sabelfeld \cite{MS02} %
provide a type system for preserving confidentiality for different kinds of channels established over a publicly observable medium in a distributed setting. Interaction between domains is restricted to the exchange of values (no code mobility), and the security setting is not distributed.
Castagna, Bugliesi and Craffa 
study Non-interference for a purely functional distributed and mobile calculus \cite{CBC02}. %
In these works no declassification mechanisms are contemplated. %
 Our previous work that is closest to this one~\cite{AC11} studies insecure information flows that are introduced by mobility in the context of a stateful distributed language  that includes flow declarations.  In the computation model that is considered, domains are memory access units. Threads own references that move along with them during migration, which gives rise to \emph{migration leaks} that result from memory synchronization issues:  non-local read and write access fail, and thus may reveal information about the location of a thread. %
Considering a shared global state in this paper allows us to make evident the distributed nature of the security setting, from which distributed behavior of programs emerges, while factoring out the technicalities of dealing with distributed and/or mobile references.  In this context, migration leaks appear as a result of inspection of the allowed flow policy, since information about current policy can carry information about the location itself.

The Fabric programming language~\cite{LAGM17} supports mobile code, allowing to dynamically combine code from multiple sources.  Programs use remote calls to explicitly transfer execution to a different node, %
by means of a mechanism of function shipping.
Stores are distributed into storage nodes, and code into worker nodes, which as said earlier are principals.  Access to a store by a worker implies an information flow channel, referred to as a \emph{read channel} that could have connections to the concept of migration leaks~\cite{AC11}.
The distributed computation model is not formally defined, so there is no formal account of the new forms of leaks in isolation, nor how they to the migration leaks discussed earlier~\cite{Alm09,AC11}.

When seen as a means for running code under different allowed flow policies, and in particular more permissive ones, analogies can be drawn between migration and the notion of declassification as a flow declaration.  Similarities extend to the fact that %
there is a lexical scope to the migrating code, as there is for the flow declaration.  However, differences between the two are more fundamentally expressed through the distinct roles of the flow policies that concern them:  while a declared flow policy %
represents an intention to perform a leak (declassification) that is rejected by the strictest baseline security lattice, an allowed flow policy that is held by a domain represents a limit that should not be crossed. %
This distinction is embodied by the formulation of two independent security properties that relate to each of the two flow policies: Distributed Non-disclosure, that places no restriction on the usage of declassification enabling constructs by the programmer, and Flow Policy Confinement, which does not speak of information flows but solely of whether declassification is allowed to be performed by a particular thread at a particular domain.
At a more technical level, migration has expression in the semantics of the language, with a potential impact on the observable state (eg. via the allowed flow policy), while the flow declaration is used solely with specification purposes.  Furthermore, in the considered language setting, domains cannot be nested into different shades of allowed flow policies, and code cannot move in and out of a flow declaration.

\paragraph{\emph{Enabling and controlling declassification.}}

Sabelfeld and Sands survey the literature regarding the subject of declassification~\cite{SS05}, from the perspective of that all declassification mechanisms provide some sort of declassification control.  Mechanisms that enable changes in the policy are included in this study.  Declassification mechanisms are then classified according to four main dimensions for control: %
\emph{what} information should be released~\cite{SM04,LZ05}, 
\emph{when} it should be allowed to happen~\cite{CM04}, %
\emph{who} should be authorized to use it~\cite{MSZ06}, %
 and \emph{where} in the program it can be stated~\cite{MS04,BS06,HTHZ05} %
 or a combination of thereof~\cite{BCR08}.
 The lexical scoping of the flow declaration could be seen to place it under the ``where'' category. However, it was proposed~\cite{AB09} as a means to \emph{enable} the expression of declassification (as opposed to attempt to control it) within an information flow control framework.  In that view, our type system for enforcing non-disclosure is best understood as a verification of accordance between code and policy specification in the form of flow declaration.
 More recently Broberg et al.~\cite{BDS15} proposed a unifying framework for reasoning about dynamic policies, according to which the term `declassification' is a particular case of dynamic policy change that is traditionally data oriented. Dynamic policy specification mechanisms are analysed through a structure of three control levels:  \emph{Level 0}, defining which set of flow relations can be used; \emph{Level 1}, describing how the active flow relation changes during execution; and \emph{Level 2}, restricting which policy changes may occur.  Using this classification, mechanisms that are designed to enable declassification, such as the flow policy declaration, operate at Level 1, while those for controlling (in the sense of restricting) its use, such as migration control,
 operate at Level~2.

It is helpful to see declassification control as dependent on the \emph{context} in which it may occur (formalized as \emph{meta information} that is argument to a \emph{meta policy}~\cite{BDS15}).  Flow Policy Confinement captures semantically the requirement that declassification between certain levels can only occur within computation domains that allow it.
Chong and Myers~\cite{CM04} associate declassification policies to variables as a way of restricting the sequence of levels through which they can be downgraded, provided some conditions are satisfied.  Control is then specified for each data item, and policies cannot be tested by the program.  Conditions are used in the definition of a generalized non-interference property to mark the steps where declassification occurs, in order to allow information release only if they are satisfied.
The Jif language~\cite{Mye99,ML00,CMVZ09}, founded on the aforementioned DLM~\cite{ML98,ML00}, offers a separate concept of \emph{authority} which regulates which declassifications are allowed.  The acts-for relation, describing delegation of trust between principals, is expressed similarly to a flow policy. %
A delegation operation enables programmers to extend the acts-for hierarchy in the scope of an expression.  Its mechanics is similar to our flow declaration, in that the program can change the valid flow policy. However, its purpose is to express control over the use of downgrading, which must be \emph{robust}~\cite{ZM01,MSZ06}, i.e. take place via declassify operations and be trusted by the principals whose policies are relaxed.  Also the use of delegation is controlled in order to ensure permission from the delegating principal~\cite{TZ07}.
Hicks et al.~\cite{HKMH06} introduce, in a featherweight version of Jif, the possibility of restricting changes to the acts-for policy to the use of declassifier functions that are specified in a global and static policy.
The Paralocks and Paragon policy specification languages~\cite{BS10,BDS13} operate over a lock state, by opening or closing locks that guard information flows.  As Paragon offers the possibility to place locks on policies, it is possible to define multiple layers of control, which can simultaneously express declassification enabling and restriction.

Declassification control criteria can be largely orthogonal to information flow legality, and can be captured semantically by means of security properties whose concerns are distinct from those of information flow properties.
Boudol and Kolund\u zija~\cite{BK07} use
standard access control primitives to control the access level of programs that perform declassifications in the setting of a local language, ensuring that a program can only declassify information that it has the right to read.  Besides presenting a standard information-flow security result (soundness of the proposed enforcement mechanism regarding Non-disclosure), they formalize a separate access control result that implies the property that a secure program does not attempt to read a reference for which it does not hold the appropriate reading clearance.  This approach was adopted to control declassification in a calculus for multiparty sessions with delegation, and enforced by means of a type system that guarantees both session safety and a form of access control~\cite{CCDR10}.
In order to achieve robustness in the presence of run-time principals, Tse and Zdancewic~\cite{TZ07} use a system of capabilities for making declassification and endorsement depend on a run-time authority.  Downgrading operations of information that is owned by a principal require that the principal has granted the appropriate capabilities.  While a type-safety result is proven, the enforced property is not formalized.
Arden et al.~\cite{ALM15,AM16} propose the Flow-Limited Authorization Model (FLAM)~\cite{ALM15}, a logic for reasoning about trust and information flow policies in which these two concepts are unified.  The model allows to control how delegations and revocations may affect authorization decisions.  It is formally shown to enforce a new~\emph{robust authorization} security condition, which rejects leaks that are present in earlier work.

\paragraph{\emph{Dynamic information flow policies.}} %

Security policies can in be used in the role of declarations that capture security relevant behavior, or of security requirements, both of which can change during program execution.  In the present work %
we use dynamic instances of both forms of flow polices: %
declared flow policies as a declassification enabling mechanism, that is set up by the program to delimit different lexical scopes within which different patterns of flows that break the basic global policy can be established;  %
distributed allowed flow policies, coupled with the notion of thread location, that govern what each thread is allowed to do, and can change dynamically (in relation to that thread to that thread) along with the migration of the thread. %
Returning to the terminology of Broberg et al.~\cite{BDS15}, these two roles would correspond to specification mechanisms of control Level~1 and~2, respectively.

Dynamic security polices can be seen as being under the responsibility of an authority.  In the present work, the units of authority are locations, and their policies take the form of allowed flow policies.
Jif's %
acts-for relation of the DLM~\cite{ML98,ML00,CMVZ09} (referred-to above), represents a policy that defines the current global run-time authority and is shaped by multiple principals.  The acts-for policy is dynamic, as it can be extended via delegation within a lexical scope~\cite{TZ07}.
Authority can be expressed over language objects in the form of dynamic security labels, while the underlying security lattice remains fixed~\cite{ZM04,TZ07}.
Hicks et al.~\cite{HTHZ05} consider a generalization of the global acts-for policy scenario that copes with asynchronous updates during run-time.  This policy is closer to the concept of allowed flow policy as external to the program's own policies, and is studied in a simpler non-distributed context, in the absence of declassification.
Swamy et al.~\cite{SHTZ06} use role-based allowed policies where membership and delegation of roles can change dynamically.
Programming languages can provide constructs for interacting with flow policies using operations for changing the policy, or %
for being aware of the current flow policy, in order to adapt its behavior accordingly. %
Boudol and Kolund\u zija's~$\kw{test~} l \kw{~then~} M \kw{~else~}N$ instruction\cite{BK07}, that tests whether the access level $l$ is granted by the context, is used to control the access level of programs that perform declassifications.
Broberg and Sands, use \emph{flow locks}~\cite{BS06} to selectively provide access to memory locations. They introduce flexible constructs for manipulating locks, namely open and close.
Extending the idea of the flow locks, the expressive policy languages Paralocks and Paragon~\cite{BS10,BDS13} introduces the ability to express policies modelling roles (in the style of role-based access control) and run-time principals and relations.
The language includes creation of new actors via $\kw{newactor}$ $\kw{forall}$, and $\kw{open}$ and $\kw{close}$ for changing the current policy. This language includes a $\kw{when}$ command is a conditional which executes one of two branches depending on the state of a particular lock.  %
Both dynamic and static tests to the authority's policy are supported in Jif, via the \verb|switch-label| and the \verb|actsFor|), and have been formally studied in related models that include dynamic labels~\cite{ZM07} and dynamic principals~\cite{TZ07}.

Enforcing dynamic (allowed) flow policies raises the challenge of that, even when they change synchronously with the program (i.e. are deterministically determined by it~\cite{DHS15}), it is not in general known, for each program point, what policy will apply at execution time.  Here we proposed solutions in the realm of migration control, by preventing a thread that can perform an action to migrate to a domain where that action is disallowed.  Inversely, in~\cite{HTHZ05}, where the allowed policy can change asynchronously, updates to the policy that would be inconsistent with the executing program (said to ``violate'' the program's policy), are considered illegal.
Swamy et. al~\cite{SHTZ06} use database-style transactions which, when a change in policy is inconsistent with some of the flows in the program, can roll it back to a consistent state.

\paragraph{\emph{Controlling code mobility.}}

A variety of formal distributed computation models have been designed with the purpose of studying mechanisms for controlling code mobility.
These range from type systems for statically controlling migration as an access control mechanism~\cite{MV05,HMR03,HRY05}, to runtime mechanisms that are based on the concept of programmable membranes~\cite{Bou05a}.  The later are a specialization of the aforementioned programmable domains for performing automatic checks to migrating code, with the aim of controlling the movements of entities across the boundaries of the domain.  This control can be performed by processes that interact with external and internal programs~\cite{LS00,SS03,Bou05a}, and can implement specialized type systems or other automatic verification mechanisms~\cite{GHS05}.

In the present work we abstract away from the particular machinery %
that implements the migration control checks, and express declaratively, via the language semantics, the condition that must be satisfied for the boundary transposition to be allowed.  
The declassification effect can be seen as a certificate that moves with code in order to aid its verification, similarly to the proof-carrying code model \cite{Nec97}.  However, as we have seen, checking the validity of the declassification effect is not simpler than checking the program against a concrete allowed policy (as presented in Subsection~\ref{subsec-confinement-runtime}). %
The concept of trust %
can be used to lift the checking requirements of code whose history of visited domains provides enough assurance~\cite{GHS05,MV05}.
These ideas could be applied to the present work, assisting the decision of trusting the declassification effect, otherwise leading to a full type check of the code.

Migration control has recently been applied in the context of building practical secure distributed systems that perform information flow and declassification control.
As mentioned earlier, in Fabric~\cite{LAGM17}, mobile code is subject to run-time verification.  Besides provider-bounded information flow checking, this includes validation of trust between the caller and callee nodes.  Since running the code depends on the success of these checks, they amount to a form of migration control.

\paragraph{\emph{Hybrid mechanisms and information flow policies.}}
Hybrid mechanisms for enforcing information flow policies have been subject of active interest for over a decade (see other detailed reviews~\cite{MC11,BBJ16}).

The idea of associating a runtime type analysis to the semantics of programs for controlling declassification was put forward as a proof-of-concept~\cite{Alm09}, and further developed in the hybrid mechanisms of the present and work.  Both are presented as abstract analysis predicates that condition a migration instruction.  Mechanisms that interact in more complex ways with the semantics of the program could be defined by means of a monitor, such as Askarov and Sabelfeld's for embedding an on-the-fly static analysis in a lock-step information flow monitor~\cite{AS09}.
Focusing on the expressivity of the analyses, the former work~\cite{Alm09} considers a similar language as the one studied here, with local thread creation and a basic $\kw{goto}$ migration instruction.  A notion of declassification was inferred at runtime, using a coarser informative type and effect system.
Future migrations were not taken into account when analyzing the declassifications that could occur at the target site, so the program 
\numbexample{
\threadnat l { \threadnat l {\flow {F} {M}} {d_2}} {d_1}}
would be rejected if $F$ was not allowed by~$W(d_1)$.
The three enforcement mechanisms of Section~\ref{sec-confinement} %
are more refined, as they only reject it if $F$ is not allowed by $W(d_2)$ by the control mechanism associated to $d_2$.
Enforcement Mechanism~II is the closest to the former work, but uses a more permissive ``checking'' type system.  
Enforcement Mechanism~III uses a type and effect system for statically calculating declassification effects that is substantially more precise than the former work, thanks to the matching relations and operations that it uses.  This enables to take advantage of the efficiency of flow policy comparisons.

Other previous work explores the usefulness of employing informative type and effect systems in building hybrid security analysis mechanisms~\cite{AS12}.  There, the term `declassification effect' represents a conservative approximation of the illegal information flows (as opposed to the declassification policies, investigated here) that are set up by the program, inferred by means of a dependency analysis, in the form of downward closure operators that represent information flow policies.  %
The advantages of determining these policies statically in order to efficiently verify runtime compliance to dynamic security policies analyses is illustrated, for an expressive core-ML language, in the context of asynchronous dynamic updates to a local allowed flow policy.
Delft et al.~\cite{DHS15} also advocate for the advantages of separating the dependency analysis and policy
enforcement in the context of a hybrid mechanism, in order to enforce compliance to dynamic policies.  Dependencies between resources that are established by a program in a simple while language with output channels are extracted by means of a more precise and generic flow-sensitive type system~\cite{HS06}, in the form of a mapping from each program variables to its (conservatively approximated) set of dependencies (variables).

In FlowFence~\cite{FPRSCP16} both application producers and consumers can declare their information flow policies.  These policies are compared using operations such as intersection and subtraction, that are reminiscent (the formal definitions are not presented) of the pseudo-subtraction operation and permissiveness relation defined in Subsection~\ref{subsec-secsetting}, which was mainly used in calculating the declassification effect in Figure~\ref{fig-typesystem-declassif}.  Our declassification effect of a program can be seen as the least permissive policy that the program complies to, and can therefore declare in its \emph{manifest}.

\section{Conclusions and Future work}\label{sec-concl}

In this work we considered a simple network model of a distributed security setting, and provided a thorough study of information flow issues that arise within it.  Adopting the standpoint of separating the enabling and controlling facets of declassification, we defined and studied two security properties that are related to declassification: Distributed non-disclosure, which regards the compliance of information leaks to declassification declarations in the code, and Flow Policy Confinement, which regards compliance of those declarations to the allowed flow policies of each site.  For each, we proposed mechanisms for enforcing them, ranging from purely static to hybrid approaches.
Finally, we tied up the both properties and respective mechanisms by showing that, in combination, they imply a generalization of Non-interference. We propose this property, which we refer to as Distributed Non-interference, as a baseline information flow property that is adequate for distributed security settings.  

We showed how the problem of migration control can be applied to that of declassification control.
By comparing three related enforcement mechanisms, we argued that the concept of declassification effect offers advantages in achieving both precision and efficiency.  We believe that similar mechanisms can be applied in other contexts.
This new security effect can be associated at runtime to a program, containing information about the declassifying environments (or of other behavior of interest) that can potentially be established by that program. %
It is flexible enough to allow programs containing operations that are forbidden at certain sites to be considered secure nevertheless, as long as these operations are guarded by an appropriate context testing construct.
The techniques that were studied are largely independent of the declassification mechanism that is used, which testifies to how a layer of control can be added to the most permissive declassification mechanisms.

We expect that our language framework can be easily adapted to studying other language constructs and settings, security properties and enforcement mechanisms.
We end with a few notes on possible lines for future work:
\begin{itemize}
\item %
In our model, domains are the sole authorities of security policy, which must be respected by migrating threads when executing at each domain.  The model can be refined to enable migrating programs to establish a security policy over how their information should be handled, possibly in a setting where references can move along with threads~\cite{AC11}.  The concern would be to prevent information that is carried by programs from becoming accessible to others with more permissive flow policies.

\item It is known that when considering a %
distributed memory model (where accessibility to remote references depends on location), memory synchronization issues can lead to migration leaks~\cite{AC11}.  In other words, distribution of memory access can be used to reveal information about location.
The present paper shows that migration leaks %
can also result from a distributed security setting combined with a program construct for inspecting the domain's flow policies.  Other forms of distribution (as for instance %
failure or communication-based distribution) and migration (eg. ambient migration~\cite{CG98}) could create other forms of dependencies.  Studying them would bring a better understanding of indirect possibilities of leaking information about the location of entities in networks.

\item %
The property of Flow Policy Confinement is perhaps the simplest form of imposing compliance of declassifications to distributed allowed flow policies.  It would be interesting to consider properties with more complex (and restrictive) concerns, possibly taking into account the history of domains that a thread has visited when defining secure code migrations.  For instance, one might want to forbid threads from moving to domains with more favorable allowed flow policies.  The enforcement of such a property would be easily achieved, using the abstractions of this paper, by introducing a condition on the allowed flow policies of origin and destination domains.  %

\item We have considered a new instance of the problem of enforcing compliance of declassifications to a dynamically changing allowed flow policy:  In our setting, changes in the allowed flow policy result from migration of programs during execution in a synchronous way.  %
  Lifting the assumption of that the domain's allowed flow policies are static, towards allowing them to change asynchronously with program execution would more realistically capture full policy dynamicity, but would be outside the realm of migration control.  Other mechanisms are therefore in demand for the scenario of a \emph{dynamic} and distributed security setting.

  \item As a more long term aim, and having in mind recent developments in the area of providing built-in security by construction into the IoT~\cite{FPRSCP16}, it would be worthwhile to explore the interplay between the technical framework that is here proposed and that is conceptual in nature, and emerging real-world global computing settings.

\end{itemize}

\bibliography{biblio}

\addtolength{\textheight}{-1cm}

{
\newpage
\appendix
\renewcommand{\semvdash}[0]{\vdash^{\secmap{},\Upsilon}}
\renewcommand{\memeqF}[3]{=^{{#1}}_{#2,#3}}
\renewcommand{\biseqt}[3]{{\approx}^{#1,#2}_{#3}}
\renewcommand{\biseqtold}[3]{{\dot{\approx}}^{#1,#2}_{#3}}

This Appendix contains the proofs of the technical results that are omitted or sketched in the paper.  %

\section{Proofs for `Controlling Information Flow'} \label{app-sec-1}

\subsection{Formalization of Distributed Non-disclosure} %

In~\cite{AC11} Non-disclosure is defined for networks, considering a distributed setting with code mobility, by means of a bisimulation on pools of threads.  In this paper we used a bisimulation on thread configurations.  We prove that the new definition is weaker than the first.

\paragraph{\emph{Definition on pools of threads.}}

In order to compare the precision of Non-Disclosure for Networks with Distributed Non-disclosure as proposed here, we recall the definition of the former, using the notations of the present paper.
\begin{defi}[$\biseqtold{W,\secmap{},\Upsilon}{\typenv}{\ell}$] \label{def-bisimdefNDN1}    %
Given a security level $\ell$, a $(W,{\secmap{}},{\Upsilon},{\typenv},{\ell})$-bisimulation is a symmetric relation $\fR$ \textbf{on pools of threads} that satisfies, for all $P_1,P_2$, and for all $({\secmap{}},{\typenv})$-compatible stores $S_1, S_2$:
\[
{P_1} ~\fR~ {P_2} ~\textit{and}~ W \vdash \iconft{T_1}{P_1} {S_1} \xarr{F}{d} \iconft{T_1'}{P_1'}{S_1'} ~\textit{and}~
\confd{T_1}{S_1} \memeqF{\secmap{1},\Upsilon}{F}{l} \confd{T_2}{S_2}
\]
with $\dom{{S_1}'}-\dom{S_1}~\cap~\dom{S_2}=\emptyset$ and $\dom{{T_1}'}-\dom{T_1}~\cap~\dom{T_2}=\emptyset$ implies that there exist ${P_2'}, {T_2'}, {S_2'}$ such that:
\[
W \vdash \iconft{T_2}{P_2}{S_2} \rarr \iconft{T_2'}{P_2'}{S_2'} ~\textit{and}~
\confd{T_1'}{S_1'} \memeqF{\secmap{1},\Upsilon}{\kltop}{l} \confd{T_2'}{S_2'} ~\textit{and}
{P_1'} ~\fR~ {P_2'}
\]
Furthermore, $S_1',S_2'$ are still $(\secmap{},\typenv)$-compatible.
The largest $(W,{\secmap{}},\Upsilon,{\typenv},\ell)$-bisimulation, the union of all such bisimulations, is denoted~$\biseqtold{W,\secmap{},\Upsilon}{\typenv}{\ell}$.
 \end{defi}
For any ${\secmap{}}$, $\Upsilon$, ${\typenv}$ and $\ell$, the set of pairs of thread configurations where threads are values is an $(W,{\secmap{}},\Upsilon,{\typenv},{\ell})$-bisimulation.  Furthermore, the union of a family of $(W,{\secmap{}},\Upsilon,{\typenv},\ell)$-bisimulations is a $(W,{\secmap{}},\Upsilon,{\typenv},\ell)$-bisimulation. Consequently, $\biseqtold{W,\secmap{},\Upsilon}{\typenv}{\ell}$ exists. %
\begin{defi}[Non-disclosure for Networks] \label{def-propertyNDN1} 
A pool of threads $P$ satisfies the Non-disclosure for Networks property with respect to an allowed-policy mapping $W$, a reference labeling $\secmap{}$, a thread labeling $\Upsilon$ and a typing environment $\typenv$, if it satisfies $P~\biseqtold{W,\secmap{},\Upsilon}{\typenv}{\ell}~P$ for all security levels $\ell$.  We then write $P\in\SecNDN(W,\secmap{},\Upsilon,\typenv)$.
\end{defi}

\hide{
When imposing restrictions on the behaviors of bisimilar pools of threads, Definition~\ref{def-propertyNDN1} %
 resets the state arbitrarily at each step of the bisimulation game.  However, as we have pointed out, in a context where migration is subjective, resetting the position tracker arbitrarily is unnecessary.  In fact, it leads to a property that is overly restrictive.  In the following example, program ${M_\textit{insec}}$ can be a direct leak that is not placed within a flow declaration:
\numbexample{
\threadnat l {\allowed F {\nil} {M_\textit{insec}}} d   \label{ex-restrictive-NDN1}
}
The above program is intuitively secure if $W(d) \klpreceq F$ and insecure otherwise, as the body of the thread is known to be executed at domain $d$.   This is reflected by the definition of Distributed Non-disclosure proposed in Definition~\ref{def-propertyDND}.  However, it is considered insecure by Definition~\ref{def-propertyNDN1}, 
as if the allowed-condition is executed over ``fresh'' thread configurations such that the thread is located at a domain where $F$ is \emph{not} allowed, then the branch with the illegal expressions ${M_\textit{insec}}$ would be executed.  Another example of a secure program that is considered insecure according to Definition~\ref{def-propertyNDN1}, but not by Definition~\ref{def-propertyDND}, is a variant of the one in Equation~(\ref{exallowedmigr}), where $d_1$=$d_2$.  That is, where the two branches are syntactically equal.
}

\paragraph{\emph{Comparison.}}

Definition~\ref{def-propertyDND} is strictly weaker than the thread pool-based Definition~\ref{def-propertyNDN1}, in the sense that it considers more programs as secure.  %
This is formalized by the following proposition. %
\begin{prop}[Proposition~\ref{prop-DNDcomparison}]\text{}
$\SecNDN(W,\secmap{},\Upsilon,\typenv) \subset \SecDND(W,\secmap{},\Upsilon,\typenv)$.
\end{prop}
\begin{proof}
We consider $P$ in $\SecNDN(W,\secmap{},\Upsilon,\typenv)$, i.e. $P$~is such that for all security levels $\ell$ we have $P~\biseqtold{W,\secmap{},\Upsilon}{\typenv}{\ell}~P$ according to Definition~\ref{def-bisimdefNDN1}.  Given any pair of position trackers $T_1,T_2$ such that $\dom{P}=\dom{T_1}=\dom{T_2}$ and $T_1 \memeqF{\secmap{1},\Upsilon}{\kltop}{l} T_2$, we prove that for all levels $\ell$ we have $\confd{P}{T_1}~\biseqt{W,\secmap{},\Upsilon}{\typenv}{\ell}~\confd{P}{T_2}$ according to Definition~\ref{def-propertyDND}.
To this end, we consider the set
\[
N\!=\!\{\confd{\confd{P_1}{T_1}}{\confd{P_2}{T_2}} ~|~ \dom{P}\!=\!\dom{T_1}\!=\!\dom{T_2} \textit{ and } T_1 \!\memeqF{\secmap{1},\Upsilon}{\kltop}{l}\! T_2 \textit{ and }
P_1~\!\biseqtold{W,\secmap{},\Upsilon}{\typenv}{\ell}~\!\!P_2 \} 
\]
and prove that $N \subseteq \biseqt{W,\secmap{},\Upsilon}{\typenv}{\ell}$ (according to Definition~\ref{def-propertyDND}).

Assume that $\confd{\confd{P_1}{T_1}}{\confd{P_2}{T_2}} \in N$, and suppose that for any given $({\secmap{}},{\typenv})$-compatible memories $S_1, S_2$ we have $W \vdash \iconft{T_1}{P_1}{S_1} \xarr{F}{d} \iconft{T_1'}{P_1'}{S_1'}$ and $\confd{T_1}{S_1} \memeqF{\secmap{1},\Upsilon}{F}{l} \confd{T_2}{S_2}$, with $\dom{{S_1}'}-\dom{S_1}~\cap~\dom{S_2}=\emptyset$ and $\dom{{T_1}'}-\dom{T_1}~\cap~\dom{T_2}=\emptyset$.  Then, by Definition~\ref{def-bisimdefNDN1} there exist ${P_2'}, {T_2'}, {S_2'}$ such that 
\[
W \vdash \iconft{T_2}{P_2}{S_2} \rarr \iconft{T_2'}{P_2'}{S_2'} ~\textit{and}~
\confd{T_1'}{S_1'} \memeqF{\secmap{1},\Upsilon}{\kltop}{l} \confd{T_2'}{S_2'} ~\textit{and}~
{P_1'} ~\fR~ {P_2'}
\]
Furthermore, $S_1',S_2'$ are still $(\secmap{},\typenv)$-compatible.  It is now easy to see that $\confd{\confd{P_1'}{T_1'}}{\confd{P_2'}{T_2'}} \in N$.

The example in Equation~(\ref{ex-restrictive-NDN1}) shows that $\SecNDN(W,\secmap{},\Upsilon,\typenv) \not= \SecDND(W,\secmap{},\Upsilon,\typenv)$. 
\end{proof}

\subsection{Type system}

\subsubsection{Subject Reduction}

In order to establish the soundness of the type system of Figure~\ref{fig-iflow-typesystem} we need a Subject Reduction result, stating that types that are given to expressions are preserved by computation.  To prove it we follow the usual steps \cite{WF94}.  In the following, $\Pse$ is the set of pseudo-values, as defined in Figure~\ref{fig-syntaxexpressions}.

\begin{rem} \text{} \label{app-remiflow-typeffval}
\begin{enumerate}
\item If $X\in\Pse$ and $\typenv\byint{j}{\secmap{}}{F} X:s,\tau$, then for all security levels $j'$, flow policies $F'$ and security effects $s'$, we have that $\typenv\byint{j'}{\secmap{}}{F'} X:s',\tau$.
\item For any flow policies $F,F'$, such that $F' \klpreceq F$, we have that ${\typenv \byint{j}{\secmap{}}{F} M : \tau}$ implies ${\typenv \byint{j}{\secmap{}}{F'} M : \tau}$.
\end{enumerate}
\end{rem}  

\begin{lem}\label{app-prop-iflow-weakstrengthlemma}\text{} 
\begin{enumerate}
\item If $\typenv\byint{j}{\secmap{}}{F} M:s,\tau$ and $x\notin\dom{\typenv}$ then $\typenv,x:\sigma\byint{j}{\secmap{}}{F} M:s,\tau$.
\item If $\typenv,x:\sigma\byint{j}{\secmap{}}{F} M:s,\tau$ and $x\notin\fv{M}$ then $\typenv\byint{j}{\secmap{}}{F} M:s,\tau$.
\end{enumerate}
\end{lem}
\begin{proof} By induction on the inference of the type judgment.   
\end{proof}

\begin{lem}[Substitution]\label{app-prop-iflow-subslemma}\text{}\\  
If $\typenv,{x:\sigma} \byint{j}{\secmap{}}{F} {M} : s, \tau$ and $\typenv \byint{j'}{\secmap{}}{F'} {X} : s',\sigma$ then $\typenv \byint{j}{\secmap{}}{F} {\substi x X M} : s, \tau$.  %
\end{lem}
\begin{proof}  By induction on the inference of $\typenv,x:\tau\byint{j}{\secmap{}}{F} M:s,\sigma$, and by case analysis on the last rule used in this typing proof, using the previous lemma. Let us examine the cases related to the new language constructs:
\begin{description} 

\item [\migtyp]

Here $M = \threadnat l {\bar{M}} {\bar d}$ and we have that $\typenv,{x:\sigma} \byint{l}{\secmap{}}{\kltop} {\bar{M}} : \bar s, \tau$, with $\tau = \unit$ %
and $\bar s = {\eft {\bot}{l \join s.w}{\bot}}$.  By induction hypothesis, then $\typenv \byint{l}{\secmap{}}{\kltop} \substi x X{\bar{M}} : \bar s, \tau$.  Therefore, by rule $\migtyp$, $\typenv \byint{j}{\secmap{}}{F} \threadnat l {\substi x X{\bar{M}}} {\bar d} : s, \tau$.

\item [\flowtyp]

Here $M = \flow {\bar F} {\bar{M}}$, $\typenv,{x:\sigma} \byint{j}{\secmap{}}{F \klmeet \bar F} {\bar{M}} : s, \tau$.  By induction hypothesis, $\typenv \byint{j}{\secmap{}}{F \klmeet \bar F} \substi{x}{X}{\bar{M}} : s, \tau$.  Then, by $\flowtyp$, we have $\typenv \byint{j}{\secmap{}}{F} {\flow {\bar{F}} {\substi{x}{X}{\bar{M}}}} : s, \tau$.  

\item [\allowtyp]

Here $M = \allowed {\bar{F}} {N_t} {N_f}$ and we have $\typenv,{x:\sigma} \byint{j}{\secmap{}}{F} {N_t} : s_t, \tau$ and $\typenv,{x:\sigma} \byint{j}{\secmap{}}{F} {N_f} : s_f, \tau$ with $j \Fpreceq{F} s_t.w, s_f.w$ and $s = s_t \join s_f\join \eft{j}{\top}{j}$.  By induction hypothesis, $\typenv,{x:\sigma} \byint{j}{\secmap{}}{F} {\substi x X {N_t}} : s_t, \tau$ and $\typenv,{x:\sigma} \byint{j}{\secmap{}}{F} {\substi x X {N_f}} : s_f, \tau$.  Therefore, by rule $\allowtyp$, we have that $\typenv,{x:\sigma} \byint{j}{\secmap{}}{F} {\allowed {\bar F} {\substi x X{N_t}} {\substi x X{N_f}}} : s, \tau$. \qedhere

\end{description}
\end{proof}

\begin{lem}[Replacement] \label{app-prop-iflow-repllemma}\text{}\\ %
If $\typenv \byint{j}{\secmap{}}{F} \EC{M} : s, \tau$ is a valid judgment, then the proof gives $M$ a typing $\typenv \byint{j}{\secmap{}}{{F} \klmeet \extrf {\cE{E}}}{M} : \bar{s}, \bar{\tau}$ for some $\bar{s}$ and $\bar{\tau}$ such that ${\bar{s}} \preceq {s}$.
In this case, if $\typenv \byint{j}{\secmap{}}{{F} \klmeet \extrf {\cE{E}}} {N} : \bar{s}',\bar{\tau}$ with ${\bar{s}}' \preceq \bar{s}$, then $\typenv \byint{j}{\secmap{}}{F} \EC{N} : {s'}, \tau$, for some ${s'}$ such that ${s'} \preceq s$.
\end{lem}

\begin{proof}  By induction on the structure of $\cE{E}$.  Let us examine the case of the flow declaration, which is the only non-standard evaluation context: {
\begin{description}

\item [$\boldsymbol{\cC{E}{M} = \flow {F'} {\cC{\bar{E}}{M}}}$]

By $\flowtyp$, we have $\typenv \byint{j}{\secmap{}}{F \klmeet F'} \cC{\bar{E}}{M} : s, \tau$.  By induction hypothesis, the proof gives $M$ a typing $\typenv \byint{j}{\secmap{}}{F \klmeet F' \klmeet \extrf{\cE{\bar{E}}}} {M} : \hat{s}, \hat{\tau}$, for $\hat{s},\hat{\tau}$ such that ${\hat{s}} \preceq {s}$.
Also by induction hypothesis, $\typenv \byint{j}{\secmap{}}{F \klmeet F'} \cC{\bar{E}}{N} : s', \tau$, for some ${s'}$ such that ${s'} \preceq s$.  Then, again by $\flowtyp$, we have $\typenv \byint{j}{\secmap{}}{F} \flow {F'}{\cC{\bar{E}}{N}} : s', \tau$. \qedhere
\end{description}}
\end{proof}

We check that the type of a thread is preserved by reduction, while its effects ``weaken''.

\begin{prop}[Subject Reduction -- Proposition~\ref{prop-iflow-subjectreduction}] \text{}\\\label{app-prop-iflow-subjectreduction}
Consider a thread $M^{m}$ such that ${\typenv \byint{\Upsilon(m)}{\secmap{}}{F} M : s, \tau}$, and suppose that $W \vdash \iconft{T}{\{M^{m}\}} S$ $\xarr{\var {F'}}{d}$ $\iconft{T'} {\{M'^{m}\} \cup P} {S'}$, for a memory $S$ that is $({\secmap{}},{\typenv})$-compatible. Then, there is an effect $s'$ such that $s' \preceq s$ and $\typenv\byint{\Upsilon(m)}{\secmap{}}{F} M':s',\tau$, and $S'$ is also $({\secmap{}},{\typenv})$-compatible.
Furthermore, if~$P = \{N^n\}$, for some expression $N$ and thread name $n$, then there exists %
$s''$ such that $s'' \preceq s$ %
and $\typenv\byint{\Upsilon(n)}{\secmap{}}{\kltop} N:s'',\unit$.
\end{prop}
\begin{proof}

Suppose that $M=\cC{\bar{E}}{\bar{M}}$ and $W \semvdash \iconft{T}{\{{\bar{M}}^m\}}{S}$ $\xarr{\bar{F}}{d}$ $\iconft{\bar T'}{\{\bar{M}'^m\} \cup P'}{\bar{S}'}$.  We start by observing that this implies $F = \bar{F} \klmeet \extrf{\cE{\bar{E}}}$, $M'=\cC{\bar{E}}{\bar{M}'}$, $P = P'$, ${\bar{T}'} = {T'}$ and ${\bar{S}'} = {S'}$.  
We can assume, without loss of generality, that $\bar{M}$ is the smallest in the sense that there is no $\cE{\hat{E}},\hat{M},\hat{N}$ such that $\cE{\hat{E}}\neq[]$ and $\cC{\hat{E}}{\hat{M}}=\bar{M}$ for which we can write $W \semvdash \iconft{T}{\{\hat{M}^m\}}{S}$ $\xarr{\hat{F}}{d}$ $\iconft{T'}{\{\hat{M}'^m\} \cup P}{S'}$.  %

By Replacement (Lemma~\ref{app-prop-iflow-repllemma}), we have $\typenv \byint{\Upsilon(m)}{\secmap{}}{{F} \klmeet \extrf{\bar{\cE{E}}}} {{\bar{M}}} : \bar{s}, \bar{\tau}$ in the proof of $\typenv \byint{\Upsilon(m)}{\secmap{}}{F} {\cC{\bar{E}}{\bar{M}}} : s, \tau$, for some $\bar{s}$ and $\bar{\tau}$.  
We proceed by case analysis on the transition $W \semvdash \iconft{T}{\{{\bar{M}}^m\}}{S}$ $\xarr{\bar{F}}{d}$ $\iconft{T'}{\{\bar{M}'^m\} \cup P}{S'}$, and prove that if $S' \neq S$ then:
\begin{itemize}
\item There is an effect $\bar{s}'$ such that $\bar{s}' \preceq \bar{s}$ and ${\typenv \byint{\Upsilon(m)}{\secmap{}}{{F} \klmeet \extrf{\bar{\cE{E}}}} {{\bar{M}'} : \bar{s}', \bar{\tau}}}$.  Furthermore, for every reference $a \in \dom{S'}$ implies ${\typenv \byint{\hat j}{\secmap{}}{\hat F} S'(a) : \hat s, \secmap{2}(a)}$, for every security level $\hat j$, flow policy $\hat F$ and security effect $\hat s$.
\item If~$P=\{N^n\}$~for some expression $N$ and thread name $n$, then there is an effect $\bar{s}''$ such that $\bar{s}.w \preceq \bar{s}''.w$ and a thread name $d'$ such that ${\typenv \byint{\Upsilon(n)}{\secmap{}}{\kltop} {N : \bar{s}'', \unit}}$.  
\end{itemize}
We present only the cases corresponding to reductions of the non-standard language constructs:
\begin{description}
\item [$\boldsymbol{\bar{M} = \allowed{F'} {N_t}{N_f}}$]
Suppose that $W(d) \klpreceq F'$ (the other case is analogous).
Here we have $\bar{M}' = N_t$, $S=S'$ and $P=\emptyset$.  By $\allowtyp$, we have that $\typenv \byint{\Upsilon(m)}{\secmap{}}{{F \klmeet \extrf{\bar{\cE{E}}}}} {N_t} : s_t, \bar{\tau}$, where ${s_t} \preceq \bar{s}$.
\item [$\boldsymbol{\bar{M} = \flow {F'} {V}}$]

Here we have $\bar{M}' = V$, $S=S'$ and $P=\emptyset$.  By rule $\flowtyp$, we have that $\typenv \byint{\Upsilon(m)}{\secmap{}}{{F \klmeet \extrf{\bar{\cE{E}}}} \klmeet F'} V : {\tef {{\bar s}} {{\tau}}}$.
and by Remark~\ref{app-remiflow-typeffval}, we have $\typenv \byint{\Upsilon(m)}{\secmap{}}{{F \klmeet \extrf{\bar{\cE{E}}}}} V : {\tef {\bar s} {\bar{\tau}}}$.
\item [$\boldsymbol{\bar{M} = \threadnat k {N} {d'}}$]

Here we have $\bar{M}' = \nil$, $P=\{N^n\}$ for some thread name $n$, $S=S'$ and $T'(n)=d$.  By $\migtyp$, we have that $\typenv \byint{\Upsilon(n)}{\secmap{}}{\kltop} {N} : \hat s,\unit$, with $s.w \preceq \hat s.w$ and $\bar \tau = \unit$, and by $\niltyp$ we have that $\typenv \byint{\Upsilon(m)}{\secmap{}}{{F \klmeet \extrf{\bar{\cE{E}}}}} {\nil} : \bar s, \unit$.
\end{description}
By Replacement (Lemma~\ref{app-prop-iflow-repllemma}), we can finally conclude that $\typenv\byint{\Upsilon(m)}{\secmap{}}{F} {\cC{\bar{E}}{\bar{M}'}}: s', {\tau}$, for some $s'~\preceq~s$.
\end{proof}

\subsubsection{Basic Properties}

\paragraph{\emph{Properties of the Semantics.}}

One can show that if a thread is able to perform a step in one memory, while creating a new thread (or not), when executing on a low-equal memory it can also perform a step and produce a low-equal result, while creating a new thread (or not).
\begin{lem}[Guaranteed Transitions] \label{app-prop-iflow-guarantee} %
Consider an allowed-policy mapping $W$, a thread $M^m$ and two states $\confd {T_1}{S_1}, \confd {T_2}{S_2}$ such that $W \vdash \iconft {T_1} {\{{M}^m\}} {S_1}$ $\xarr{F}{d}$ $\iconft {T_1'}{P_1'}{S_1'}$, and for some $F'$ we have $\confd {T_1}{S_1} \memeqF{\secmap{1},\Upsilon}{F \klmeet F'}{\low} \confd {T_2}{S_2}$.  Then: 
\begin{itemize}
\item If $P_1' = \{M_1'^m\}$, and $(\dom{{S_1}'}-\dom{S_1})\cap\dom{S_2}=\emptyset$, %
then there exist $M_2'$, $T_2'$ and $S_2'$ such that $W \vdash \iconft {T_2}{\{{M}^m\}}{S_2} \xarr{F}{d} \iconft{T_2'}{\{M_2'^m\}}{S_2'}$ and $\confd {T_1'}{S_1'} \memeqF{\secmap{1},\Upsilon}{F \klmeet F'}{\low} \confd {T_2'}{S_2'}$.
\item If~$P_1' = \{M'^m,N^n\}$~for some expression $N$ and $(\dom{{T_1}'}-$ $\dom{T_1})\cap\dom{T_2}=\emptyset$, %
then there exist $M_2'$, $T_2'$ and $S_2'$ such that we have $W \vdash \iconft {T_2}{\{{M}^m\}}{S_2} \xarr{F}{d} \iconft{T_2'}{\{M'^m, N^n\}}{S_2'}$ and $\confd {T_1'}{S_1'}\memeqF{\secmap{1},{{\Upsilon}}}{F \klmeet F'}{\low} \confd {T_2'}{S_2'}$.
\end{itemize}
\end{lem}
\begin{proof}
By case analysis on the proof of $W \vdash \iconft {T_1} {\{{M}^m\}} {S_1}$ $\xarr{F}{d}$ $\iconft {T_1'}{P_1'}{S_1'}$. In most cases, this transition does not modify or depend on the state $\confd {T_1}{S_1}$, and we may let $P_2'=P_1'$ and $\confd{T_2'}{S_2'}=\confd{T_2}{S_2}$. %
\begin{description}
\item [$\boldsymbol{M=\cC{E}{\allowed A {N_t} {N_f}}}$ \textbf{and} $\boldsymbol{W(T_1(m)) \klpreceq A}$]  Here, $P_1'=\{\cC{E}{N_t}^m\}$, $F = \extrf{\cE{E}}$, and $\confd{T_1'}{S_1'}=\confd{T_1}{S_1}$.  There are two possible cases:
\begin{itemize}
\item If $W(T_2(m)) \klpreceq A$, then $W \vdash \iconft {T_2}{\{{M}^m\}}{S_2} \xarr{F}{d} \iconft{T_2}{\{{N_t}^m\}}{S_2}$.  
\item If $W(T_2(m)) \not\klpreceq A$, then $W \vdash \iconft {T_2}{\{{M}^m\}}{S_2} \xarr{F}{d} \iconft{T_2}{\{{N_f}^m\}}{S_2}$.  
\end{itemize}
Clearly, in both cases, by assumption, $\confd {T_1}{S_1} \memeqF{\secmap{1},\Upsilon}{F \klmeet F'}{\low} \confd {T_2}{S_2}$.

The case where ${W(T_1(m)) \not\klpreceq A}$ is analogous.
\item [$\boldsymbol{M=\cC{E}{\threadnat {k} {N} {d'}}}$]  Here, for a thread name $n$, $P_1'=\{\cC{E}{\nil}^m, {N}^{n}\}$, $F = \extrf{\cE{E}}$, $T_1'=T_1 \cup \{n \mapsto d'\}$, and ${S_1'}={S_1}$.  Since we assume that $n \notin \dom{T_2}$, we also have that $W \vdash \iconft {T_2}{\{{M}^m\}}{S_2}$ $\xarr{F}{d}$ $\iconft{T_2 \cup \{n \mapsto d'\}}{\{{\cC{E}{\nil}}^m, N^n\}}{S_2}$.  Clearly, $\confd {T_1 \cup \{n \mapsto d'\}}{S_1}$ $\memeqF{\secmap{1},{{\Upsilon}}}{F \klmeet F'}{\low}$ $\confd {T_2 \cup \{n \mapsto d'\}}{S_2}$. \qedhere
  \end{description}
\end{proof}

\hide{
The following result states that, if the evaluation of a thread $M^m$ differs in the context of two distinct states while not creating two distinct reference names or thread names, this is because either $M^m$ is performing a dereferencing operation, which yields different results depending on the memory, or because $M^m$ is testing the allowed policy. %

\begin{lem}[Splitting Computations -- Lemma~\ref{prop-iflow-split}]\label{app-prop-iflow-split}\text{}\\ 
If we have $W \vdash \iconft{T_1}{\{M^m\}}{S_1}$ $\xarr{F}{d}$ $\iconft{T_1'}{P_1'}{S_1'}$ and $W \vdash \iconft{T_2}{\{M^m\}}{S_2}$ $\xarr{F'}{d}$ $\iconft{T_2'}{P_2'}{S_2'}$ with ${P_1}' \neq {P_2}'$, then $P_1' = {\{{M_1}'^m\}}$, $P_2' = {\{{M_2}'^m\}}$ and either:
\begin{itemize}
\item $\exists \cE{E}, a$ such that $F=\secpolcon{\cE{E}}=F'$, $M=\cC{E}{\deref {a}}$, and $M_1'=\cC{E}{S_1(a)}$, $M_2'=\cC{E}{S_2(a)}$ with $\confd{T_1'}{S_1'}=\confd{T_1}{S_1}$ and $\confd{T_2'}{S_2'}=\confd{T_2}{S_2}$, or
\item $\exists \cE{E}, \bar F, {N_t}, {N_f}$ such that $F=\secpolcon{\cE{E}}=F'$, $M=\cC{E}{\allowed {\bar F}{N_t}{N_f}}$, and $T_1(m) \neq T_2(m)$ with $\confd{T_1'}{S_1'}=\confd{T_1}{S_1}$ and $\confd{T_2'}{S_2'}=\confd{T_2}{S_2}$.
\end{itemize}
\end{lem}
\begin{proof}
By case analysis on the transition $W \vdash \iconft{T_1}{\{M^m\}}{S_1}$ $\xarr{F}{d}$ $\iconft{T_1'}{P_1'}{S_1'}$.
Note that the only rules that depend on the state are those for the reduction of $\cC{E}{\deref {a}}$ and of $\cC{E}{\allowed {F'}{N_t}{N_f}}$.  
\end{proof} 
}

\paragraph{\emph{Effects.}}

\begin{lem}[Update of Effects]\label{app-prop-iflow-updateeffects}\text{}    %
\begin{enumerate}
\item If $\typenv \byint{j}{\secmap{}}{F} {\EC {\deref a}} : s, \tau$ then $\secmap{1}(a) \preceq s.r$.
\item If $\typenv \byint{j}{\secmap{}}{F} {\EC{\assign a {V}}} : s, \tau$, then $s.w \preceq \secmap{1}(a)$.
\item If $\typenv \byint{j}{\secmap{}}{F} {\EC{\rfr {l,\theta} V}} : s, \tau$, then $s.w \preceq l$.
\item If $\typenv \byint{j}{\secmap{}}{F} {\EC{\threadnat l M {d'}}} : s, \tau$, then $s.w \preceq l$.  %
\item If $\typenv \byint{j}{\secmap{}}{F} {\EC{\allowed F {N_t}{N_f}}} : s, \tau$, then $j \preceq s.r, s.t$.
\end{enumerate}
\end{lem}
\begin{proof}  By induction on the structure of $\cE{E}$.  \end{proof}

\paragraph{\emph{High Expressions.}}

Some expressions can be easily classified as ``high'' by the type system, which only considers their syntax.  These cannot perform changes to the ``low'' memory simply because their code does not contain any instruction that could perform them.  Since the \mentionind{security effect!writing effect}{writing effect} is intended to be a lower bound to the level of the references that the expression can create or assign to, expressions with a high writing effect can be said to be \emph{syntactically high}:
\begin{defi}[Syntactically High Expressions] \label{app-def-synhigh} %
An expression $M$ is syntactically $(\secmap{},\typenv,j,F,l)$-high if there exists $s,\tau$ such that ${\typenv \byint{j}{\secmap{}}{F} M : s, \tau}$ with $s.w \not\Fpreceq{F} l$. 
The expression $M$ is a syntactically $(\secmap{},\typenv,j,F,l)$-high function if there exists $j',F',s,\tau,\sigma$ such that ${\typenv \byint{j'}{\secmap{}}{F'} M: \bot, \tau \xarr{j,F}{s} \sigma}$ with $s.w\not\Fpreceq{F}~l$.
\end{defi}

Syntactically high expressions have an operationally high behavior.
\begin{lem}[High Expressions] \label{app-prop-iflow-highexpr} %
If $M$ is a syntactically $(\secmap{},\typenv,j,F,l)$-high expression, and $\Upsilon$ is such that $\Upsilon(m)=j$, then, for all allowed-policy mappings $W$, the thread $M^m$ is an operationally $(W,\secmap{},\Upsilon,\typenv,F,l)$-high thread.
\end{lem}
\begin{proof} %
We show that, for any given allowed-policy mapping $W$, if $\Upsilon(m)=j$, then the set 
\[ \{ M^m ~|~ \exists j ~.~ M \textit{ is syntactically $(\secmap{},\typenv,j,F,l)$-high}\} \]
is a set of operationally $(W,\secmap{},\Upsilon,\typenv,F,\ell)$-high threads, i.e.: if $M$ is syntactically $(\secmap{},\typenv,j,F,l)$-high, that is if there exists $s,\tau$ such that ${\typenv \byint{j}{\secmap{}}{F} M : s, \tau}$ with $s.w \not\Fpreceq{F} l$, and, for some policy mapping $W$ and all states $\confd T S$ such that $S$ is $({\secmap{}},{\typenv})$-compatible, if $W \semvdash \iconft{T}{\{M^m\}}{S} \xarr{F'}{d} \iconft{T'}{\{M'^m\} \cup P}{S'}$ then $\confd T S \memeqF{\secmap{1},\Upsilon}{F}{l} \confd {T'} {S'}$.  This is enough since, by Subject Reduction (Theorem \ref{app-prop-iflow-subjectreduction}), $M'$ is syntactically $(\secmap{},\typenv,j,F,l)$-high and $S'$ is still $(\secmap{},\typenv)$-compatible, and if $P=\{N^n\}$ for some expression $N$ and thread name $n$, then by Remark~\ref{app-remiflow-typeffval} also $N$ is syntactically $(\secmap{},\typenv,k,F,l)$-high for some $k$.
We proceed by cases on the proof of the transition $W \semvdash \iconft{T}{\{M^m\}}{S}$ $\xarr{F'}{d}$ $\iconft{T'}{\{M'^m\} \cup P}{S'}$. The lemma is trivial in all the cases where $\confd{T}{S} = \confd{T'}{S'}$.
\begin{description}
\item [$\boldsymbol{M=\cC{E}{\threadnat k N {d'}}}$]  Here $P = \{N^n\}$~for some thread name $n$, $S'=S$, $T'=\update n {d'} T$ and $\Upsilon(n)=k$.  By Update of Effects (Lemma~\ref{app-prop-iflow-updateeffects}), $\weff{s}\preceq k$. This implies $k\not\Fpreceq{F}l$, thus $\Upsilon(n)\not\Fpreceq{F}l$, hence $T'\memeqF{\secmap{1},{\Upsilon}}{F}{l}T$. \qedhere
\end{description}
\end{proof}

\subsubsection{Soundness}

\paragraph{\emph{Behavior of ``Low''-Terminating Expressions.}}

According to the intended meaning of the termination effect, the termination or non-termination of expressions with low termination effect should only depend on the low part of the state.  In other words, two computations of a same typable thread running under two ``low''-equal states should either both terminate or both diverge.  In particular, this implies that termination-behavior of these expressions cannot be used to leak ``high'' information when composed with other expressions (via \mentionind{security leak!termination leak}{termination leaks}).

We aim at proving that any typable thread $M^m$ that has a low-termination effect always presents the same behavior according to a \emph{strong} bisimulation on low-equal states:  if two continuations $M_1^m$ and $M_2^m$ of $M^m$ are related, and if $M_1^m$ can perform an execution step over a certain state, then $M_2^m$ can perform the same low changes to any low-equal state in precisely one step, while the two resulting continuations are still related.
  This implies that any two computations of $M^m$ under low-equal states should have the same ``length'', and in particular they are either both finite or both infinite.  To this end, we design a reflexive binary relation on expressions with low-termination effects that is closed under the transitions of Guaranteed Transitions (Lemma \ref{app-prop-iflow-guarantee}). 

The inductive definition of $\fT^{\secmap{},\typenv}_{j,F,\low}$ is given in Figure~\ref{app-figfT}.  Notice that it is a symmetric relation.  In order to ensure that expressions that are related by $\fT^{\secmap{},\typenv}_{j,F,\low}$ perform the same changes to the low memory, its definition requires that the references that are created or written using (potentially) different values are high.
\begin{figure}
\figline
\begin{defi}[$\fT^{\secmap{},\typenv}_{j,F,\low}$]  
 We have that $M_1 ~\fT^{\secmap{},\typenv}_{j,F,\low}~ M_2$ if ${\typenv \byint{j}{\secmap{}}{F} M_1 : s_1, \tau}$ and ${\typenv \byint{j}{\secmap{}}{F} M_2 : s_2, \tau}$ for some $s_1$, $s_2$ and $\tau$ with $s_1.t \Fpreceq{F} \low$ and $s_2.t \Fpreceq{F} \low$ and one of the following holds:
\begin{description}
\item[Clause 1]$M_1$ and $M_2$ are both values, or
\item[Clause 2]$M_1=M_2$, or   %
\item[Clause 3]$M_1=\seq {\bar M_1}{\bar N}$ and $M_2=\seq {\bar M_2}{\bar N}$ where $\bar M_1~\fT^{\secmap{},\typenv}_{j,F,\low}~\bar M_2$, or %
\item[Clause 4]$M_1=\rfr{l,\theta} {\bar M_1}$ and $M_2=\rfr{l,\theta}{\bar M_2}$ where  $\bar M_1~\fT^{\secmap{},\typenv}_{j,F,\low}~\bar M_2$, and $l \not\Fpreceq{F} \low$, or
\item[Clause 5]$M_1=\deref {\bar M_1}$ and $M_2=\deref{\bar M_2}$ where  $\bar M_1~\fT^{\secmap{},\typenv}_{j,F,\low}~\bar M_2$, or
\item[Clause 6]$M_1=\assign{\bar{M}_1} {\bar{N}_1}$ and $M_2=\assign{\bar{M}_2}{\bar{N}_2}$ with $\bar{M}_1~\fT^{\secmap{},\typenv}_{j,F,\low}~\bar{M}_2$, and $\bar{N}_1~\fT^{\secmap{},\typenv}_{j,F,\low}~\bar{N}_2$, and $\bar{M}_1, \bar{M}_2$ both have type $\rfrt {{\theta}}{l}$ for some ${\theta}$ and $l$ such that ${l} \not\Fpreceq{F} \low$, or
\item[Clause 7]$M_1=\flow{F'}{\bar M_1}$ and $M_2=\flow{F'}{\bar M_2}$ with $\bar M_1~\fT^{\secmap{},\typenv}_{j,F \klmeet F',\low}~\bar M_2$. %
\end{description}
\end{defi}

\caption{The relation $\fT^{\secmap{},\typenv}_{j,F,\low}$} \label{app-figfT}
\figline
\end{figure}

\begin{rem}\label{app-remiflow-valT}
If for $\secmap{}$, $\typenv$, $j$, $F$ and $\low$ we have $M_1 ~\fT^{\secmap{},\typenv}_{j,F,\low}~ M_2$ and $M_1\in\Val$, then $M_2\in\Val$.
\end{rem}

From the following lemma one can conclude that the relation $\fT^{\secmap{},\typenv}_{j,F,\low}$ relates the possible outcomes of expressions that are typable with a low termination effect, and that perform a high read over low-equal memories.
\begin{lem} \label{app-prop-iflow-forkt}
If %
  $\typenv\byint{j}{\secmap{}}{F} \EC{\deref {a}}:s,\tau$ with ${s}.t\Fpreceq{F} \low$ and $\secmap{1}(a)\not\preceq_{F\klmeet\secpolcon{\cE{E}}} \low$, 
then for any values $V_0,V_1 \in \Val$ such that $\typenv\byint{j}{\secmap{}}{\kltop} V_i:\bot,\theta$ we have $\cC{E}{V_0}~\fT^{\secmap{},\typenv}_{j,F,\low}~\cC{E}{V_1}$.
\end{lem}
\begin{proof} By induction on the structure of $\cE{E}$.
\begin{description}

\item [$\boldsymbol{\cC{E}{\deref {a}}= {\flow {F'}{\cC{E_1}{\deref {a}}}}}$]  By rule $\flowtyp$ we have $\typenv \byint{j}{\secmap{}}{F \klmeet F'} V : {\tef s {\tau}}$.  By induction hypothesis $\cC{E_1}{V_0}~\fT^{j}_{F \klmeet F',\low}~\cC{E_1}{V_1}$, so we conclude by Replacement (Lemma~\ref{app-prop-iflow-repllemma}) and Clause 7.
Therefore $\bar{s}.t \Fpreceq{F} \low$, and since $\secmap{1}(a)\not\Fpreceq{F \klmeet \cE{E}} \low$ implies $\secmap{1}(a)\not\Fpreceq{F \klmeet \cE{E_1}} \low$, then by induction hypothesis we have $\cC{E_1}{V_0}~\fT^{\secmap{},\typenv}_{j,F,\low}~\cC{E_1}{V_1}$.  By Replacement (Lemma~\ref{app-prop-iflow-repllemma}) and Clause 8 we can conclude. \qedhere
\end{description}
\end{proof}

We can now prove that $\fT^{\secmap{},\typenv}_{j,F,\low}$ behaves as a ``kind of'' \mentionind{bisimulation!for non-disclosure for networks}{strong bisimulation}:
\begin{prop}[Strong Bisimulation for Low-Termination]\label{app-prop-iflow-strong}\text{}\\ 
Consider a given allowed-policy mapping $W$, reference labeling $\secmap{}$, thread labeling $\Upsilon$, typing environment $\typenv$, flow policy $F$, security level $\low$, two expressions $M_1$ and $M_2$ and thread name $m$.  If, for states $\confd {T_1}{S_1}$,$\confd {T_2}{S_2}$ with $S_1, S_2$ being $({\secmap{}},{\typenv})$-compatible we have that: %
\myexample{
M_1~\fT^{\secmap{},\typenv}_{\Upsilon(m),F,\low}~M_2 ~\textit{and}~ %
W \vdash \iconft {T_1} {\{{M_1}^m\}} {S_1} \xarr{F'}{d} \iconft {T_1'}{P_1'}{S_1'} ~\textit{and}~ \confd {T_1}{S_1} \memeqF{\secmap{1},\Upsilon}{F \klmeet F'}{\low} \confd {T_2}{S_2}, 
}
with $(\dom{{S_1}'}-\dom{S_1})\cap\dom{S_2}=\emptyset$, %
then there exist $P_2'$, $T_2'$ and $S_2'$ such that 
\myexample{
W \vdash \iconft{T_2}{\{{M_2}^m\}}{S_2}\xarr{F'}{d} \iconft{T_2'}{P_2'}{S_2'} ~\textit{and}~ M_1'~\fT^{\secmap{},\typenv}_{j,F,\low}~M_2' ~\textit{and}~ \confd {T_1'}{S_1'} \memeqF{\secmap{1},\Upsilon}{F \klmeet F'}{\low} \confd {T_2'}{S_2'}
}
Furthermore, if $P_1'=\{M_1'^m\}$ then $P_2'=\{M_2'^m\}$, if $P_1'=\{M_1'^m,N^n\}$ for some thread $N^n$ and $(\dom{{T_1}'}-$ $\dom{T_1})\cap\dom{T_2}=\emptyset$ %
then $P_2'=\{M_2'^m,N^n\}$, and $S_1', S_2'$ are still $({\secmap{}},{\typenv})$-compatible.
\end{prop}
\begin{proof} 
By case analysis on the clause by which $M_1~\fT^{\secmap{},\typenv}_{j,F,\low}~M_2$, and by induction on the definition of $\fT^{\secmap{},\typenv}_{j,F,\low}$.  In the following, we use Subject Reduction (Theorem~\ref{app-prop-iflow-subjectreduction}) to guarantee that the termination effect of the expressions resulting from $M_1$ and $M_2$ is still low with respect to $\low$ and $F$.  This, as well as typability (with the same type) for $j$, $F$ and ${\low}$, is a requirement for being in the $\fT^{\secmap{},\typenv}_{j,F,\low}$ relation.
\begin{description}

\item [Clause 2]  Here $M_1 = M_2$.  By Guaranteed Transitions (Lemma \ref{app-prop-iflow-guarantee}), then:
  \begin{itemize}
  \item If~$P_1' = \{M_1'^m\}$, and $(\dom{{S_1}'}-\dom{S_1})\cap\dom{S_2}=\emptyset$, %
then there exist $M_2'$, $T_2'$ and $S_2'$ such that $W \vdash \iconft {T_2}{\{{M}^m\}}{S_2} \xarr{F'}{d} \iconft{T_2'}{\{M_2'^m\}}{S_2'}$ and $\confd {T_1'}{S_1'} \memeqF{\secmap{1},\Upsilon}{F \klmeet F'}{\low} \confd {T_2'}{S_2'}$.  There are two cases to consider:
    \begin{description}
    \item [$\boldsymbol{M_2'=M_1'}$] Then we have $M_1'~\fT^{\secmap{},\typenv}_{j,F,\low}~M_2'$, by Clause 2 and Subject Reduction (Theorem \ref{app-prop-iflow-subjectreduction}).
    \item [$\boldsymbol{M_2' \neq M_1'}$] Then by Splitting Computations (Lemma \ref{app-prop-iflow-split}) we have two possibilities:  \\
      (1)~ there exist $\cE{E} $ and $a$ such that $M_1'=\cC{E}{S_1(a)}$, $F'=\extrf{\cE{E}}$, $M_2'=\cC{E}{S_2(a)}$, $\confd{T_1'}{S_1'}=\confd{T_1}{S_1}$ and $\confd{T_2'}{S_2'}=\confd{T_2}{S_2}$.  Since $S_1(a)\ne S_2(a)$, we have that $\secmap{1}(a)\not\Fpreceq{F\klmeet F'}\low$.  Therefore, $M_1'~\fT^{\secmap{},\typenv}_{j,F,\low}~M_2'$, by Lemma~\ref{app-prop-iflow-forkt} above. \\
      (2)~ there exists $\cE{E}$ such that $M_1'=$ $\cC{E}{\allowed {\bar F}{N_t}{N_f}}$, $F'=\secpolcon{\cE{E}}$, and $T_1(m) \neq T_2(m)$ with $\confd{T_1'}{S_1'}=\confd{T_1}{S_1}$ and $\confd{T_2'}{S_2'}=\confd{T_2}{S_2}$.  Since $T_1(m)\ne T_2(m)$, we have $\Upsilon(m) = j \not\Fpreceq{F}\low$, and by Update of Effects (Lemma~\ref{app-prop-iflow-updateeffects}) we have $j\Fpreceq{F} s.t$, so $s.t \not\Fpreceq{F}\low$, which contradicts the assumption.
    \end{description}
  \item If~$P_1' = \{M_1'^m,N^n\}$~for some expression $N$ and $(\dom{{T_1}'}-$ $\dom{T_1})\cap\dom{T_2}=\emptyset$, %
then for some $M$, $l$ and $d$ we have $M_1=\cC{E}{{\threadnat l M d}}$, and there exist $M_2'$, $T_2'$ and $S_2'$ such that $W \vdash \iconft {T_2}{\{{M}^m\}}{S_2} \xarr{F}{d} \iconft{T_2'}{\{M_2',N^n\}}{S_2'}$ and $\confd {T_1'}{S_1'}\memeqF{\secmap{1},{\Upsilon}}{F \klmeet F'}{\low} \confd {T_2'}{S_2'}$.  By Splitting Computations (Lemma \ref{app-prop-iflow-split}), necessarily $M_1'^m=M_2'^m$.  Then we have $M_1'~\fT^{\secmap{},\typenv}_{j,F,\low}~M_2'$, by Clause 2 and Subject Reduction (Theorem \ref{app-prop-iflow-subjectreduction}).
  \end{itemize}
\item [Clause 7] Here we have $M_1=\flow{\bar{F}}{\bar{M}_1}$ and $M_2=\flow{\bar{F}}{\bar{M}_2}$ and $\bar{M}_1~\fT^{\secmap{},\typenv}_{j,F \klmeet \bar F,\low}$ $\bar{M}_2$.  There are two cases.
  \begin{description}
    \item [$\boldsymbol{\bar{M}_1}$ can compute] In this case we have $M_1'=\flow{\bar{F}}{\bar{M}_1'}$ with $W \vdash \iconft{T_1} {\{{\bar{M}_1}^m\}} {S_1}$ $\xarr{F''}{d}$ $\iconft{T_1'}{\{\bar{M}_1'^m\}}{S_1'}$ with $F'=\bar{F} \klmeet F''$.  To use the induction hypothesis, there are three possible cases:
    \begin{itemize}
    \item If~$\bar P_1' = \{\bar{M}_1'^m \}$, and $(\dom{{S_1}'}-\dom{S_1})\cap\dom{S_2}=\emptyset$, %
then there exist $\bar{M}_2'$, $T_2'$ and $S_2'$ such that $W \vdash \iconft{T_2}{\{{\bar{M}_2}^m\}}{S_2} \xarr{F''}{d} \iconft{T_2'}{\{\bar{M}_2'^m\}}{S_2'}$ with $\bar{M}_1'~\fT^{\secmap{},\typenv}_{j,F \klmeet \bar F,\low}~\bar{M}_2'$ and $\confd {T_1'}{S_1'} \memeqF{\secmap{1},\Upsilon}{F \klmeet F' \klmeet \bar F}{\low} \confd {T_2'}{S_2'}$.  Notice that $\confd {T_1'}{S_1'} \memeqF{\secmap{1},\Upsilon}{F \klmeet F'}{\low} \confd {T_2'}{S_2'}$.
    \item If~$\bar{P}_1' = \{\bar{M}_1'^m, N^n\}~$ for some expression $N$ and $(\dom{{T_1}'}-$ $\dom{T_1})\cap\dom{T_2}=\emptyset$, %
then there exist $\bar{M}_2'$, $T_2'$, $S_2'$ such that $W\vdash  \iconft{T_2}{\{\bar{M}_2^m\}}{S_2} \xarr{F''}{d} \iconft{T_2'}{\{\bar{M}_2'^m, {N}^n\}}{S_2'}$ with $\bar{M}_1'~\fT^{\secmap{},\typenv}_{j,F \klmeet \bar F,\low}~\bar{M}_2'$ and $\confd {T_1'}{S_1'} \memeqF{\secmap{1},{\Upsilon}}{F \klmeet F' \klmeet \bar F}{\low} \confd {T_2'}{S_2'}$.  Notice that we have $\confd {T_1'}{S_1'}$ $\memeqF{\secmap{1},{\Upsilon}}{F \klmeet F'}{\low} \confd {T_2'}{S_2'}$.
    \end{itemize}
In all three cases, we use Clause 7 and Subject Reduction (Theorem \ref{app-prop-iflow-subjectreduction}) to conclude. 
    \item [$\boldsymbol{\bar{M}_1}$ is a value]  In this case $P_1' = \{\bar M_1^m\}$, $F'=\kltop$ and $\confd{T_1'}{S_1'}=\confd{T_1}{S_1}$. Then $\bar{M}_2\in\Val$ by Remark~\ref{app-remiflow-valT}, so $W \vdash \iconft{T_2}{\{{M_2}^m\}}{S_2}$ $\xarr{F'}{d}$ $\iconft{T_2}{\{{\bar{M}_2}^m\}}{S_2}$.  We conclude using Clause 1 and Subject Reduction (Theorem \ref{app-prop-iflow-subjectreduction}). \qedhere
\end{description}
\end{description}
\end{proof}

We have seen in Remark~\ref{app-remiflow-valT} that when two expressions are related by $\fT^{\secmap{},\typenv}_{j,F,\low}$ and one of them is a value, then the other one is also a value.  From a semantic point of view, when an expression has reached a value it means that it has successfully completed its computation.  We will now see that when two expressions are related by $\fT^{\secmap{},\typenv}_{j,F,\low}$ and one of them is unable to \emph{resolve} into a value, in any sequence of unrelated computation steps, then the other one is also unable to do so.
We shall use the notion of \emph{derivative of an expression} $M$: %
\begin{defi}[Derivative of an Expression]
Given an expression $M$, we say that $M'$ is a $(W,\secmap{},\Upsilon,\typenv,j)$-derivative of~$M$ if and only if
\begin{itemize}
\item $M'=M$, or 
\item there exist $m$ such that $\Upsilon(m)=j$, $F$, $d$, $P$, two states $\confd {T_1}{S_1}$ and $\confd {T_1'}{S_1'}$ such that $S_1,S_1'$ are $(\secmap{},\typenv)$-compatible,
  and a derivative $M''$ of $M$ such that:
\[
W \vdash \iconft{T_1}{\{M''^m\}}{S_1} \xarr{F}{d} \iconft{T_1'}{\{M'^m\} \cup P}{S_1'}
\]
\end{itemize}
\end{defi}

\begin{defi}[Non-resolvable Expressions]
 An expression $M$ is $(W,\secmap{},\typenv)$-non-resolvable, denoted $M \dagger^{W,\secmap{},\Upsilon,\typenv}_j$, if there is no $(W,\secmap{},\Upsilon,\typenv,j)$-derivative $M'$ of $M$ such that $M' \in \Val$.
\end{defi}
\begin{lem}\label{app-prop-iflow-nonresolvableT}
If $M ~\fT^{\secmap{},\typenv}_{j,F,\low}~ N$ for some $F$, $\low$ and $j$, then $M \dagger^{W,\secmap{},\Upsilon,\typenv}_j$ implies $N \dagger^{W,\secmap{},\Upsilon,\typenv}_j$.
\end{lem}
\begin{proof}
{
Let us suppose that $\lnot N\dagger^{W,\secmap{},\Upsilon,\typenv}_j$.  That means that there exists a finite number of states $\confd{T_1}{S_1},$ $\ldots,$ $\confd{T_n}{S_n}$, and $\confd{T_1'}{S_1'}$, \ldots, $\confd{T_n'}{S_n'}$, of expressions $N_1$, \ldots, $N_n$ and of thread names $m_1, \ldots, m_n$ with $\Upsilon(m_1)= \ldots =\Upsilon(m_n) = j$ and
\[
\begin{array}{rcl}
W \vdash \iconft{T_1}{\{N^{m_1}\}}{S_1} &\xarr{}{}& \iconft{T_1'}{\{N_1^{m_1}\} \cup P_1}{S_1'} ~\text{and}\\
W \vdash \iconft{T_2}{\{N_1^{m_2}\}}{S_2} &\xarr{}{}& \iconft{T_2'}{\{N_2^{m_2}\} \cup P_2}{S_2'} ~\text{and}\\
&\vdots\\
W \vdash \iconft{T_n}{\{N_{n}^{m_n}\}}{S_n} &\xarr{}{}& \iconft{T_n'}{\{N_n^{m_n}\} \cup P_n}{S_n'}
\end{array}
\]
and $S_1,\ldots,S_n,S_1',\ldots,S_n'$ are $(\secmap{},\typenv)$-compatible, and such that $N_n \in \Val$.  By Strong Bisimulation for Low-Termination (Proposition~\ref{app-prop-iflow-strong}), we have that there exists a finite number of states $\confd{\bar T_1'}{\bar S_1'}$, \ldots, $\confd{\bar T_n'}{\bar S_n'}$, of expressions $M_1$, \ldots, $M_n$, and of pools of threads $P_1, \ldots, P_n$ such that
\[
\begin{array}{rcl}
W \vdash \iconft{T_1}{\{M^{m_1}\}}{S_1} &\xarr{}{}& \iconft{\bar T_1'}{\{M_1^{m_1}\} \cup \bar P_1}{\bar S_1'} ~\text{and}\\
W \vdash \iconft{T_2}{\{M_1^{m_2}\}}{S_2} &\xarr{}{}& \iconft{\bar T_2'}{\{M_2^{m_2}\} \cup \bar P_2}{\bar S_2'} ~\text{and}\\
&\vdots\\
W \vdash \iconft{T_n}{\{M_{n-1}^{m_n}\}}{S_n} &\xarr{}{}& \iconft{\bar T_n'}{\{M_n^{m_n}\} \cup \bar P_n}{\bar S_n'}
\end{array}
\]
such that ${\bar S_1'},\ldots,{\bar S_n'}$ are $(\secmap{},\typenv)$-compatible and:
\[
M_1 ~\fT^{\secmap{},\typenv}_{j,F,\low}~ N_1, ~\textit{and}~ \ldots, ~\textit{and}~ M_n ~\fT^{\secmap{},\typenv}_{j,F,\low}~ N_n
\]
By Remark~\ref{app-remiflow-valT}, we then have that $M_n \in \Val$.  Since $M_n$ is a $(W,\secmap{},\Upsilon,\typenv,j)$-derivative of $M$, we conclude that $\lnot M\dagger^{W,\secmap{},\typenv}_{j}$.
}
\end{proof}

The following lemma deduces \mentionind{operationally high thread}{operational ``highness''} of threads from that of its sub-expressions.
\begin{lem}[Composition of High Expressions]  Suppose that $M$ is typable with respect to $\secmap{}$, $\typenv$, $j$ and $F$.  Then, for all allowed-policy mappings $W$: \label{app-prop-iflow-composehigh}
\begin{enumerate}
\item  If $M=\app{M_1} {M_2}$ and either
  \begin{itemize} 
  \item $M_1\dagger^{W,\secmap{},\Upsilon,\typenv}_{j}$ and ${M_1}^m$ $\in \semhigh^{W,\secmap{},\Upsilon,\typenv}_{F,\ell}$, or
  \item ${M_1}^m, {M_2}^m \in \semhigh^{W,\secmap{},\Upsilon,\typenv}_{F,\ell}$ and $M_1$ is a syntactically $(\secmap{},\typenv,j,F,l)$-high function,
  \end{itemize}
then $M^m \in \semhigh^{W,\secmap{},\Upsilon,\typenv}_{F,\ell}$.
\item If $M=\cond {M_1} {M_t} {M_f}$ and ${M_1}^m,{M_t}^m,{M_f}^m\!\in\!\semhigh^{W,\secmap{},\Upsilon,\typenv}_{F,\low}$, then $M^m\!\in\!\semhigh^{W,\secmap{},\Upsilon,\typenv}_{F,\low}$.
\item If $M=\rfr{l,\theta} {M_1}$ and $l \not\Fpreceq{F} \low$ and ${M_1}^m\in\semhigh^{W,\secmap{},\Upsilon,\typenv}_{F,\low}$, then $M^m\in\semhigh^{W,\secmap{},\Upsilon,\typenv}_{F,\low}$.
\item If $M=\seq {M_1} {M_2}$ and either
  \begin{itemize} 
  \item $M_1\dagger^{W,\secmap{},\Upsilon,\typenv}_{j}$ and ${M_1}^m$ $\in\semhigh^{W,\secmap{},\Upsilon,\typenv}_{F,\low}$, or
  \item ${M_1}^m, {M_2}^m \in\semhigh^{W,\secmap{},\Upsilon,\typenv}_{F,\low}$,
  \end{itemize}
then $M^m\in\semhigh^{W,\secmap{},\Upsilon,\typenv}_{F,\low}$.
\item If $M=\assign {M_1} {M_2}$ and $M_1$ has type $\rfrt {\theta} {l}$ with $l \not\Fpreceq{F}\low$ and either
  \begin{itemize} 
  \item $M_1\dagger^{W,\secmap{},\Upsilon,\typenv}_{j}$ and ${M_1}^m$ $\in\semhigh^{W,\secmap{},\Upsilon,\typenv}_{F,\low}$, or
  \item ${M_1}^m, {M_2}^m \in\semhigh^{W,\secmap{},\Upsilon,\typenv}_{F,\low}$,
  \end{itemize}
then $M^m\in\semhigh^{W,\secmap{},\Upsilon,\typenv}_{F,\low}$.
\item If $M=\flow {F'} {M_1}$ and ${M_1}^m \in\semhigh^{W,\secmap{},\Upsilon,\typenv}_{F,\low}$, then ${M}^m\in\semhigh^{W,\secmap{},\Upsilon,\typenv}_{F,\low}$.
\item If $M=\allowed {F} {M_t} {M_f}$ and %
${M_t}^m, {M_f}^m \in\semhigh^{W,\secmap{},\Upsilon,\typenv}_{F,\low}$, then $M^m\in\semhigh^{W,\secmap{},\Upsilon,\typenv}_{F,\low}$.
\end{enumerate}
\end{lem}
\begin{lem}\label{app-prop-iflow-highT}
If $M_1 ~\fT^{\secmap{},\typenv}_{\Upsilon(m),F,\low}~ M_2$ for some $\secmap{}$, $\Upsilon$, $\typenv$, $F$, $\low$ and $m$, then $M_1^m \in \semhigh^{W,\secmap{},\Upsilon,\typenv}_{F,\ell}$, implies $M_2^m \in \semhigh^{W,\secmap{},\Upsilon,\typenv}_{F,\ell}$.
\end{lem}

\begin{proof}  By induction on the definition of $M_1 ~\fT^{\secmap{},\typenv}_{j,F,\low}~ M_2$, using Lemma~\ref{app-prop-iflow-nonresolvableT}.

\begin{description}
\item [Clause 7] Here we have $M_1=\flow{F'}{\bar{M}_1}$ and $M_2=\flow{F'}{\bar{M}_2}$ with $\bar{M}_1~\fT^{\secmap{},\typenv}_{j,F\klmeet F',\low}$ $\bar{M}_2$.  Clearly we have that $\bar M_1 \in \semhigh^{W,\secmap{},\Upsilon,\typenv}_{F,\low}$, so by induction hypothesis also $\bar M_2 \in \semhigh^{W,\secmap{},\Upsilon,\typenv}_{F,\low}$.  Therefore, by Composition of High Expressions (Lemma~\ref{app-prop-iflow-composehigh}) we have that $M_2 \in \semhigh^{W,\secmap{},\Upsilon,\typenv}_{F,\low}$. \qedhere

\end{description}

\end{proof}

\paragraph{\emph{Behavior of Typable Low Expressions.}}

In this second phase of the proof, we consider the general case of threads that are typable with any termination level.  As in the previous subsection, we show that a typable expression behaves as a strong bisimulation, provided that it is operationally low.  For this purpose, we make use of the properties identified for the class of low-terminating expressions by allowing only these to be followed by low-writes.  Conversely, high-terminating expressions can only be followed by high-expressions (see Definitions~\ref{app-def-operationallyhigh} and \ref{app-def-synhigh}).

The following result shows that the behavior of typable high threads (i.e. those with a high security level) that are location sensitive (i.e. depend on their location) is operationally high.
\begin{lem}[Location Sensitive Typable High Threads -- Lemma~\ref{prop-iflow-potentially}] \label{app-prop-iflow-potentially}
For a given flow policy~$F$ and security level $\low$, consider a thread $M^m$ such that ${\typenv \byint{j}{\secmap{}}{F} M : s, \tau}$ and $M = \EC{\allowed {F'} {N_t}{N_f}}$ with $j \not\Fpreceq{F} \low$.  Then, for all allowed-policy mappings $W$ and thread labelings $\Upsilon$ such that $\Upsilon(m)=j$, we have that $M^m \in \semhigh^{W,\typenv}_{F,\low}$.
\end{lem}
\begin{proof} By induction on the structure of $\cE{E}$, using Update of Effects (Lemma~\ref{app-prop-iflow-updateeffects}) and High Expressions (Lemma~\ref{app-prop-iflow-highexpr}).
Consider that we have $M = \cC{E}{M_0}$, where $M_0={\allowed {F'} {N_t}{N_f}}$.
\begin{description}
\item [$\boldsymbol{\cC{E}{M_0} = M_0}$] Then, by~\allowtyp, we have $\typenv \byint{j}{\secmap{}}{F}$ $\allowed {F'} {N_t} {N_f} : s,{\tau}$ where $\typenv \byint{j}{\secmap{}}{F} N_t : {{s_t,{\tau}}}$,  $\typenv \byint{j}{\secmap{}}{F} N_f : {{s_f,{\tau}}}$ and ${j \Fpreceq{F} s_t.w, s_f.w}$.  This means $s_t.w, s_f.w \not\Fpreceq{F} \low$, so by High Expressions (Lemma~\ref{app-prop-iflow-highexpr}), then ${N_t}^m,{N_f}^m \in \semhigh^{W,\secmap{},\Upsilon,\typenv}_{F,\low}$.  By Composition of High Expressions (Lemma \ref{app-prop-iflow-composehigh}), $M^m \in \semhigh^{W,\secmap{},\Upsilon,\typenv}_{F,\low}$.
\item [$\boldsymbol{\cC{E}{M_0} = \app{\cC{E_1}{M_0}}{M_1}}$]  Then by rule $\apptyp$ we have that $\typenv \byint{j}{\secmap{}}{F} {\cC{E_1}{M_0}} : {\tef {s_1} {\tau_1 \xarr{F,j}{s_1'} \sigma_1}}$ and $\typenv \byint{j}{\secmap{}}{F} {M_1} : {{\tef{s_1''}{\tau_1}}}$ with %
$s_1.r \Fpreceq{F} s_1'.w$ and $s_1.t\Fpreceq{F} s_1''.w$.  By Update of Effects (Lemma~\ref{app-prop-iflow-updateeffects}) we have $j \preceq s_1.r$, which implies that $j \Fpreceq{F} s_1.r$ and $s_1.r \not\Fpreceq{F} \low$.  Therefore, $s_1'.w \not\Fpreceq{F} \low$, which means that ${\cC{E_1}{M_0}}$ is a syntactically $(\secmap{},\typenv,j,F,\low)$-high function, and $M_1$ is syntactically $(\secmap,\typenv,j,F,\low)$-high.  By High Expressions (Lemma \ref{app-prop-iflow-highexpr}) we have ${M_1}^m \in \semhigh^{W,\secmap{},\Upsilon,\typenv}_{F,\low}$.  By induction hypothesis ${\cC{E_1}{M_0}}^m \in \semhigh^{W,\secmap{},\Upsilon,\typenv}_{F,\low}$.  Then, by Lemma \ref{app-prop-iflow-composehigh}, $M^m \in \semhigh^{W,\secmap{},\Upsilon,\typenv}_{F,\low}$.
\item [$\boldsymbol{\cC{E}{M_0} = \app V {\cC{E_1}{M_0}}}$]  Then by $\apptyp$ we have $\typenv \byint{j}{\secmap{}}{F} V : {\tef {s_1} {\tau_1 \xarr{F,j}{s_1'} \sigma_1}}$ and $\typenv \byint{j}{\secmap{}}{F} {\cC{E_1}{M_0}} : {{\tef{s_1''}{\tau_1}}}$ with $s_1''.r\Fpreceq{F} s_1'.w$ and $s_1.t\Fpreceq{F} s_1''.w$.  By Update of Effects (Lemma~\ref{app-prop-iflow-updateeffects}) we have $j \preceq s''_1.r$, which implies that $j \Fpreceq{F} s''_1.r$ and $s_1''.r \not\Fpreceq{F} \low$.  Therefore, $s_1'.w \not\Fpreceq{F} \low$, which means that $V$ is a syntactically $(\secmap{},\typenv,j,F,\low)$-high function. %
By induction hypothesis ${\cC{E_1}{M_0}}^m \in \semhigh^{W,\secmap{},\Upsilon,\typenv}_{F,\low}$.  Then, by Composition of High Expressions (Lemma \ref{app-prop-iflow-composehigh}), $M^m \in \semhigh^{W,\secmap{},\Upsilon,\typenv}_{F,\low}$.
\item [$\boldsymbol{\cC{E}{M_0} = \cond {\cC{E_1}{M_0}} {M_t} {M_f}}$]  Then by $\condtyp$ we have that $\typenv \byint{j}{\secmap{}}{F} {\cC{E_1}{M_0}} : {\tef {s_1} {\bool}}$, and $\typenv \byint{j}{\secmap{}}{F} {M_t} : {{\tef{s_1'}{\tau_1}}}$ and $\typenv \byint{j}{\secmap{}}{F} {M_f} : {{\tef{s_1''}{\tau_1}}}$ with ${s}_1.r\Fpreceq{F} {s}_1'.w, {s}_1'.w$.  By Update of Effects (Lemma~\ref{app-prop-iflow-updateeffects}) we have $j \preceq s_1.r$, which implies that $j \Fpreceq{F} s_1.r$ and $s_1.r \not\Fpreceq{F} \low$.
Therefore, ${s}_1'.w, {s}_1'.w \not\Fpreceq{F} \low$, so by High Expressions (Lemma \ref{app-prop-iflow-highexpr}) we have ${M_t}^m,{M_t}^m \in \semhigh^{W,\secmap{},\Upsilon,\typenv}_{F,\low}$.  By induction hypothesis ${\cC{E_1}{M_0}}^m \in \semhigh^{W,\secmap{},\Upsilon,\typenv}_{F,\low}$.  Then, by Composition of High Expressions (Lemma \ref{app-prop-iflow-composehigh}), $M^m \in \semhigh^{W,\secmap{},\Upsilon,\typenv}_{F,\low}$.
\item [$\boldsymbol{\cC{E}{M_0} = \seq {\cC{E_1}{M_0}} {M_1}}$]  Then by $\seqtyp$ we have that $\typenv \byint{j}{\secmap{}}{F} {\cC{E_1}{M_0}} : {\tef {s_1} {\tau_1}}$ and $\typenv \byint{j}{\secmap{}}{F} {M_1} : {{\tef{s_1'}{\tau_1'}}}$ with ${s}_1.t\Fpreceq{F} {s}_1'.w$.  By Update of Effects (Lemma~\ref{app-prop-iflow-updateeffects}) we have $j \preceq s_1.t$, which implies that $j \Fpreceq{F} s_1.t$ and $s_1.t \not\Fpreceq{F} \low$.
Therefore, ${s}_1'.w \not\Fpreceq{F} \low$, and by High Expressions (Lemma \ref{app-prop-iflow-highexpr}) we have ${M_1}^m \in \semhigh^{W,\secmap{},\Upsilon,\typenv}_{F,\low}$.  By induction hypothesis ${\cC{E_1}{M_0}}^m \in \semhigh^{W,\secmap{},\Upsilon,\typenv}_{F,\low}$.  Then, by Composition of High Expressions (Lemma \ref{app-prop-iflow-composehigh}), $M^m \in \semhigh^{W,\secmap{},\Upsilon,\typenv}_{F,\low}$.
\item [$\boldsymbol{\cC{E}{M_0} = \rfr {l,\theta} {\cC{E_1}{M_0}}}$]  Then by $\reftyp$ we have that $\typenv \byint{j}{\secmap{}}{F} {\cC{E_1}{M_0}} : {\tef {s_1} {\theta}}$ with $s_1.t \Fpreceq{F} l$.  By Update of Effects (Lemma~\ref{app-prop-iflow-updateeffects}) we have $j \preceq s_1.t$, which implies that $j \Fpreceq{F} s_1.t$ and $s_1.t \not\Fpreceq{F} \low$.
Therefore, $l \not\Fpreceq{F} \low$, and by induction hypothesis ${\cC{E_1}{M_0}}^m \in \semhigh^{W,\secmap{},\Upsilon,\typenv}_{F,\low}$.  Then, by Composition of High Expressions (Lemma \ref{app-prop-iflow-composehigh}), $M^m \in \semhigh^{W,\secmap{},\Upsilon,\typenv}_{F,\low}$.
\item [$\boldsymbol{\cC{E}{M_0} = \deref {\cC{E_1}{M_0}}}$]  Easy, by induction hypothesis.
\item [$\boldsymbol{\cC{E}{M_0} = \assign {\cC{E_1}{M_0}} {M_1}}$]  Then by $\assigntyp$ we have $\typenv \byint{j}{\secmap{}}{F} {\cC{E_1}{M_0}} : {\tef {s_1} {\rfrt {\theta} {\bar l}}}$ and $\typenv \byint{j}{\secmap{}}{F} {M_1} : {{\tef{s_1'}{\tau_1}}}$ with ${s}_1.t\Fpreceq{F} {s}_1'.w$ and $s_1.r \Fpreceq{F} \bar l$.  By Update of Effects (Lemma~\ref{app-prop-iflow-updateeffects}) we have $j \preceq s_1.r$, which implies that $j \Fpreceq{F} s_1.r$ and $s_1.r \not\Fpreceq{F} \low$.
Therefore, $\bar l \not\Fpreceq{F} \low$ and ${s}_1'.w \not\Fpreceq{F} \low$.  Hence, by High Expressions (Lemma~\ref{app-prop-iflow-highexpr}) we have ${M_1}^m \in \semhigh^{W,\secmap{},\Upsilon,\typenv}_{F,\low}$.  By induction hypothesis ${\cC{E_1}{M_0}}^m \in \semhigh^{W,\secmap{},\Upsilon,\typenv}_{F,\low}$.  Then, by Composition of High Expressions (Lemma \ref{app-prop-iflow-composehigh}), $M^m \in \semhigh^{W,\secmap{},\Upsilon,\typenv}_{F,\low}$.
\item [$\boldsymbol{\cC{E}{M_0} = \assign V {\cC{E_1}{M_0}}}$]  Then by $\assigntyp$ we have $\typenv \byint{j}{\secmap{}}{F} V : {\tef {s_1} {\rfrt {\theta} {\bar l}}}$ and $\typenv \byint{j}{\secmap{}}{F} {\cC{E_1}{M_0}} : {{\tef{s_1'}{\tau_1}}}$ with $s_1'.r \Fpreceq{F} \bar l$.  By Update of Effects (Lemma~\ref{app-prop-iflow-updateeffects}) we have $j \preceq s'_1.r$, which implies that $j \Fpreceq{F} s'_1.r$ and $s'_1.r \not\Fpreceq{F} \low$.
Therefore, $\bar l \not\Fpreceq{F} \low$, and by induction hypothesis ${\cC{E_1}{M_0}}^m \in \semhigh^{W,\secmap{},\Upsilon,\typenv}_{F,\low}$.  Then, by Composition of High Expressions (Lemma \ref{app-prop-iflow-composehigh}), $M^m \in \semhigh^{W,\secmap{},\Upsilon,\typenv}_{F,\low}$.
\item [$\boldsymbol{\cC{E}{M_0} = \flow{F'}{\cC{E_1}{M_0}}}$]  Then by~$\flowtyp$ we have $\typenv \byint{j}{\secmap{}}{F \klmeet F'} {\cC{E_1}{M_0}} : {{\tef{s_1}{\tau_1}}}$.  By induction hypothesis ${\cC{E_1}{M_0}}^m \in \semhigh^{W,\secmap{},\Upsilon,\typenv}_{F \klmeet F',\low}$, which implies ${\cC{E_1}{M_0}}^m \in \semhigh^{W,\secmap{},\Upsilon,\typenv}_{F,\low}$.  Then, by Composition of High Expressions (Lemma \ref{app-prop-iflow-composehigh}), we conclude that $M^m \in \semhigh^{W,\secmap{},\Upsilon,\typenv}_{F,\low}$.
\end{description}

\end{proof}

\begin{figure}
\figline
\begin{defi}[$\fR^{W,\secmap{},\Upsilon,\typenv}_{j,F,\low}$]  \label{app-defR}
 We have that $M_1 ~\fR^{W,\secmap{},\Upsilon,\typenv}_{j,F,\low}~ M_2$ if ${\typenv \byint{j}{\secmap{}}{F} M_1 : s_1, \tau}$ and ${\typenv \byint{j}{\secmap{}}{F} M_2 : s_2, \tau}$ for some $\typenv$, $s_1$, $s_2$ and $\tau$ and one of the following holds:

\begin{description}
\item[Clause 1']${M_1}^m, {M_2}^m \in \semhigh^{W,\secmap{},\Upsilon,\typenv}_{F,\low}$, for all allowed-policy mappings $W$, or
\item[Clause 2']$M_1=M_2$, or
\item[Clause 3']$M_1=\cond{\bar{M}_1}{\bar N_t }{{\bar N}_f}$ and $M_2=\cond{\bar{M}_2}{\bar{N_t}}{\bar{N_f}}$ with $\bar{M}_1~\fR^{W,\secmap{},\Upsilon,\typenv}_{j,F,\low}~\bar{M}_2$, and ${\bar{N_t}}^m, {\bar{N_f}}^m \in \semhigh^{W,\secmap{},\Upsilon,\typenv}_{F,\low}$, or
\item[Clause 4']$M_1=\app{\bar{M}_1}{\bar{N}_1}$ and $M_2=\app{\bar{M}_2}{\bar{N}_2}$ with  $\bar{M}_1~\fR^{W,\secmap{},\Upsilon,\typenv}_{j,F,\low}~\bar{M}_2$, and ${\bar{N}_1}^m,{\bar{N}_2}^m \in \hbisim_{F,\low}$, and $\bar{M}_1,\bar{M}_2$ are syntactically $(\secmap{},\typenv,j,F,\low)$-high functions, or
\item[Clause 5']$M_1=\app{\bar{M}_1}{\bar{N}_1}$ and $M_2=\app{\bar{M}_2}{\bar{N}_2}$ with $\bar M_1~\fT^{\secmap{},\typenv}_{j,F,\low}~\bar M_2$, and $\bar{N}_1~\fR^{W,\secmap{},\Upsilon,\typenv}_{j,F,\low}~\bar{N}_2$, and $\bar{M}_1, \bar{M}_2$ are syntactically $(\secmap{},\typenv,j,F,\low)$-high functions, or
\item[Clause 6']$M_1=\seq{\bar{M}_1}{\bar{N}}$ and $M_2=\seq{\bar{M}_2}{\bar{N}}$ with $\bar{M}_1~\fR^{W,\secmap{},\Upsilon,\typenv}_{j,F,\low}~\bar{M}_2$, and ${\bar{N}}^m\in\semhigh^{W,\secmap{},\Upsilon,\typenv}_{F,\low}$, or
\item[Clause 7']$M_1=\seq{\bar{M}_1}{\bar{N}}$ and $M_2=\seq{\bar{M}_2}{\bar{N}}$ with $\bar{M}_1~\fT^{\secmap{},\typenv}_{j,F,\low}~\bar{M}_2$, or
\item[Clause 8']$M_1=\rfr{l,\theta} {\bar{M}_1}$ and $M_2=\rfr{l,\theta}{\bar{M}_2}$ with $\bar{M}_1~\fR^{W,\secmap{},\Upsilon,\typenv}_{j,F,\low}~\bar{M}_2$, and $l \not\Fpreceq{F} \low$, or
\item[Clause 9']$M_1=\deref {\bar{M}_1}$ and $M_2=\deref{\bar{M}_2}$ with $\bar{M}_1~\fR^{W,\secmap{},\Upsilon,\typenv}_{j,F,\low}~\bar{M}_2$, or
\item[Clause 10']$M_1=\assign{\bar{M}_1} {\bar{N}_1}$ and $M_2=\assign{\bar{M}_2}{\bar{N}_2}$ with $\bar{M}_1~\fR^{W,\secmap{},\Upsilon,\typenv}_{j,F,\low}~\bar{M}_2$, and ${\bar{N}_1}^m, {\bar{N}_2}^m \in \semhigh^{W,\secmap{},\Upsilon,\typenv}_{F,\low}$, and $\bar{M}_1, \bar{M}_2$ both have type $\rfrt {{\theta}}{l}$ for some ${\theta}$ and $l$ such that ${l} \not\Fpreceq{F} \low$, or
\item[Clause 11']$M_1=\assign{\bar{M}_1} {\bar{N}_1}$ and $M_2=\assign{\bar{M}_2}{\bar{N}_2}$ with $\bar{M}_1~\fT^{\secmap{},\typenv}_{j,F,\low}~\bar{M}_2$, and $\bar{N}_1~\fR^{W,\secmap{},\Upsilon,\typenv}_{j,F,\low}~\bar{N}_2$, and $\bar{M}_1, \bar{M}_2$ both have type $\rfrt {{\theta}}{l}$ for some ${\theta}$ and $l$ such that ${l} \not\Fpreceq{F} \low$, or
\item[Clause 12']$M_1=\flow{F'}{\bar{M}_1}$ and $M_2=\flow{F'}{\bar{M}_2}$ with $\bar{M}_1~\fR^{W,\secmap{},\Upsilon,\typenv}_{j,F\klmeet F',\low}$ $\bar{M}_2$. %
\end{description}
\end{defi}
\caption{The relation $\fR^{W,\secmap{},\Upsilon,\typenv}_{j,F,\low}$} \label{app-figfR}
\figline
\end{figure}

We now design a binary relation on expressions that uses $\fT^{\secmap{},\typenv}_{j,F,\low}$ to ensure that high-terminating expressions are always followed by operationally high ones.
The definition of $\fR^{W,\secmap{},\Upsilon,\typenv}_{j,F,\low}$ is given in Figure~\ref{app-figfR}.  Notice that it is a symmetric relation.  In order to ensure that expressions that are related by $\fR^{W,\secmap{},\Upsilon,\typenv}_{j,F,\low}$ perform the same changes to the low memory, its definition requires that the references that are created or written using (potentially) different values are high, and that the body of the functions that are applied are syntactically high.
\begin{rem}\label{app-prop-iflow-remTtoR}
If $M_1 ~\fT^{\secmap{},\typenv}_{j,F,\low}~ M_2$, then $M_1 ~\fR^{W,\secmap{},\Upsilon,\typenv}_{j,F,\low}~ M_2$.
\end{rem}
The above remark is used to prove the following lemma.
\begin{lem}\label{app-prop-iflow-highR}
If $M_1 ~\fR^{W,\secmap{},\Upsilon,\typenv}_{\Upsilon(m),F,\low}~ M_2$, for some $\secmap{}$, $\Upsilon$, $\typenv$, $F$, $\low$ and $m$, then $M_1^m\in\semhigh^{W,\secmap{},\Upsilon,\typenv}_{F,\low}$ implies $M_2^m\in\semhigh^{W,\secmap{},\Upsilon,\typenv}_{F,\low}$.
\end{lem}
\begin{proof}  By induction on the definition of $M_1 ~\fR^{W,\secmap{},\Upsilon,\typenv}_{\Upsilon(m),F,\low}~ M_2$, using Lemma~\ref{app-prop-iflow-highT}.
\end{proof}

We have seen in Splitting Computations (Lemma \ref{app-prop-iflow-split}) that two computations of the same expression can split only if the expression is about to read a reference that is given different values by the memories in which they compute.  In Lemma~\ref{app-prop-iflow-fork} we saw that the relation $\fT^{\secmap{},\typenv}_{j,F,\low}$ relates the possible outcomes of expressions that are typable with a low termination effect.  Finally, from the following lemma one can conclude that the above relation $\fR^{W,\secmap{},\Upsilon,\typenv}_{j,F,\low}$ relates the possible outcomes of typable expressions in general.
\begin{lem} \label{app-prop-iflow-fork}
  If %
  $\typenv\byint{j}{\secmap{}}{F} \EC{\deref {a}}:s,\tau$ with $l\not\preceq_{F\klmeet\secpolcon{\cE{E}}} \low$, 
then for any values $V_0,V_1 \in \Val$ such that $\typenv\byint{j}{\secmap{}}{F} V_i:\theta$ we have $\cC{E}{V_0}~\fR^{W,\secmap{},\Upsilon,\typenv}_{j,F,\low}~\cC{E}{V_1}$.
\end{lem}
\begin{proof} By induction on the structure of $\cE{E}$ using Replacement (Lemma~\ref{app-prop-iflow-repllemma}), Lemma~\ref{app-prop-iflow-forkt}, Lemma \ref{app-prop-iflow-highexpr}.
\begin{description}
\item [$\boldsymbol{\cC{E}{\deref {a}}= {\flow {F'}{\cC{E_1}{\deref {a}}}}}$]  By rule $\flowtyp$ we have $\typenv \byint{j}{\secmap{}}{F \klmeet F'} V : {\tef s {\tau}}$.  By induction hypothesis $\cC{E_1}{V_0}~\fT^{j}_{F\klmeet F',\low}~\cC{E_1}{V_1}$, so we conclude by Replacement (Lemma~\ref{app-prop-iflow-repllemma}) and Clause 12'.  \qedhere

\end{description}
\end{proof}

We now state a crucial result of the paper: the relation $\fT^{\secmap{},\typenv}_{j,F,\low}$ is a ``sort of'' \mentionind{bisimulation!for non-disclosure for networks}{strong bisimulation}.  
\begin{prop}[Strong Bisimulation for Typable Low Threads]  \label{app-prop-iflow-weak} \text{}\\
Consider a given allowed-policy mapping $W$, reference labeling $\secmap{}$, thread labeling $\Upsilon$, typing environment $\typenv$, flow policy $F$, security level $\low$, two expressions $M_1$ and $M_2$ and thread name $m$ such that $M_1^m \notin \semhigh^{\secmap{},\Upsilon,\typenv}_{F,\low}$.  If, for states $\confd {T_1}{S_1}$,$\confd {T_2}{S_2}$ with $S_1, S_2$ being $({\secmap{}},{\typenv})$-compatible we have that there exist $d,F',P_1',T_1',S_1'$ such that
\myexample{
M_1~\fR^{W,\secmap{},\Upsilon,\typenv}_{\Upsilon(m),F,\low}~M_2 ~\textit{and}~ %
W \vdash \iconft {T_1} {\{{M_1}^m\}} {S_1} \xarr{F'}{d} \iconft {T_1'}{P_1'}{S_1'} ~\textit{and}~ \confd {T_1}{S_1} \memeqF{\secmap{1},\Upsilon}{F \klmeet F'}{\low} \confd {T_2}{S_2}, 
}
with $(\dom{{S_1}'}-\dom{S_1})\cap\dom{S_2}=\emptyset$, %
then there exist $P_2'$, $T_2'$ and $S_2'$ such that 
\myexample{
W \vdash \iconft{T_2}{\{{M_2}^m\}}{S_2}\xarr{F'}{d} \iconft{T_2'}{P_2'}{S_2'} ~\textit{and}~ M_1'~\fR^{W,\secmap{},\Upsilon,\typenv}_{j,F,\low}~M_2' ~\textit{and}~ \confd {T_1'}{S_1'} \memeqF{\secmap{1},\Upsilon}{F \klmeet F'}{\low} \confd {T_2'}{S_2'}
}
Furthermore, if $P_1'=\{M_1'^m\}$ then $P_2'=\{M_2'^m\}$, if $P_1'=\{M_1'^m,N^n\}$ for some thread $N^n$ and $(\dom{{T_1}'}-$ $\dom{T_1})\cap\dom{T_2}=\emptyset$ %
then $P_2'=\{M_2'^m,N^n\}$, and $S_1', S_2'$ are still $({\secmap{}},{\typenv})$-compatible.
\end{prop}
\begin{proof} 

By case analysis on the clause by which $M_1~\fR^{W,\secmap{},\Upsilon,\typenv}_{j,F,\low}~M_2$, and by induction on the definition of $\fR^{W,\secmap{},\Upsilon,\typenv}_{j,F,\low}$.  In the following, we use Subject Reduction (Theorem \ref{app-prop-iflow-subjectreduction}) to guarantee typability (with the same type) for $j$, ${\low}$ and ${F}$, which is a requirement for being in the $\fR^{W,\secmap{},\Upsilon,\typenv}_{j,F,\low}$ relation.  We also use the Strong Bisimulation for Low Termination Lemma (Lemma~\ref{app-prop-iflow-strong}).  Lemma~\ref{app-prop-iflow-nonresolvableT} is used to prove the cases for Clauses~5', 7' and~11'. 
\begin{description}
\item [Clause 2']  Here $M_1 = M_2$.  
  By Guaranteed Transitions (Lemma \ref{app-prop-iflow-guarantee}), then:
  \begin{itemize}
  \item If~$P_1' = \{M_1'^m\}$, and $(\dom{{S_1}'}-\dom{S_1})\cap\dom{S_2}=\emptyset$, %
then there exist $M_2'$, $T_2'$ and $S_2'$ such that $W \semvdash \iconft{T_2}{\{M_2^m\}}{S_2} \xarr{F'}{d}\iconft{T_2'}{\{M_2'^m\}}{S_2'}$ with $\confd{T_1'}{S_1'} \memeqF{\secmap{1},\Upsilon}{F\klmeet F'}{\low}\confd{T_2'}{S_2'}$.  There are two cases to consider:
    \begin{description}
    \item [$\boldsymbol{M_2'=M_1'}$] Then we have $M_1'~\fT^{\secmap{},\typenv}_{j,F,\low}~M_2'$ and $S_1', S_2'$ still $({\secmap{}},{\typenv})$-compatible, by Clause 2' and Subject Reduction (Theorem \ref{app-prop-iflow-subjectreduction}).
    \item [$\boldsymbol{M_2' \neq M_1'}$] Then by Splitting Computations (Lemma \ref{app-prop-iflow-split}) we have two possibilities:\\
      (1)~ there exists $\cE{E} $ and $a$ such that $M_1'=\cC{E}{S_1(a)}$, $F'=\extrf{\cE{E}}$, $M_2'=\cC{E}{S_2(a)}$, $\confd{T_1'}{S_1'}=\confd{T_1}{S_1}$ and $\confd{T_2'}{S_2'}=\confd{T_2}{S_2}$.  Since $S_1(a)\ne S_2(a)$, we have that $\secmap{1}(a)\not\Fpreceq{F\klmeet F'}\low$.  Therefore, $M_1'~\fR^{W,\secmap{},\Upsilon,\typenv}_{j,F,\low}~M_2'$, by Lemma~\ref{app-prop-iflow-fork} above.\\
      (2)~ there exists $\cE{E}$ such that $M_1'=$ $\cC{E}{\allowed {\bar F}{N_t}{N_f}}$, $F'=\secpolcon{\cE{E}}$, and $T_1(m) \neq T_2(m)$ with $\confd{T_1'}{S_1'}=\confd{T_1}{S_1}$ and $\confd{T_2'}{S_2'}=\confd{T_2}{S_2}$.  Since $T_1(m)\ne T_2(m)$, we have $\Upsilon(m) = j \not\Fpreceq{F}\low$, and by %
Lemma~\ref{app-prop-iflow-potentially} we have $M_1 \in \semhigh^{W,\secmap{},\Upsilon,\typenv}_{F,\low}$, which contradicts our assumption.
      \end{description}

  \item If~$P_1' = \{M_1'^m,N^n\}$, for some expression $N$ and $(\dom{{T_1}'}-$ $\dom{T_1})\cap\dom{T_2}=\emptyset$, %
then for some $M$, $l$ and $d$ we have $M_1=\cC{E}{{\threadnat l M d }}$, and there exist $M_2'$, $T_2'$ and $S_2'$ such that we have $W \semvdash$ $\iconft{T_2}{\{M_2^m\}}{S_2} \xarr{F'}{d}\iconft{T_2'}{\{M_2'^m,N^n\}}{S_2'}$ and with $\confd{T_1'}{S_1'} \memeqF{\secmap{1},\Upsilon}{F\klmeet F'}{\low}\confd{T_2'}{S_2'}$ where $M_1'=M_2'$.  Then we have $M_1'~\fT^{\secmap{},\typenv}_{j,F,\low}~M_2'$ and $S_1', S_2'$ are still $({\secmap{}},{\typenv})$-compatible, by Clause 2' and Subject Reduction (Theorem \ref{app-prop-iflow-subjectreduction}).

  \end{itemize}
\item [Clause 12'] Here $M_1=\flow{F'}{\bar{M}_1}$ and $M_2=\flow{F'}{\bar{M}_2}$ with $\bar{M}_1~\fR^{j}_{F \klmeet F',\low}$ $\bar{M}_2$.  We can assume that ${\bar{M}_1}^m \notin \semhigh^{W,\secmap{},\Upsilon,\typenv}_{F\klmeet F',\low}$, since otherwise ${\bar{M}_1}^m \in \semhigh^{W,\secmap{},\Upsilon,\typenv}_{F,\low}$ and by Composition of High Expressions (Lemma \ref{app-prop-iflow-composehigh}) $M_1^m\in\semhigh^{W,\secmap{},\Upsilon,\typenv}_{F,\low}$.  There are then two possibilities:
\begin{itemize}
\item  If~$P_1' = \{M_1'^m \}$, and $(\dom{{S_1}'}-\dom{S_1})\cap\dom{S_2}=\emptyset$, %
then $W \semvdash \iconft {T_1} {\{{\bar{M}_1}^m\}} {S_1}$ $\xarr{F''}{d}$ $\iconft {T_1'}{\{\bar{M}_1'^m\} \cup P}{S_1'}$ with $F'=\bar{F}\klmeet F''$, for some $P$.  By induction hypothesis, we have that $W \semvdash \iconft {T_2} {\{{\bar{M}_2}^m\}} {S_2}$ $\xarr{F''}{d}$ $\iconft {T_2'}{\{\bar{M}_2'^m\}}{S_2'}$, and $M_1'~\fR^{j}_{F\klmeet \bar{F},\low}~M_2'$, and $S_1', S_2'$ still $({\secmap{}},{\typenv})$-compatible, and also $\confd {T_1'}{S_1'}$ $\memeqF{\secmap{1},\Upsilon}{F\klmeet \bar{F}}{\low}$ $\confd {T_2'}{S_2'}$.  Notice that $\confd {T_1'}{S_1'}\memeqF{\secmap{1},\Upsilon}{F}{\low} \confd {T_2'}{S_2'}$.  We use Subject Reduction (Theorem \ref{app-prop-iflow-subjectreduction}) and Clause 12' to conclude.
\item  If~$P_1' = \{M_1'^m, N^n\}~$ for some expression $N$ and $(\dom{{T_1}'}-$ $\dom{T_1})\cap\dom{T_2}=\emptyset$, %
then $W \semvdash \iconft {T_1} {\{{\bar{M}_1}^m\}} {S_1}$ $\xarr{F''}{d}$ $\iconft {T_1'}{\{\bar{M}_1'^m, N^n\} \cup P}{S_1'}$ with $F'=\bar{F}\klmeet F''$.  By induction hypothesis, we have that $W \semvdash \iconft {T_2} {\{{\bar{M}_2}^m\}} {S_2}$ $\xarr{F''}{d}$ $\iconft {T_2'}{\{\bar{M}_2'^m, N^n\}}{S_2'}$, and $M_1'~\fR^{j}_{F\klmeet \bar{F},\low}~M_2'$, and $S_1', S_2'$ still $({\secmap{}},{\typenv})$-compatible, and also $\confd {T_1'}{S_1'}$ $\memeqF{\secmap{1},\Upsilon}{F\klmeet \bar{F}}{\low}$ $\confd {T_2'}{S_2'}$.  Notice that $\confd {T_1'}{S_1'}\memeqF{\secmap{1},\Upsilon}{F}{\low} \confd {T_2'}{S_2'}$.  We use Subject Reduction (Theorem \ref{app-prop-iflow-subjectreduction}) and Clause 12' to conclude. \qedhere
\end{itemize}

\end{description}

\end{proof}

\paragraph{\emph{Behavior of Sets of Typable Threads.}}

To conclude the proof of the Soundness Theorem, it remains to exhibit an appropriate bisimulation on thread configurations.

\begin{defi}[$\fA^{W,\secmap{},\Upsilon,\typenv}_{\low}$] \label{app-defbis} 
Given an allowed-policy mapping $W$, a reference labeling $\secmap{}$, a thread labeling $\Upsilon$, a typing environment $\typenv$ and a security level $\low$, the relation $\fA^{W,\secmap{},\Upsilon,\typenv}_{\low}$ is inductively defined as follows: %
\[
\begin{array}{c}
a)~ 
\fra{M^m \in \semhigh^{W,\secmap{},\Upsilon,\typenv}_{\kltop,\low}}
{\{M^m\} ~\fA^{W,\secmap{},\Upsilon,\typenv}_{\low}~ \emptyset} \quad
b)~ \fra{M^m \in \semhigh^{W,\secmap{},\Upsilon,\typenv}_{\kltop,\low}}{\emptyset ~\fA^{W,\secmap{},\Upsilon,\typenv}_{\low}~ \{M^m\}} \quad 
c)~ \fra{ M_1 ~\fR^{W,\secmap{},\Upsilon,\typenv}_{\Upsilon(m),\kltop,\low}~ M_2}{\{{M_1}^m\} ~\fA^{W,\secmap{},\Upsilon,\typenv}_{\low}~ \{ {M_2}^m \}} \\[6mm]
d)~ \fra{ P_1 ~\fA^{W,\secmap{},\Upsilon,\typenv}_{\low}~ P_2 ~~~ Q_1 ~\fA^{W,\secmap{},\Upsilon,\typenv}_{\low}~ Q_2}{P_1 \cup Q_1 ~\fA^{W,\secmap{},\Upsilon,\typenv}_{\low}~ P_2 \cup Q_2} 
\end{array}
\]
\end{defi}

\begin{prop} \label{app-prop-iflow-bis}
Given an allowed-policy mapping $W$, a reference labeling $\secmap{}$, a thread labeling $\Upsilon$, a typing environment $\typenv$ and a security level $\low$, for all allowed-policy mappings $W$ the relation 
\[\fB^{W,\secmap{},\Upsilon,\typenv}_{\low}~ = \{(\confd{P_1}{T_1},\confd{P_2}{T_2})~|~P_1~\fA^{W,\secmap{},\Upsilon,\typenv}_{\low}~P_2 \textit{ and } T_1 \memeqF{\secmap{1},\Upsilon}{\kltop}{l} T_2\}\]
is a $(W,{\secmap{}},\Upsilon,{\typenv},\ell)$-bisimulation according to Definition~\ref{def-bisimdefNDN2}.
\end{prop}
\begin{proof}
 It is easy to see, by induction on the definition of $\fA^{W,\secmap{},\Upsilon,\typenv}_{\low}$, that the relation $\fB^{W,\secmap{},\Upsilon,\typenv}_{\low}$ is symmetric. 
We show, by induction on the definition of $\fA^{W,\secmap{},\Upsilon,\typenv}_{\low}$, that if $\confd{P_1}{T_1}~\fB^{W,\secmap{},\Upsilon,\typenv}_{\low}~\confd{P_2}{T_2}$ and if for any given $({\secmap{}},{\typenv})$-compatible stores $S_1, S_2$ such that $\confd{T_1}{S_1}\memeqF{\secmap{1},\Upsilon}{F}{\low}\confd{T_2}{S_2}$ we have $W \semvdash \iconft{T_1}{P_1}{S_1} \xarr{F}{d} \iconft{T_1'}{P_1'}{S_1'}$, and $(\dom{{S_1}'}-\dom{S_1})\cap\dom{S_2}=\emptyset$ and $(\dom{{T_1}'}-$ $\dom{T_1})\cap\dom{T_2}=\emptyset$, 
then there exist $T_2'$, $P_2'$ and $S_2'$ such that $W \semvdash \iconft{T_2}{P_2}{S_2} \rarr \iconft{T_2'}{P_2'}{S_2'}$ and $\confd{P_1'}{T_1'}~\fB^{W,\secmap{},\Upsilon,\typenv}_{\low}~\confd{P_2'}{T_2'}$ and $\confd{T_1'}{S_1'}\memeqF{\secmap{1},\Upsilon}{\kltop}{\low}\confd{T_2'}{S_2'}$. Furthermore, $S_1',S_2'$ are still $(\secmap{},\typenv)$-compatible.
\begin{description}
  \item [Rule $\boldsymbol{a)}$]  Then $P_1= \{M^m\}$, $P_2 =\emptyset$, and $M^m \in\semhigh^{W,\secmap{},\Upsilon,\typenv}_{\kltop,\low}$.  Therefore, $P_1' \subseteq \semhigh^{W,\secmap{},\Upsilon,\typenv}_{G,\low}$, and $\confd{T_1'}{S_1'}\memeqF{\secmap{1},\Upsilon}{\kltop}{\low}\confd{T_1}{S_1}$ and $S'$ is still $(\secmap{},\typenv)$-compatible.  We have that $W \semvdash \iconft{T_2}{P_2}{S_2} \rarr \iconft{T_2}{P_2}{S_2}$ and by transitivity $\confd{T_1'}{S_1'}\memeqF{\secmap{1},\Upsilon}{\kltop}{\low}\confd{T_2}{S_2}$.  By Rules a) and d), we have $P_1'  ~\fA^{W,\secmap{},\Upsilon,\typenv}_{\low}~ \emptyset$. Then, $\confd{P_1'}{T_1'}~\fB^{W,\secmap{},\Upsilon,\typenv}_{\low}~\confd{\emptyset}{T_2'}$.  Rule $\boldsymbol{b)}$ is analogous.
  \item [Rule $\boldsymbol{c)}$]  Then $P_1 = \{{M_1}^m\}$ and $P_2 = \{ {M_2}^m \}$, and we have $M_1 ~\fR^{W,\secmap{},\Upsilon,\typenv}_{\Upsilon(m),\kltop,\low}~ M_2$.  If $M_1^m\in \semhigh^{W,\secmap{},\Upsilon,\typenv}_{\kltop,\low}$, then by Rule a), we have that $P_1' ~\fA^{W,\secmap{},\Upsilon,\typenv}_{\low}~ \emptyset$ and $\confd{T_1'}{S_1'}\memeqF{\secmap{1},\Upsilon}{\kltop}{\low}\confd{T_2'}{S_2'}$.  Also, by Lemma~\ref{app-prop-iflow-highR}, we have that $M_2^m\in \semhigh^{W,\secmap{},\Upsilon,\typenv}_{\kltop,\low}$, so by Rule b) $\emptyset ~\fA^{W,\secmap{},\Upsilon,\typenv}_{\low}~ P_2$.  By Rule d), we have $P_1' ~\fA^{W,\secmap{},\Upsilon,\typenv}_{\low}~ P_2$.  Then, $\confd{P_1'}{T_1'}~\fB^{W,\secmap{},\Upsilon,\typenv}_{\low}~\confd{P_2'}{T_2'}$.

If $M_1^m\notin \semhigh^{W,\secmap{},\Upsilon,\typenv}_{\kltop,\low}$, there are two cases to be considered:
    \begin{description}
    \item [$\boldsymbol{P_1'=\{{M_1'}^m\}}$]  Then by Strong Bisimulation for Typable Low Threads (Proposition~\ref{app-prop-iflow-weak}) there exist $T_2'$, $M_2'$ and $S_2'$ such that $W \semvdash \iconft {T_2}{\{{M_2}^m\}}{S_2}$ $\xarr{F'}{d}$ $\iconft{T_2'}{\{{M_2'}^m\}}{S_2'}$ with $M_1'~\fR^{W,\secmap{},\Upsilon,\typenv}_{\Upsilon(m),\kltop,\low}~M_2'$, and $\confd {T_1'}{S_1'}\memeqF{\secmap{1},\Upsilon}{\kltop}{\low} \confd {T_2'}{S_2'}$ and $S_1', S_2'$ are still $({\secmap{}},{\typenv})$-compatible.  Then, by Rule c), $\{{M_1'}^m\}$ $\fA^{W,\secmap{},\Upsilon,\typenv}_{\low}$ $\{{M_2'}^m\}$.   Then, $\confd{\{M_1'^m\}}{T_1'}~\fB^{W,\secmap{},\Upsilon,\typenv}_{\low}~\confd{\{M_2'^m\}}{T_2'}$.

    \item [$\boldsymbol{{P_1'=\{{M_1'}^m,N^n\}}}$]  We proceed as in the previous case to conclude that there exists $M_2'^m$ such that $\{{M_1'}^m\} ~\fA^{W,\secmap{},\Upsilon,\typenv}_{\low}~ \{{M_2'}^m\}$.  By Subject Reduction (Theorem~\ref{app-prop-iflow-subjectreduction}), by Lemma~\ref{app-prop-iflow-weakstrengthlemma}, and by Clause 2' we have $N~\fR^{W,\secmap{},\Upsilon,\typenv}_{\Upsilon(n),\kltop,\low}~N$, and so by Rule c) we have $\{{N^n}\} ~\fA^{W,\secmap{},\Upsilon,\typenv}_{\low}~ \{{N}^n\}$.  Therefore, by Rule d), we have $\{{M_1'}^m,N^n\}$ $\fA^{W,\secmap{},\Upsilon,\typenv}_{\low}$ $\{{M_2'}^m, N^n\}$.   Then, $\confd{\{{M_1'}^m,N^n\}}{T_1'}$ $\fB^{W,\secmap{},\Upsilon,\typenv}_{\low}$ $\confd{\{{M_2'}^m,N^n\}}{T_2'}$.
    \end{description}
  \item [Rule $\boldsymbol{d)}$]  Then $P_1 = \bar{P}_1 \cup \bar{Q}_1$ and $P_2 = \bar{P}_2 \cup \bar{Q}_2$, with $\bar{P}_1~\fA^{W,\secmap{},\Upsilon,\typenv}_{\low}~\bar{P}_2$ and $\bar{Q}_1~\fA^{W,\secmap{},\Upsilon,\typenv}_{\low}~\bar{Q}_2$.  Suppose that $W \semvdash \iconft{T_1}{\bar{P}_1}{S_1} \xarr{F}{d} \iconft{T_1'}{\bar{P}_1'}{S_1'}$ -- the case where $\bar{Q}_1$ reduces is analogous.  By induction hypothesis, there exist $T_2'$, $\bar{P}_2'$ and $S_2'$ such that $W \semvdash \iconft {T_2}{\bar{P}_2}{S_2} \rarr \iconft{T_2'}{\bar{P}_2'}{S_2'}$ with $\bar{P}_1'~\fA^{W,\secmap{},\Upsilon,\typenv}_{\low}~\bar{P}_2'$, and $\confd {T_1'}{S_1'}\memeqF{\secmap{1},\Upsilon}{\kltop}{\low} \confd {T_2'}{S_2'}$ and $S_1', S_2'$ are still $({\secmap{}},{\typenv})$-compatible.  Then, we have $W \semvdash \iconft {T_2}{\bar{P}_2 \cup \bar{Q}_2}{S_2}$ $\rarr$ $\iconft{T_2'}{\bar{P}_2' \cup \bar{Q}_2}{S_2'}$, and by Rule d) we have $\bar{P}_1' \cup \bar{Q}_1~\fA^{W,\secmap{},\Upsilon,\typenv}_{\low}~\bar{P}_2' \cup \bar{Q}_2$.   Then, $\confd{\bar{P}_1' \cup \bar{Q}_1}{T_1'}~\fB^{W,\secmap{},\Upsilon,\typenv}_{\low}$ $\confd{\bar{P}_2' \cup \bar{Q}_2}{T_2'}$. \qedhere
\end{description}
\end{proof}

\hide{
We can now prove the main-result regarding Distributed Non-disclosure:
\begin{thm}[Soundness of Typing Non-disclosure for Networks -- Theorem~\ref{prop-iflow-soundness}] \text{} \label{app-prop-iflow-soundness}
\noindent Consider a pool of threads $P$, an allowed-policy mapping $W$, a reference labeling $\secmap{}$, a thread labeling $\Upsilon$ and a typing environment $\typenv$.  If for all $M^{m} \in P$ there exist $s$, and $\tau$ such that ${\typenv \byint{\Upsilon(m)}{\secmap{}}{\kltop} M : s, \tau}$, then $P$ satisfies the Distributed Non-disclosure policy, i.e. $P\in\SecDND(W,\secmap{},\Upsilon,\typenv)$.
\end{thm}
\begin{proof}
For all $M^m \in P$ and for all choices of security levels $\low$, by assumption and by Clause 2 of Definition~\ref{app-defR}, we have that $M ~\fR^{W,\secmap{},\Upsilon,\typenv}_{\Upsilon(m),\kltop,\low}~ M$.  By Rule c) of Definition~\ref{app-defbis} we then have  ${\{{M}^m\} ~\fB^{W,\secmap{},\Upsilon,\typenv}_{\low}~ \{ {M}^m \}}$.  Therefore, by Rule d) we have that ${P ~\fB^{W,\secmap{},\Upsilon,\typenv}_{\low}~ P}$, from which we conclude that for all position trackers $T_1,T_2$ such that $\dom{P}=\dom{T_1}=\dom{T_2}$ and $T_1 \memeqF{\secmap{1},\Upsilon}{\kltop}{l} T_2$ we have $\confd{P}{T_1} ~\fB^{W,\secmap{},\Upsilon,\typenv}_{\low}~ \confd{P}{T_2}$.  By Proposition~\ref{app-prop-iflow-bis} we conclude that $\confd{P}{T_1} ~\biseqt{W,\secmap{},\Upsilon}{\typenv}{\low}~ \confd{P}{T_2}$.
\end{proof}
}

\section{Proofs for `Controlling Declassification'}

\subsection{Formalization of Flow Policy Confinement} \label{subsec-formalizationFPC}

In~\cite{Alm09,AC13} Flow Policy Confinement is defined for networks, considering a distributed setting with code mobility, by means of a bisimulation on \emph{located threads}.  In this paper we define Flow Policy Confinement for pools of threads without fixing the initial position of threads, and is based on a bisimulation on thread configurations.  We prove that the two variations of the definition are equivalent.

\paragraph{\emph{Definition on located threads.}}

The property is defined co-inductively on sets of located threads, consisting of pairs $\confd d {M^m}$ that carry information about the location $d$ of a thread $M^m$. The location of each thread determines which allowed flow policy it should obey at that point, and is used to place a restriction on the flow policies that decorate the transitions:  at any step, they should comply to the allowed flow policy of the domain where the thread who performed it is located.
\begin{defi}[$(W,\secmap{},\typenv)$-Confined Located Threads] \label{app-def-operationallyconf1}
Given an allowed-policy mapping $W$, a set of $(W,\secmap{},\typenv$)-\emph{confined located threads} is a set $\semconfLT$ of located threads that satisfies, for all $\confd d {M^{m}} \in \semconfLT$ and states $\confd T S$ with $S$ being $(W,\secmap{},\typenv)$-compatible memory $S$:
\myexample{
\begin{array}{l}
\confd d {M^{m}} \in \semconfLT ~\textit{and}~ T(m)=d ~\textit{and}~ %
W \vdash \iconft {T} {\{{M}^m\}} {S}$ $\xarr{F'}{d}$ $\iconft {T'}{\{M'^m\}\cup P}{S'} ~\textit{implies} \\
W(T(m)) \klpreceq F ~\text{and} ~\confd {T'(m)} {M'^{m}} \in \semconfLT.
\end{array}
}
Furthermore, if $P=\{N^n\}$ then also $\confd {T'(n)} {N^{n}}\in\semconfLT$, and $S'$ is still $(W,\secmap{},\typenv)$-compatible.\\
The largest set of $(W,\secmap{},\typenv)$-confined located threads %
is denoted~$\semconfLT_W^{\secmap{},\typenv}$.
\end{defi}
For any $W$, ${\secmap{}}$ and $\typenv$, the set of located threads where threads are values is a set of $(W,\secmap{},\typenv)$-confined located threads.  Furthermore, the union of a family of sets of $(W,\secmap{},\typenv)$-confined located threads is a set of $(W,\secmap{},\typenv)$-confined located threads. Consequently, $\semconf^{W,\secmap{},\typenv}$ exists. %

The notion of confinement used in~\cite{Alm09,AC13} is implicitly defined for thread configurations (by means of located threads), instead of for pools of threads, such as in the present paper.  The following definition generalizes that notion to one that works for any initial position tracker, and is thus defined for pools of threads.
\begin{defi}[Flow Policy Confinement (on located threads)]  \label{app-def-propertyFPC1}
A pool of threads $P$ satisfies Flow Policy Confinement with respect to an allowed-policy mapping $W$, a reference labeling $\secmap{}$ and a typing environment $\typenv$, if for all domains $d \in \Dom$ and threads $M^m \in P$ we have that $\confd{d}{M^m} \in \semconfLT_W^{\secmap{},\typenv}$.  We then write $P \in \SecFPCLTplus(W,\secmap{},\typenv)$.
\end{defi}

\paragraph{\emph{Comparison.}}

Flow Policy Confinement, defined over thread configurations, is equivalent to when defined over located threads.
\begin{prop}[Proposition~\ref{prop-FPCcomparison}] \label{app-prop-FPCcomparison}
  $\SecFPC(W,\secmap{},\typenv) = \SecFPCLTplus(W,\secmap{},\typenv)$.
\end{prop}
\begin{proof}
We assume that $P \in \SecFPCLTplus(W,\secmap{},\typenv)$, and show that the set
\[\semconfTC = \{\confd{P}{T}~|~ \forall M^m \in P ~.~ \confd {T(m)} {M^m} \in \semconfLT_W^{\secmap{},\typenv} \textit{ and } \dom{P}=\dom{T}\}\]
satisfies $\semconfTC \subseteq \semconfTC_W^{\secmap{},\typenv}$.
Given any thread configuration $\confd P T \in \semconfTC$, and any thread $M^m \in P$, it is clear that $\confd {T(m)} {M^m} \in \semconfLT_W^{\secmap{},\typenv}$.  If 
\[W \semvdash \iconft{T}{P}{S} \xarr{F}{d} \iconft{T'}{P'}{S'} \]
then for some $P'',T''$ such that $P=\{M^m\} \cup P''$ and $T = \{m \mapsto d\} \cup T''$ we have that either:
\begin{itemize}
\item $W \semvdash \iconft{T}{\{M^m\}}{S} \xarr{F}{d} \iconft{T'}{\{M'^m\}}{S'}$ where $P' = {\{M'^m\}} \cup P''$ and $T' = \{m \mapsto d\} \cup T''$.  In this case, by hypothesis, $W(T(m)) \klpreceq F$ and also $\confd {T'(m)} {M'^{m}} \in \semconfLT_W^{\secmap{},\typenv}$.
\item $W \semvdash \iconft{T}{\{M^m\}}{S} \xarr{F}{d} \iconft{T'}{\{M'^m,N^n\}}{S'}$ where $P' = {\{M'^m,N^n\}} \cup P''$  and $T' = \{m \mapsto d, n \mapsto d'\} \cup T''$ for some $d'$.  In this case, by hypothesis, $W(T(m)) \klpreceq F$ and also $\confd {T'(m)} {M'^{m}},\confd {T'(n)} {N^{n}}\in\semconfLT_W^{\secmap{},\typenv}$.
\end{itemize}
In both cases, we can conclude that $\forall M^m \in P' ~.~ \confd {T'(m)} {M^m} \in \semconfLT_W^{\secmap{},\typenv}$ and $\dom{P'}=\dom{T'}$, so $\confd{P'}{T'}\in\semconfTC$.  We then have that $\semconfTC \subseteq \semconfTC_W^{\secmap{},\typenv}$, and also that for all $\confd P T$ such that $\dom{P}=\dom{T}$, we have that $\confd P T \in \semconfTC$.  Then, $P \in \SecFPC(W,\secmap{},\typenv)$.

We now assume that $P \in \SecFPC(W,\secmap{},\typenv)$, and show that the set
\[\semconfLT = \{\confd{T(m)}{M^m}~|~ \exists P,T ~.~ \confd P T \in \semconfTC_W^{\secmap{},\typenv}  \textit{ and } M^m \in P\}\]
satisfies $\semconfLT \subseteq \semconfLT_W^{\secmap{},\typenv}$.
Given any located thread $\confd d {M^{m}} \in \semconfLT$ and position tracker $T$ such that $T(m)=d$, and any $(W,\secmap{},\typenv)$-compatible memory $S$, it is clear that also $\confd{\{M^m\}}{\{m \mapsto d\}} \in \semconfTC_W^{\secmap{},\typenv}$.  Then, if:
\begin{itemize}
\item $W \semvdash \iconft{T}{\{M^m\}}{S} \xarr{F}{d} \iconft{T'}{\{M'^m\}}{S'}$, then by hypothesis, we have $W(d) \klpreceq F$ and $\confd{\{M'^m\}}{T'} \in \SecFPC(W,\secmap{},\typenv)$.  Therefore, $\confd{T'(m)}{M'^m} \in \semconfLT$.
\item $W \semvdash \iconft{T}{\{M^{m}\}}{S} \xarr{F}{d} \iconft{T'}{\{M'^{m},N^n\}}{S'}$, then by hypothesis, we have $W(d) \klpreceq F$ and $\confd{\{M'^{m},N^n\}}{T'} \in \SecFPC(W,\secmap{},\typenv)$.  Therefore, both $\confd{T'(m)}{M'^m},$ $\confd{T'(n)}{N'^n}  \in \semconfLT$.
\end{itemize}
We then have that $\semconfLT \subseteq \semconfLT_W^{\secmap{},\typenv}$, and also that for all $\confd d {M^m}$ such that $M^m \in P$, we have that $\confd d {M^m} \in \semconfLT$.  Then, $P \in \SecFPCLTplus(W,\secmap{},\typenv)$.
\end{proof}

\subsection{Type System}

\subsubsection{Subject reduction and safety}

In order to establish the soundness of the type system of Figure~\ref{fig-confinementI-typesystem} we need a Subject Reduction result, stating that types that are given to expressions are preserved by computation.  To prove it we follow the usual steps \cite{WF94}.
In the following, $\Pse$ is the set of pseudo-values, as defined in Figure~\ref{fig-syntaxexpressions}.

\begin{rem}\label{app-remconfine-typeffval}
\text{}
\begin{enumerate}
\item If $X\in\Pse$ and $W;\typenv \byund{A}{\secmap{}}{} X : {\tau}$, then for all flow policies $A'$, we have that $W;\typenv \byund{A'}{\secmap{}}{} X : {\tau}$.
\item For any flow policies $A$ and $A'$ such that $A' \klpreceq A$, we have that ${W;\typenv \byund{A}{\secmap{}}{} M : \tau}$ implies ${W;\typenv \byund{A'}{\secmap{}}{} M : \tau}$.
\end{enumerate}
\end{rem}

\begin{lem}\label{app-propconfinement-weakstrengthlemma}\text{} 
\begin{enumerate}
\item If $W;\typenv \byund{A}{\secmap{}}{} M : {{{\tau}}}$ and $x\notin\dom{\typenv}$ then $W;\typenv,x:\sigma\byund{A}{\secmap{}}{} M : {{{\tau}}}$.
\item If $W;\typenv,x:\sigma\byund{A}{\secmap{}}{} M : {{{\tau}}}$ and $x\notin\fv{M}$ then $W;\typenv \byund{A}{\secmap{}}{} M : {{{\tau}}}$.
\end{enumerate}
\end{lem}
\begin{proof} By induction on the inference of the type judgment.  
\end{proof}

\begin{lem}[Substitution]\label{app-propconfinement-subslemma}\text{}\\  
If $W;\typenv,x:\sigma\byund{A}{\secmap{}}{} M : {\tau}$ and $W;\typenv \byund{A'}{\secmap{}}{} {X} : \sigma$ then $W;\typenv \byund{A}{\secmap{}}{}  {\substi x X M} : \tau$.  %
\end{lem}
\begin{proof}  By induction on the inference of $W;\typenv,x:\sigma\byund{A}{\secmap{}}{} M : {\tau}$, and by case analysis on the last rule used in this typing proof, using the previous lemma.
\begin{description} 
\item [\niltyp]

Here ${\substi x X M} = M$, and since $x \notin \fv{M}$ then by Lemma~\ref{app-propconfinement-weakstrengthlemma} we have $W;\typenv \byund{j}{\secmap{}}{F} {M} : s, \tau$.

\item [\vartyp]

If $M=x$ then $\sigma=\tau$ and $\substi x X{M} = X$.   By Remark~\ref{app-remconfine-typeffval}, we have $W;\typenv \byund{A}{\secmap{}}{} {X} : \tau$.
If $M \neq x$ then $\substi x X{M} = M$, where $x \notin \fv{M}$.  Therefore, by Lemma~\ref{app-propconfinement-weakstrengthlemma},  we have $W;\typenv \byund{A}{\secmap{}}{} {M} : \tau$.

\item [\abstyp]

Here $M = \lam y {\bar{M}}$, and $W;\typenv,{x:\sigma},{y:\bar{\tau}} \byund{\bar A}{\secmap{}}{} {\bar{M}} : \bar{\sigma}$ where $\tau = \bar{\tau} \xarr{\bar{A}}{} \bar{\sigma}$.  We can assume that $y \notin \dom{W;\typenv,{x:\sigma}}$ (otherwise rename $y$).  Therefore $\substi x X{\lam y {\bar{M}}} = {\lam y {\substi x X{\bar{M}}}}$.  By assumption and Lemma~\ref{app-propconfinement-weakstrengthlemma} we can write $W;\typenv,{y:\bar{\tau}} \byund{A'}{\secmap{}}{} {X} : \sigma$.  By induction hypothesis, $W;\typenv,{y:\bar{\tau}} \byund{\bar A}{\secmap{}}{} \substi x X{\bar{M}} : \bar{\sigma}$.  Then, by $\abstyp$, $W;\typenv \byund{A}{\secmap{}}{} {\lam y {\substi x X{\bar{M}}}} : \tau$.

\item [\rectyp]

Here $M = \fix y {\bar{X}}$, and $W;\typenv,{x:\sigma},{y:\tau} \byund{A}{\secmap{}}{} {\bar{X}} : \tau$.  We can assume $y \notin \dom{W;\typenv,{x:\sigma}}$ (otherwise rename $y$).  Therefore $\substi x X{\fix y {\bar{X}}} = {\fix y {\substi x X{\bar{X}}}}$.  By assumption and Lemma~\ref{app-propconfinement-weakstrengthlemma} we have $W;\typenv,{y:\tau} \byund{A'}{\secmap{}}{} {X} : \sigma$.  By induction hypothesis, $W;\typenv,{y:\tau} \byund{A}{\secmap{}}{} {\substi x X{\bar{X}}} : \tau$.  Then, by $\rectyp$, $W;\typenv \byund{A}{\secmap{}}{} {\fix y {\substi x X{\bar{X}}}} : \tau$.

\item [\condtyp]

Here $M = \cond {\bar{M}} {{\bar N}_t} {{\bar N}_f}$ and we have $W;\typenv,{x:\sigma} \byund{A}{\secmap{}}{} {\bar{M}} : \bool$, and $W;\typenv,{x:\sigma} \byund{A}{\secmap{}}{} {N_t} : \tau_1$ and $W;\typenv,{x:\sigma} \byund{A}{\secmap{}}{} {N_f} : \tau_2$.  By induction hypothesis, $W;\typenv,{x:\sigma} \byund{A}{\secmap{}}{} {\substi x X{\bar{M}}} : \bool$, $W;\typenv,{x:\sigma} \byund{A}{\secmap{}}{} {\substi x X {N_t}} : \tau_1$ and $W;\typenv,{x:\sigma} \byund{A}{\secmap{}}{} {\substi x X {N_f}} : \tau_2$.  Therefore, we have $W;\typenv,{x:\sigma} \byund{A}{\secmap{}}{} {\cond {\substi x X{\bar{M}}} {\substi x X{N_t}} {\substi x X{N_f}}} : \tau$ by rule $\condtyp$.

\item [\migtyp]

Here $M = \threadnat l {\bar M} {d}$ and we have that $W;\typenv,{x:\sigma} \byund{W(d)}{\secmap{}}{} {\bar{M}} : \tau$, with $\tau = \unit$.  By induction hypothesis, then $W;\typenv,{x:\sigma} \byund{W(d)}{\secmap{}}{} \substi x X{\bar{M}} : \tau$.  Therefore, by rule $\migtyp$, $W;\typenv,{x:\sigma} \byund{A}{\secmap{}}{} \threadnat l {\substi x X{\bar{M}}} {d}: \tau$.

\item [\flowtyp]

Here $M = \flow {\bar F} {\bar{M}}$ and $W;\typenv,{x:\sigma} \byund{A}{\secmap{}}{} {\bar{M}} : \tau$, with  $A \klpreceq {\bar F}$.  By induction hypothesis, $W;\typenv,{x:\sigma} \byund{A}{\secmap{}}{} \substi{x}{X}{\bar{M}} : \tau$.  Then, by $\flowtyp$, $W;\typenv,{x:\sigma} \byund{A}{\secmap{}}{} {\flow {\bar{F}} {\substi{x}{X}{\bar{M}}}} : \tau$.  

\item [\allowtyp]

Here $M = \allowed {\bar F} {{\bar N}_t} {{\bar N}_f}$ and we have that $W;\typenv,{x:\sigma} \byund{A}{\secmap{}}{} {\bar{M}} : \bool$ and $W;\typenv,{x:\sigma} \byund{A \klmeet {\bar F}}{\secmap{}}{} {{\bar N}_t} : \tau_1$ and $W;\typenv,{x:\sigma} \byund{A}{\secmap{}}{} {{\bar N}_f} : \tau_2$.  By induction hypothesis, $W;\typenv,{x:\sigma} \byund{A \klmeet {\bar F}}{\secmap{}}{} {\substi x X {{\bar N}_t}} : \tau_1$ and $W;\typenv,{x:\sigma} \byund{A}{\secmap{}}{} {\substi x X {{\bar N}_f}} : \tau_2$.  Therefore, by rule $\condtyp$, we have that $W;\typenv,{x:\sigma} \byund{A}{\secmap{}}{} {\allowed {\bar F} {\substi x X{N_t}} {\substi x X{N_f}}} : s, \tau$.

\end{description}
The proofs for the cases $\loctyp$, $\boolttyp$ and $\boolftyp$ are analogous to the one for $\niltyp$, while the proofs for $\reftyp$, $\apptyp$, $\seqtyp$, $\dereftyp$ and $\assigntyp$ are analogous to the one for $\condtyp$.
\end{proof}

\begin{lem}[Replacement]\label{app-propconfinement-repllemma}\text{}\\ %
If $W;\typenv \byund{A}{\secmap{}}{} \EC{M} : \tau$ is a valid judgment, then the proof gives $M$ a typing $W;\typenv \byund{A \klmeet \extrf {\cE{E}}}{\secmap{}}{} {M} : \bar{\tau}$ for some $\bar{\tau}$.  
In this case, if $W;\typenv \byund{A \klmeet \extrf {\cE{E}}}{\secmap{}}{} {N} : \bar{\tau}$, then $W;\typenv \byund{A}{\secmap{}}{} \EC{N} : \tau$.
\end{lem}

\begin{proof}  By induction on the structure of $\cE{E}$.
\begin{description}
\item [$\boldsymbol{\cC{E}{M} = M}$]

This case is direct.
\item [$\boldsymbol{{\cC{E}{M} = \cond {\cC{\bar{E}}{M}} {{\bar N}_t} {{\bar N}_f}}}$]

By $\condtyp$, we have $W;\typenv \byund{A}{\secmap{}}{} \cC{\bar{E}}{M} : \bool$, and also $W;\typenv \byund{A}{\secmap{}}{} {{\bar N}_t} : \tau$ and $W;\typenv \byund{A}{\secmap{}}{} {{\bar N}_f} : \tau$.  By induction hypothesis, the proof gives $M$ a typing $W;\typenv \byund{A \klmeet \extrf{\cE{\bar{E}}}}{\secmap{}}{} {M} : \hat{\tau}$, for some $\hat{\tau}$.

Also by induction hypothesis, $W;\typenv \byund{A}{\secmap{}}{} \cC{\bar{E}}{N} : \bool$.  Again by $\condtyp$, we have $W;\typenv \byund{A}{\secmap{}}{} \cond {\cC{\bar{E}}{N}} {{\bar N}_t} {{\bar N}_f} : \tau$.

\item [$\boldsymbol{\cC{E}{M} = \flow {\bar F} {\cC{\bar{E}}{M}}}$]

By $\flowtyp$, we have $W;\typenv \byund{A}{\secmap{}}{} \cC{\bar{E}}{M} : \tau$ and $A \klpreceq \bar F$.  By induction hypothesis, the proof gives $M$ a typing $W;\typenv \byund{A \klmeet \extrf{\cE{\bar{E}}}}{\secmap{}}{} {M} : \hat{\tau}$, for some $\hat{\tau}$.

Also by induction hypothesis, $W;\typenv \byund{A}{\secmap{}}{}  \cC{\bar{E}}{N} : \tau$.
Then, again by $\flowtyp$, we have $W;\typenv \byund{A}{\secmap{}}{} \flow {\bar F}{\cC{\bar{E}}{N}} : \tau$.

\end{description}
The proofs for the cases $\cC{E}{M} = \extrf{\assign{\cC{E}{M}}N}$, $\cC{E}{M} = \extrf{\assign V {\cC{E}{M}}}$, $\cC{E}{M} = \extrf{\deref{\cC{E}{M}}}$, $\cC{E}{M} = \extrf{{\app {\cC{E}{M}} N}}$, $\cC{E}{M} = \extrf{{\app V {\cC{E}{M}}}}$, $\cC{E}{M} = \extrf{\seq{\cC{E}{M}}N}$ and $\cC{E}{M} = \extrf{\rf_{l,\theta}{\cC{E}{M}}}$,  are all analogous to the proof for the case $\cC{E}{M} =$ $\cond {\cC{\bar{E}}{M}} {N_t} {N_f}$.
\end{proof}

We check that the type of a thread and the compatibility of memories is preserved by reduction.
\begin{prop}[Subject Reduction -- Proposition~\ref{prop-subjectreduction-confinementI}]\label{app-prop-subjectreduction-confinementI}
Consider an allowed-policy mapping $W$ and a thread $M^{m}$ such that ${W;\typenv \byund{A}{\secmap{}}{} M : \tau}$, and suppose $W \vdash \iconft{T}{\{M^{m}\}} S$ $\xarr{\var {F}}{d}$ $\iconft{T'} {\{M'^{m}\} \cup P} {S'}$, for a memory $S$ that is $(W,{\secmap{}},{\typenv})$-compatible. Then, $W;\typenv\byund{A \klmeet W(T(m))}{\secmap{}}{} M':\tau$, and $S'$ is also $(W,{\secmap{}},{\typenv})$-compatible.
\end{prop}
\begin{proof}
Suppose that we have $M=\cC{\bar{E}}{\bar{M}}$ and that $W \vdash \iconft{T}{\{{\bar{M}}^m\}}{S}$ $\xarr{\bar{F}}{d}$ $\iconft{\bar T'}{\{\bar{M}'^m\} \cup P'}{\bar{S}'}$.  We start by observing that this implies $F = \bar{F} \klmeet \extrf{\cE{\bar{E}}}$, $M'=\cC{\bar{E}}{\bar{M}'}$, $P = P'$, ${\bar{T}'} = {T'}$ and ${\bar{S}'} = {S'}$.  
We can assume, without loss of generality, that $\bar{M}$ is the smallest in the sense that there is no $\cE{\hat{E}},\hat{M},\hat{N}$ such that $\cE{\hat{E}}\neq[]$ and $\cC{\hat{E}}{\hat{M}}=\bar{M}$ for which we can write $W \vdash \iconft{T}{\{\hat{M}^m\}}{S}$ $\xarr{\hat{F}}{d}$ $\iconft{T'}{\{\hat{M}'^m\} \cup P}{S'}$.  %

By Replacement (Lemma~\ref{app-propconfinement-repllemma}), we have $W;\typenv \byund{{A} \klmeet \extrf{\bar{\cE{E}}}}{\secmap{}}{} {{\bar{M}}} : \bar{\tau}$ in the proof of $W;\typenv \byund{A}{\secmap{}}{} {\cC{\bar{E}}{\bar{M}}} : \tau$, for some $\bar{\tau}$.  
We proceed by case analysis on the transition $W \vdash \iconft{T}{\{{\bar{M}}^m\}}{S}$ $\xarr{\bar{F}}{d}$ $\iconft{T'}{\{\bar{M}'^m\} \cup P}{S'}$, and prove that if $S' \neq S$ then:
\begin{itemize}
\item There exists $\bar \tau'$ such that $W;\typenv\byund{{A} \klmeet \extrf{\bar{\cE{E}}} \klmeet W(T(m))}{\secmap{}}{} {\bar M'}:{\bar \tau'}$ and $\bar \tau \klpreceqtyp \bar \tau'$.  Furthermore, for every reference $a \in \dom{S'}$ we have ${\typenv \byund{\kltop}{\secmap{}}{} S'(a) : \secmap{2}(a)}$.
\item If~$P=\{N^n\}$~for some expression $N$ and thread name $n$, then also %
$W;\typenv\byund{{W(T'(n))}}{\secmap{}}{} N:\unit$.  (Note that in this case $S=S'$.)
\end{itemize}

By case analysis on the structure of $\bar M$:
\begin{description}
\item [$\boldsymbol{\bar{M} = \app {\lam x {\hat{M}}} V}$]

Here we have $\bar{M}' = \substi x V {\hat{M}}$, $S=S'$ and $P=\emptyset$.  By rule $\apptyp$, there exist $\hat{\tau}$ and $\hat{\sigma}$ such that $W;\typenv \byund{{A \klmeet \extrf{\bar{\cE{E}}}}}{\secmap{}}{} {\lam x {\hat{M}}} : \hat{\tau} \xarr{A \klmeet \extrf{\bar{\cE{E}}}}{} \hat{\sigma}$ and $W;\typenv \byund{{A \klmeet \extrf{\bar{\cE{E}}}}}{\secmap{}}{} {V} : \hat{\tau}$ with $\hat \sigma = \bar \tau$.
By $\abstyp$, then $W;\typenv,x:\hat{\tau} \byund{{A \klmeet \extrf{\bar{\cE{E}}}}}{\secmap{}}{} {{\hat{M}}} : \hat{\sigma}$. %
Therefore, by Lemma~\ref{app-propconfinement-subslemma}, we get $W;\typenv \byund{{A \klmeet \extrf{\bar{\cE{E}}}}}{\secmap{}}{} {\substi x V {\hat{M}}} : {\bar \tau}$.  By Remark~\ref{app-remconfine-typeffval}, $W;\typenv \byund{{A \klmeet \extrf{\bar{\cE{E}}}} \klmeet {W(T(m))}}{\secmap{}}{} {\substi x V {\hat{M}}} : \bar{\tau}$.
\item [$\boldsymbol{\bar{M} = \fix x X}$]

Here we have $\bar{M}' = {\app {\substi x {\fix x X}} X}$, $S=S'$ and $P=\emptyset$.  By rule $\rectyp$, we have $W;\typenv,x:\bar{\tau} \byund{{A \klmeet \extrf{\bar{\cE{E}}}}}{\secmap{}}{} {X} : {\bar \tau}$. %
Therefore, by Lemma~\ref{app-propconfinement-subslemma}, we get $W;\typenv \byund{{A \klmeet \extrf{\bar{\cE{E}}}}}{\secmap{}}{} {\substi x {\fix x X} {X}} : {\bar \tau}$.  By Remark~\ref{app-remconfine-typeffval}, $W;\typenv \byund{{A \klmeet \extrf{\bar{\cE{E}}}} \klmeet {W(T(m))}}{\secmap{}}{} {\substi x {\fix x X} {X}} : \bar{\tau}$.
\item [$\boldsymbol{\bar{M} = \cond {\tt} {N_t}{N_f}}$]

Here we have $\bar{M}' = N_t$, $S=S'$ and $P=\emptyset$.  By $\condtyp$, we have that $W;\typenv \byund{{A \klmeet \extrf{\bar{\cE{E}}}}}{\secmap{}}{} {N_t} : \bar{\tau}$.  By Remark~\ref{app-remconfine-typeffval}, $W;\typenv \byund{{A \klmeet \extrf{\bar{\cE{E}}}} \klmeet {W(T(m))}}{\secmap{}}{} {N_t} : \bar{\tau}$.
\item [$\boldsymbol{\bar{M} = \rfr {l,\theta} V}$]

Here we have $\bar{M}' = a$, $lab = a:\rfrt \theta {l}$~for some reference name $a$, type $\theta$ and security level $l$, $S' = S \cup \{(a,V)\}$ and $P = \emptyset$.  
By \reftyp, $\bar \tau = \rfrt {\theta} {}$ and $W;\typenv \byund{{A \klmeet \extrf{\bar{\cE{E}}}}}{\secmap{}}{} V : {\theta}$, and by Remark~\ref{app-remconfine-typeffval} then $W;\typenv \byund{\kltop}{\secmap{}}{} V : {\theta}$.  By Lemma~\ref{app-propconfinement-weakstrengthlemma} we have $W;\typenv \byund{\kltop}{\secmap{}}{} S'(a) : {\theta}$ for every $a \in \dom{S'}$.
By $\loctyp$, we have $W;\typenv \byund{{A \klmeet \extrf{\bar{\cE{E}}}}}{\secmap{}}{} a : {\rfrt {\theta} {}}$, and $\bar \tau = {\rfrt {\theta} {}}$.
\item [$\boldsymbol{\bar{M} = \deref {a}}$]

Here we have $\bar{M}' = S(a)$, $S=S'$ and $P=\emptyset$.  
By $\dereftyp$, we have that $W;\typenv \byund{{A \klmeet \extrf{\bar{\cE{E}}}}}{\secmap{}}{} {a} : {\rfrt {\bar \tau} {}}$, and by $\loctyp$ we know that $\secmap{2}(a)=\bar \tau$.
By compatibility assumption, then $W;\typenv \byund{\kltop}{\secmap{}}{} S(a) :\secmap{2}(a)$, and by Remark~\ref{app-remconfine-typeffval}, then $W;\typenv \byund{A \klmeet \extrf{\bar{\cE{E}}} \klmeet {W(T(m))}}{\secmap{}}{} S(a) :\bar \tau$.
\item [$\boldsymbol{\bar{M} = \assign a V}$]

Here we have $\bar{M}' = \nil$, and $P = \emptyset$.  
By \assigntyp, $\bar \tau = \unit$, and $W;\typenv \byund{{A \klmeet \extrf{\bar{\cE{E}}}}}{\secmap{}}{} a : \rfrt {\theta} {}$ and $W;\typenv \byund{{A \klmeet \extrf{\bar{\cE{E}}}}}{\secmap{}}{} V : {\theta}$, for some $\theta$.  By \loctyp, $\theta = \secmap{2}(a)$ and by Remark~\ref{app-remconfine-typeffval} we have $W;\typenv \byund{\kltop}{\secmap{}}{} V : \secmap{2}(a)$.
By $\niltyp$, we have that $W;\typenv \byund{A \klmeet \extrf{\bar{\cE{E}}} \klmeet {W(T(m))}}{\secmap{}}{} \nil :\bar \tau$, with $\bar \tau = \unit$.
\item [$\boldsymbol{\bar{M} = \flow {F'} {V}}$]

Here we have $\bar{M}' = V$, $S=S'$ and $P=\emptyset$.  By rule $\flowtyp$, we have that $W;\typenv \byund{A \klmeet \extrf{\bar{\cE{E}}}}{\secmap{}}{} V : {\bar \tau}$
and by Remark~\ref{app-remconfine-typeffval}, we have $W;\typenv \byund{A \klmeet \extrf{\bar{\cE{E}}} \klmeet {W(T(m))}}{\secmap{}}{} V : {\bar{\tau}}$.
\item [$\boldsymbol{\bar{M} = \allowed {F'} {N_t}{N_f}}$ and $\boldsymbol{W(T(m)) \klpreceq F'}$]

Here we have $\bar{M}' = N_t$, $S=S'$ and $P=\emptyset$.  By $\allowtyp$, we have that $W;\typenv \byund{{A \klmeet \extrf{\bar{\cE{E}}}} \klmeet {F'}}{\secmap{}}{} {N_t} : \bar{\tau}$.  By Remark~\ref{app-remconfine-typeffval}, then $W;\typenv \byund{{A \klmeet \extrf{\bar{\cE{E}}}} \klmeet {W(T(m))}}{\secmap{}}{} {N_t} : \bar{\tau}$.
\item [$\boldsymbol{\bar{M} = \allowed {F'} {N_t}{N_f}}$ and $\boldsymbol{W(T(m)) \not\klpreceq F'}$]

Here we have $\bar{M}' = N_f$, $S=S'$ and $P=\emptyset$.  By $\allowtyp$, we have that $W;\typenv \byund{{A \klmeet \extrf{\bar{\cE{E}}}}}{\secmap{}}{} {N_f} : \bar{\tau}$.  By Remark~\ref{app-remconfine-typeffval}, then $W;\typenv \byund{{A \klmeet \extrf{\bar{\cE{E}}}} \klmeet {W(T(m))}}{\secmap{}}{} {N_f} : \bar{\tau}$.
\item [$\boldsymbol{\bar{M} = \threadnat k {N} {\bar d}}$]

Here we have $\bar{M}' = \nil$,  $P=\{\hat N^n\}$ for some thread name $n$, $S=S'$, and $T'(n)=\bar d$.  By $\niltyp$, we have $W;\typenv \byund{A}{\secmap{}}{} {\nil} : \unit$.  By Remark~\ref{app-remconfine-typeffval}, then $W;\typenv \byund{A \klmeet \extrf{\bar{\cE{E}}}\klmeet W(d)}{\secmap{}}{} {\nil} : \unit$.
Furthermore, by $\migtyp$, we have that $W;\typenv \byund{W(\bar{d})}{\secmap{}}{} {N} : \unit$ and $\bar \tau = \unit$.
\end{description}
The cases $\bar{M} = \cond {\ff} {N_t}{N_f}$ and $\bar{M} = \seq V {\hat{M}}$ are analogous to the one for $\bar{M} = \cond {\tt} {N_t}{N_f}$.

By Replacement (Lemma~\ref{app-propconfinement-repllemma}), we finally conclude that
$W;\typenv\byund{{A} \klmeet W(T(m))}{\secmap{}}{} {\cC{\bar{E}}{\bar{M}'}}:{\tau}$.
\end{proof}

\begin{prop}[Safety -- Proposition~\ref{prop-staticmigrcontrol-safety}] \label{app-prop-staticmigrcontrol-safety}
Given an allowed-policy mapping $W$, consider a closed thread $M^{m}$ such that ${W;\emptyset \byund{A}{\secmap{}}{} M : \tau}$.  Then, for any memory $S$ that is $(W,{\secmap{}},{\emptyset})$-compatible and position-tracker $T$, either the program $M$ is a value, or
\end{prop}
\begin{proof} By induction on the derivation of ${W;\typenv \byund{A}{\secmap{}}{} M : \tau}$.  If $M \in Val$, then case~\ref{case1'} holds.
\begin{description}
\item [$\flowtyp$] If $M = \flow {F} N$, then ${W;\typenv \byund{A}{\secmap{}}{} N : \tau}$.  By induction hypothesis, then one of the following cases holds for $N$:
  \begin{description}
  \item [Case~\ref{case1'}]  Then case~\ref{case2'} holds for $M$, with $M'=N$, $d=T(m)$ and $F'=F$.
  \item [Case~\ref{case2'}]  Suppose that there exists $F'$, $N'$, $P$, $S'$ and $T'$ such that $W \vdash \iconft{T}{\{N^{m}\}} S$ $\xarr{\var {F'}}{d}$ $\iconft{T'} {\{N'^{m}\} \cup P} {S'}$.  Then $W \vdash \iconft{T}{\{M^{m}\}} S$ $\xarr{\var {F' \cup F}}{d}$ $\iconft{T'} {\{{\flow F {N'}}^{m}\} \cup P} {S'}$, so case~\ref{case2} holds for $M$.
  \end{description}
\item [$\allowtyp$] If $M = \allowed {F} {N_t} {N_f}$, then we have that $W \vdash \iconft{T}{\{M^{m}\}} S$ $\xarr{\emptyset}{T(m)}$ $\iconft{T} {\{M'^{m}\}} {S}$, with $M'=N_t$ of $M'=N_f$ depending on whether $W(T(m)) \klpreceq F$ or $W(T(m)) \not\klpreceq F$ (respectively).  Therefore, case~\ref{case2'} holds for $M$.
\item [$\migtyp$] If $M = \threadnat l N d$, then $W \vdash \iconft{T}{\{M^{m}\}} S$ $\xarr{W}{T(m)}$ $\iconft{T} {\{{\nil}^m,N^{n}\}} {S}$, so case~\ref{case2} holds for~$M$. \qedhere
\end{description}
\end{proof}

\subsubsection{Soundness}

\begin{thm}[Soundness of Enforcement Mechanism~I -- Theorem~\ref{prop-confinementI-soundness}]  \label{app-prop-confinementI-soundness}
Consider an allowed-policy mapping $W$, reference labeling $\secmap{}$, typing environment $\typenv$, and a thread configuration $\confd P T$ such that for all $M^{m} \in P$ there exists $\tau$ such that ${W;\typenv \byund{W(T(m))}{\secmap{}}{} M : \tau}$.  Then $\confd P T$ is a $(W,\secmap{},\typenv)$-confined thread configuration.
\end{thm}
\begin{proof} 
Consider the following set:
\myexample{
C = \{ \confd{P}{T} ~|~ \forall M^m \in P, \exists \tau ~.~ {W;\typenv \byund{W(T(m))}{\secmap{}}{} M : \tau}\} %
}
We show that $C$ is a set of $(W,\secmap{},\typenv)$-confined thread configurations.  
If $\confd P T \in C$ and for a given $(W,\secmap{},\typenv)$-compatible store $S$ we have that there exist $P',T',S'$ such that $W \vdash \iconft{T}{P} S$ $\xarr{F}{d}$ $\iconft{T'} {P'} {S'}$, then, there is a thread $M^m$ such that $P = \{M^m\} \cup \bar P$ and $W \vdash \iconft{T}{\{M^{m}\}} S$ $\xarr{F}{d}$ $\iconft{T'} {\{M^{m}\} \cup \bar{P}'} {S'}$, with $P' = {\{M^{m}\} \cup \bar{P}'} \cup \bar P$ and $T(m)=d$.
By induction on the inference of ${W;\typenv \byund{W(T(m))}{\secmap{}}{} M : \tau}$, we prove that $W(d) \klpreceq F$, and ${W;\typenv \byund{W(T'(m))}{\secmap{}}{} M' : \tau}$.  Furthermore, if $\bar P =N^n$ for some expression $N$ and thread name $n$, then ${W;\typenv \byund{W(T'(n))}{\secmap{}}{} N : \unit}$.

Since typability of the threads that result from the transition step, as well as the $(W,\secmap{},\typenv)$-compatibility of the stores, is guaranteed by Subject Reduction (Proposition~\ref{prop-subjectreduction-confinementI}), we prove only the conditions regarding the compliance of the declared flow policies to the current domain's allowed flow policy.
Assuming that ${W;\typenv \byund{W(d)}{\secmap{}}{} M : \tau}$, %
and by case analysis on the last rule in the corresponding typing proof:
\begin{description}
\item [\flowtyp]

Here $M = \flow {\bar F} {\bar{M}}$, and we have ${W;\typenv \byund{W(d)}{\secmap{}}{} {\bar M} : \tau}$, with 
$W(d) \klpreceq {\bar F}$.  There are two cases to consider:
\begin{description}
\item [$\bar M$ can compute] Then $W \vdash \iconft{T}{\{{\bar M}\}} S$ $\xarr{{\bar F}'}{d}$ $\iconft{T'} {\{{{\bar M}'}\} \cup \bar P'} {S'}$, with $F = \bar F \klmeet \bar F'$.  By induction hypothesis, then $W(d) \klpreceq {\bar F}'$. %
  Since $W(d) \klpreceq {\bar F}$, then $W(d) \klpreceq F$.  %
\item [$\bar M \in \Val$] Then we have $F = \kltop$, so $W(T(m)) \klpreceq F$ holds vacuously. %
  
\end{description}

\item [\allowtyp]
Here $M = {\allowed {\bar A} {\bar{N_t}} {\bar{N_f}}}$, and we have %
$F = \kltop$.  Therefore $W(T(m)) \klpreceq F$ holds vacuously.

\item [\migtyp]

In this case $M = \threadnat l {\bar{M}} {\bar d}$, with ${W;\typenv \byund{W(\bar{d})}{\secmap{}}{} {\bar M} : \unit}$.  Then we have $F = \kltop$, so $W(T(m)) \klpreceq F$ holds vacuously.

\end{description}

The cases for \reftyp, \dereftyp, \assigntyp, \seqtyp, \apptyp~and \condtyp~are similar to \flowtyp, since for the cases where the sub-expressions are not all values, these sub-expressions require typing assumptions that have the same $A$ parameter as the concluding judgment, and for the cases where the expression is ready to reduce, the flow policy that decorates the transition is $\kltop$.  The cases for \rectyp~and \migtyp~are similar to \allowtyp, since the constructs are not evaluation contexts, and therefore the transition is decorated with the top flow policy $F = \kltop$.
\end{proof}

\subsection{Runtime Type Checking}

\subsubsection{Subject Reduction and weak safety}

\begin{prop}[Subject Reduction -- Proposition~\ref{prop-subjectreduction-confinementII}]\label{app-prop-subjectreduction-confinementII}
  Consider a thread $M^{m}$ such that ${\typenv \byund{A}{\secmap{}}{} M : \tau}$ and suppose that $W \vdash \iconft{T}{\{M^{m}\}} S$ $\xarr{\var {F}}{d}$ $\iconft{T'} {\{M'^{m}\} \cup P} {S'}$, for a memory $S$ that is $(W,{\secmap{}},{\typenv})$-compatible. Then, $\typenv\byund{A \klmeet W(T(m))}{\secmap{}}{} M':\tau$, and $S'$ is also $(W,{\secmap{}},{\typenv})$-compatible.
\end{prop}
\begin{proof}
The proof is the same as the one for Proposition~\ref{app-prop-subjectreduction-confinementI}, with exception for the treatment of the case where the step corresponds to the creation of a new thread:
\begin{description}
\item [$\boldsymbol{\bar{M} = \threadnat l {N} d}$]

Here we have $\bar{M}' = \nil$, $P=\{\hat N^n\}$ for some thread name $n$, $T'(n)=d$ and (by the migration condition) ${\typenv \byund{W(d)}{\secmap{}}{} {N} : {\unit}}$, where $\bar \tau = \unit$.  By $\niltyp$ we have that $\typenv \byund{A}{\secmap{}}{} {\nil} : \unit$. \qedhere
\end{description}
\end{proof}

\begin{prop}[Safety (weakened) -- Proposition~\ref{prop-migrcontrol-safety}] \label{app-prop-migrcontrol-safety}
  Consider a closed thread $M^{m}$ such that ${\emptyset \byund{A}{\secmap{}}{} M : \tau}$.  Then, for any allowed-policy mapping $W$, memory $S$ that is $(W,{\secmap{}},{\emptyset})$-compatible and position-tracker $T$, either:
  \begin{enumerate}
\item the program $M$ is a value, or \label{case1}
\item $W \vdash \iconft{T}{\{M^{m}\}} S$ $\xarr{\var {F'}}{T(m)}$ $\iconft{T'} {\{M'^{m}\} \cup P} {S'}$, for some $F'$, $M'$, $P$, $S'$ and $T'$, or \label{case2}
\item $M=\cC{{E}}{\threadnat {l} {N} {d}}$, for some ${E}$, $l$, $d$ such that ${\emptyset \byund{\klbot}{\secmap{}}{} {N} : {\unit}}$ but ${\emptyset \not\byund{W({d})}{\secmap{}}{} {N} : {\unit}}$.  \label{case3}
\end{enumerate}
\end{prop}
\begin{proof} By induction on the derivation of ${\typenv \byund{A}{\secmap{}}{} M : \tau}$.  If $M \in Val$, then case~\ref{case1} holds.
\begin{description}
\item [$\flowtyp$] If $M = \flow {F} N$, then ${\typenv \byund{A}{\secmap{}}{} N : \tau}$.  By induction hypothesis, then one of the following cases holds for $N$:
  \begin{description}
  \item [Case~\ref{case1}]  Then case~\ref{case2} holds for $M$, with $M'=N$, $d=T(m)$ and $F'=F$.
  \item [Case~\ref{case2}]  Suppose that there exists $F'$, $N'$, $P$, $S'$ and $T'$ such that $W \vdash \iconft{T}{\{N^{m}\}} S$ $\xarr{\var {F'}}{d}$ $\iconft{T'} {\{N'^{m}\} \cup P} {S'}$.  Then $W \vdash \iconft{T}{\{M^{m}\}} S$ $\xarr{\var {F' \cup F}}{d}$ $\iconft{T'} {\{{\flow F {N'}}^{m}\} \cup P} {S'}$, so case~\ref{case2} holds for $M$.
  \item [Case~\ref{case3}]  Suppose that there exists an evaluation context ${\bar E}$, a security level $l$ and a domain name $d$ such that $N=\cC{{E}}{\threadnat {l} {N'} {d}}$ and ${\typenv \byund{\klbot}{\secmap{}}{} {N'} : {\unit}}$ but ${\typenv \not\byund{W({d})}{\secmap{}}{} {N'} : {\unit}}$.  Then, case~\ref{case3} also holds for $M$ for the evaluation context $E={\flow F {\cC{\bar E}{}}}$.
  \end{description}
\item [$\allowtyp$] If $M = \allowed {F} {N_t} {N_f}$, then we have that $W \vdash \iconft{T}{\{M^{m}\}} S$ $\xarr{\emptyset}{T(m)}$ $\iconft{T} {\{M'^{m}\}} {S}$, with $M'=N_t$ of $M'=N_f$ depending on whether $W(T(m)) \klpreceq F$ or $W(T(m)) \klpreceq F$ (respectively).  Therefore, case~\ref{case2} holds for $M$.
\item [$\migtyp$] If $M = \threadnat l N d$, then $\typenv \byund{\klbot}{\secmap{}}{} {{N}} : {\unit}$, and either:
  \begin{description}
  \item [$\boldsymbol{{\typenv \byund{W({d})}{\secmap{}}{} {N} : {\unit}}}$]  Then, $W \vdash \iconft{T}{\{M^{m}\}} S$ $\xarr{\emptyset}{T(m)}$ $\iconft{T} {\{{\nil}^m,N^{n}\}} {S}$.  Therefore, case~\ref{case2} holds for $M$.
  \item [$\boldsymbol{{\typenv \not\byund{W({d})}{\secmap{}}{} {N} : {\unit}}}$]  Then, case~\ref{case3} holds for $M$ for the evaluation context $E={\cE{\bar E}}$.  \qedhere
  \end{description}
\end{description}
\end{proof}

\subsection{Declassification effect}

\subsubsection{Subject Reduction and weak safety}  \label{app-decleffect-subjreduction}

In order to prove Subject Reduction, we follow the usual steps \cite{WF94}.  In the following, $\Pse$ is the set of pseudo-values, as defined in Figure~\ref{fig-syntaxexpressions}.

\begin{rem}\label{app-rem-valTransf}
If ${\typenv \byinov{\secmap{}}{} M \trans N : s, \tau}$ then $M \in\Val$ iff %
$N\in\Val$ and $s=\kltop$. %
\end{rem}

\begin{rem}\label{app-rem-pseudovalues}
If ${\typenv \byinov{\secmap{}}{} M \trans N : s, \tau}$ and $M \in\Pse$ with $\rn{M} \subseteq \dom{\secmap{}}$, then $N\in\Pse$ and $s=\kltop$.  %
\end{rem}

\begin{lem}\label{app-prop-decleffect-confinement-weakstrengthlemma}\text{} 
\begin{enumerate}
\item If $\typenv \byinov{\secmap{}}{} M \hookrightarrow \hat{M}: s, {{{\tau}}}$ and $x\notin\dom{\typenv}$ then $W;\typenv,x:\sigma\byinov{\secmap{}}{} M \hookrightarrow \hat{M}: s, {{{\tau}}}$.
\item If $W;\typenv,x:\sigma\byinov{\secmap{}}{} M \hookrightarrow \hat{M}: s, {{{\tau}}}$ and $x\notin\fv{M}$ then $\typenv \byinov{\secmap{}}{} M \hookrightarrow \hat{M}: s, {{{\tau}}}$.
\end{enumerate}
\end{lem}
\begin{proof} By induction on the inference of the type judgment.  
\end{proof}

\begin{lem}[Substitution]\label{app-prop-decleffect-confinement-subslemma}\text{}\\  
If $W;\typenv,x:\sigma\byinov{\secmap{}}{} M \hookrightarrow N : s, {\tau}$ and $\typenv \byinov{\secmap{}}{} {{X_1}} \hookrightarrow X_2: \kltop, \sigma'$ with $\sigma \klpreceqtyp \sigma'$, then $\typenv \byinov{\secmap{}}{}  {\substi x {X_1} M} \hookrightarrow {\substi x {X_2} N} : s, \tau'$ with $\tau \klpreceq \tau'$.  %
\end{lem}
\begin{proof}  By induction on the inference of $W;\typenv,x:\sigma\byinov{\secmap{}}{} M \hookrightarrow N : s, {\tau}$, and by case analysis on the last rule used in this typing proof, using Lemma~\ref{app-prop-decleffect-confinement-weakstrengthlemma}.
\begin{description} 
\item [\niltypI]

Here ${\substi x {X_1} M} = M$, and ${\substi x {X_2} N} = N$, and since $x \notin \fv{M}$ then by Lemma~\ref{app-prop-decleffect-confinement-weakstrengthlemma} we have $\typenv \byinov{\secmap{}}{F} {M} \hookrightarrow N : s, \tau$.

\item [\vartypI]

If $M=x$ then $N=x$, $s=\kltop$, $\sigma=\tau$, $\substi x {X_1}{M} = {X_1}$, and $\substi x {X_2}{N} = {X_2}$.   %
It is then direct that $\typenv \byinov{\secmap{}}{} {{X_1}} \hookrightarrow {X_2}: s, \sigma'$, and we take $\tau'=\sigma'$.
If $M \neq x$ then $N \neq x$, $\substi x {X_1}{M} = M$ and $\substi x {X_1}{M} = M$ where $x \notin \fv{M}$.  Therefore, by Lemma~\ref{app-prop-decleffect-confinement-weakstrengthlemma},  we have $\typenv \byinov{\secmap{}}{} {M} \hookrightarrow N: s, \tau$.

\item [\abstypI]

Here $M = \lam y {\bar{M}}$, $N = \lam y {\bar{N}}$, $s = \kltop$, and $W;\typenv,{x:\sigma},{y:\bar{\tau}} \byinov{\secmap{}}{} {\bar{M}} \hookrightarrow {\bar{N}}: \bar s, \bar{\sigma}$ where $\tau = \bar{\tau} \xarr{}{\bar{s}} \bar{\sigma}$.  We can assume that $y \notin \dom{W;\typenv,{x:\sigma}}$ (otherwise rename $y$).  Then $\substi x {X_1}{\lam y {\bar{M}}} = {\lam y {\substi x {X_1}{\bar{M}}}}$ and $\substi x {X_2}{\lam y {\bar{N}}} = {\lam y {\substi x {X_2}{\bar{N}}}}$.  By assumption and Lemma~\ref{app-prop-decleffect-confinement-weakstrengthlemma} we can write $W;\typenv,{y:\bar{\tau}} \byinov{\secmap{}}{} {{X_1}} \hookrightarrow X_2: \kltop, \sigma'$.  By induction hypothesis, $W;\typenv,{y:\bar{\tau}} \byinov{\secmap{}}{} \substi x {X_1}{\bar{M}} \hookrightarrow \substi x {X_2}{\bar{N}} : \bar s, \bar{\sigma}'$ with $\bar{\sigma} \klpreceqtyp \bar{\sigma}'$.  By \abstypI, $\typenv \byinov{\secmap{}}{} {\lam y {\substi x {X_1}{\bar{M}}}} \hookrightarrow {\lam y {\substi x {X_2}{\bar{N}}}} : s, \bar{\tau} \xarr{}{\bar{s}} \bar{\sigma}'$, and we take $\tau' = \bar{\tau} \xarr{}{\bar{s}} \bar{\sigma}'$.

\item [\rectypI]

Here $M = \fix y {\bar{{X_1}}}$, $M = \fix y {\bar{{X_2}}}$, by Remark~\ref{app-rem-pseudovalues} we have $s=\kltop$, and $W;\typenv,{x:\sigma},{y:\tau}$ $\byinov{\secmap{}}{} {\bar{{X_1}}} \hookrightarrow {\bar{{X_2}}}: s, \tau$.  We can assume that $y \notin \dom{W;\typenv,{x:\sigma}}$ (otherwise rename $y$).  Then $\substi x {X_1}{\fix y {\bar{{X_1}}}} = {\fix y {\substi x {X_1}{\bar{{X_1}}}}}$ and $\substi x {X_2}{\fix y {\bar{{X_2}}}} = {\fix y {\substi x {X_2}{\bar{{X_2}}}}}$.  By assumption and Lemma~\ref{app-prop-decleffect-confinement-weakstrengthlemma} we have $W;\typenv,{y:\tau} \byinov{\secmap{}}{} {{X_1}} \hookrightarrow X_2: s, \sigma'$.  By induction hypothesis, $W;\typenv,{y:\tau} \byinov{\secmap{}}{} {\substi x {X_1}{\bar{{X_1}}}} \hookrightarrow {\substi x {X_2}{\bar{{X_2}}}} : \kltop, \tau'$, with $tau \klpreceqtyp \tau'$.  Then, by \rectypI, $\typenv \byinov{\secmap{}}{} {\fix y {\substi x {X_1}{\bar{{X_1}}}}}$ $\hookrightarrow {\fix y {\substi x {X_2}{\bar{{X_2}}}}}: \kltop, \tau'$.

\item [\reftypI]

Here $M = \rfr {\theta} {\bar M}$, $N = \rfr {\theta} {\bar N}$ and we have $\typenv \byinov{\secmap{}}{} \bar {M} \hookrightarrow \bar {N}: \bar s, {\theta'}$ where $\theta \klpreceqtyp \theta'$ and $\tau = \rfrt{{\theta}}{}$. 
By induction hypothesis, $\typenv \byinov{\secmap{}}{} {\substi{x}{{X_1}}{\bar{M}}} \hookrightarrow {\substi{x}{{X_2}}{\bar{N}}}: \bar s, \theta''$ with $\theta' \klpreceqtyp \theta''$.  Then, $\theta \klpreceqtyp \theta''$, so we conclude by \reftypI~that $\typenv \byinov{\secmap{}}{} {\rfr {\theta} {\substi{x}{{X_1}}{\bar{M}}}} \hookrightarrow {\rfr {\theta} {\substi{x}{{X_2}}{\bar{N}}}}: s, \tau$.

\item [\condtypI]

Here $M = \cond {\bar{M}} {{\bar M}_t} {{\bar M}_f}$, $N = \cond {\bar{N}} {{\bar N}_t} {{\bar N}_f}$ and we have that $W;\typenv,{x:\sigma} \byinov{\secmap{}}{} {\bar{M}} \hookrightarrow \bar N : \bar s, \bool$, $W;\typenv,{x:\sigma} \byinov{\secmap{}}{} \bar {M_t} \hookrightarrow \bar {N_t}: \bar{s_t}, \bar{\tau_t}$ and $W;\typenv,{x:\sigma} \byinov{\secmap{}}{} \bar{M_f} \hookrightarrow \bar{N_f}: \bar{s_f}, \bar{\tau_f}$ with $s = \bar s \klmeet \bar{s_t} \klmeet \bar{s_f}$, $\bar{\tau_t} \kleqtyp \bar{\tau_f}$ and $\tau = \bar{\tau_t} \klmeet \bar{\tau_f}$.  
By induction hypothesis, $W;\typenv,{x:\sigma} \byinov{\secmap{}}{} {\substi x {X_1}{\bar{M}}} \hookrightarrow {\substi x {X_2}{\bar{N}}} : \bar s, \bool$, $W;\typenv,{x:\sigma} \byinov{\secmap{}}{} {\substi x {X_1} {\bar{M_t}}} \hookrightarrow {\substi x {X_2} {\bar{N_t}}} : \bar{s_t}, \bar{\tau_t}'$ and $W;\typenv,{x:\sigma} \byinov{\secmap{}}{} {\substi x {X_1} {\bar{M_f}}} \hookrightarrow {\substi x {X_2} {\bar{N_f}}}: \bar{s_f}, \bar{\tau_f}$' with $\bar{\tau_t} \klpreceqtyp \bar{\tau_t}'$ and $\bar{\tau_f} \klpreceqtyp \bar{\tau_f}'$.  Still $\bar{\tau_t}' \kleqtyp \bar{\tau_f}'$, so by  \condtypI~we have $W;\typenv,{x:\sigma} \byinov{\secmap{}}{}$ $(\lkw{if} {\substi x {X_1}{\bar{M}}}$ $\mkw{then}$ ${\substi x {X_1}{\bar{M_t}}}\!\mkw{else} \!{\substi x {X_1}{\bar{M_f}}})$ $\hookrightarrow$ ${\cond {\substi x {X_2}{\bar{N}}} {\substi x {X_2}{\bar{N_t}}} {\substi x {X_2}{\bar{N_f}}}}: s, \bar{\tau_t}' \klmeet \bar{\tau_f}'$.  We then take $\tau' = \bar{\tau_t}' \klmeet \bar{\tau_f}'$.

\item [\apptypI]

Here $M = \app {\bar M_1} {\bar M_2}$ and we have that $\typenv \byinov{\secmap{}}{} \bar {M_1} \hookrightarrow \bar {N_1}: \bar{s_1}, \bar{\theta} \xarr{}{\bar{s_3}} \bar{\sigma}$ and $\typenv \byinov{\secmap{}}{} \bar {M_2} \hookrightarrow \bar {N_2}: \bar{s_2}, \bar{\theta}''$ where $s = \bar{s_1} \klmeet \bar{s_2} \klmeet \bar{s_3}$, $\bar{\theta} \klpreceqtyp \bar{\theta}''$, and $\tau = \bar{\sigma}$.
By induction hypothesis, $W;\typenv,{x:\sigma} \byinov{\secmap{}}{} {\substi x {X_1} {\bar{M_1}}} \hookrightarrow {\substi x {X_2} {\bar{N_1}}} : \bar{s_t}, \bar{\theta} \xarr{}{\bar{s_3}'} \bar{\sigma}'$ and $W;\typenv,{x:\sigma} \byinov{\secmap{}}{} {\substi x {X_2} {\bar{M_f}}} \hookrightarrow {\substi x {X_2} {\bar{N_2}}}: \bar{s_f}, \bar{\theta}'''$ with ${\bar{s_3}} \klpreceq {\bar{s_3}'}$, $\bar{\sigma} \klpreceqtyp \bar{\sigma}'$ and $\bar{\theta}'' \klpreceqtyp \bar{\theta}'''$.  It is still the case that $\bar \theta \klpreceqtyp \bar \theta'''$.  Therefore, by rule \apptypI~we have that $W;\typenv,{x:\sigma} \byinov{\secmap{}}{} {\app {\substi x {X_1}{\bar{M_1}}} {\substi x {X_1}{\bar{M_2}}}}$ $\hookrightarrow$ ${\cond {\substi x {X_2}{\bar{M_2}}} {\substi x {X_2}{\bar{N_t}}} {\substi x {X_2}{\bar{N_2}}}}: s, \bar \sigma'$, and we take $\tau' = \bar \sigma'$.

\item [\migtypI]

Here $M = \threadnat l {\bar M} {d}$, $N = \threadanot l {\bar N} {d} {\bar s}$, $s = \kltop$ and we have $W;\typenv,{x:\sigma} \byinov{\secmap{}}{} {\bar{M}} \hookrightarrow \bar{N}: \bar s, \bar{\tau}$ and $\tau = \unit$.  By induction hypothesis, then $\typenv \byinov{\secmap{}}{} {\substi x {X_1}{\bar{M}}} \hookrightarrow {\substi x {X_2}{\bar{N}}} : \bar s, \bar{\tau}$.  Therefore, by rule \migtypI, we have $\typenv \byinov{\secmap{}}{} {\threadnat l {\substi x {X_1}{\bar{M}}} {d}}$ $\hookrightarrow$ ${\threadanot l {\substi x {X_2}{\bar{N}}} {d} {\bar s}}: \kltop, \tau$.

\item [\flowtypI]

Here $M = \flow {\bar F} {\bar{M}}$, $N = \flow {\bar F} {\bar{N}}$ and $W;\typenv,{x:\sigma} \byinov{\secmap{}}{} {\bar{M}} \hookrightarrow {\bar{N}} : \bar s, \tau$ with $s = \bar s \klmeet \bar F$.  By induction hypothesis, $\typenv \byinov{\secmap{}}{} {\substi{x}{{X_1}}{\bar{M}}} \hookrightarrow {\substi{x}{{X_2}}{\bar{N}}}: \bar s, \tau'$ with $\tau \klpreceqtyp \tau'$.  By \flowtypI, $\typenv \byinov{\secmap{}}{} {\flow {\bar{F}} {\substi{x}{{X_1}}{\bar{M}}}}$ $\hookrightarrow {\flow {\bar{F}} {\substi{x}{{X_2}}{\bar{N}}}}: s, \tau'$.  

\item [\allowtypI] 

Here $M = \allowed {\bar F} {{\bar M}_t} {{\bar M}_f}$, and $N =$ $\allowed {\bar F} {{\bar N}_t} {{\bar N}_f}$ and $W;\typenv,{x:\sigma} \byinov{\secmap{}}{} {\bar{M}} \hookrightarrow {\bar{M}}: \bar s, \bool$ and also $W;\typenv,{x:\sigma} \byinov{\secmap{}}{} {{\bar M}_t} \hookrightarrow {{\bar N}_t}: {{\bar s}_t}, \bar{\tau_t}$ and $W;\typenv,{x:\sigma} \byinov{\secmap{}}{} {{\bar M}_f} \hookrightarrow {{\bar N}_f}: \bar{s_f}, \bar{\tau_f}$, with $s = {\bar s}_t \klpseudominus \bar F \klmeet \bar{s_f}$, $\bar{\tau_t} \kleqtyp \bar{\tau_f}$ and $\tau = \bar{\tau_t} \klmeet \bar{\tau_f}$.  
By induction hypothesis, $\typenv \byinov{\secmap{}}{} {{\substi x {X_1} {{\bar M}_t}}} \hookrightarrow {{\substi x {X_2} {{\bar N}_t}}}: {{\bar s}_t}, \bar{\tau_t'}$ and $\typenv \byinov{\secmap{}}{} {{\substi x {X_1} {{\bar M}_f}}} \hookrightarrow {{\substi x {X_2} {{\bar N}_f}}}: {{\bar s}_f}, \bar{\tau_f'}$ with $\bar{\tau_t} \klpreceqtyp \bar{\tau_t'}$ and $\bar{\tau_f} \klpreceqtyp \bar{\tau_f'}$.  Still, we have that $\bar{\tau_t'} \kleqtyp \bar{\tau_f'}$.  We then have, by rule \allowtypI, that $\typenv \byinov{\secmap{}}{} {\allowed {\bar F} {\substi x {X_1}{\bar{M_t}}} {\substi x {X_1}{\bar{M_f}}}} \hookrightarrow$ ${\allowed {\bar F} {\substi x {X_1}{\bar{N_t}}} {\substi x {X_1}{\bar{N_f}}}} : s, {{\bar s}_t}' \klpseudominus \bar F \klmeet {{\bar s}_f}'$.  We then take $\tau' = \bar{\tau_t} \klmeet \bar{\tau_f}$.

\end{description}
The proofs for the cases \loctypI, \boolttypI~and \boolftypI~are analogous to the one for \niltypI, while the proofs for \seqtypI, \dereftypI~and \assigntypI~are analogous to (or simpler than) the one for \apptypI.
\end{proof}

\begin{lem}[Replacement]\label{app-prop-decleffect-confinement-repllemma} \text{} %
\begin{enumerate}
\item If $\typenv \byinov{\secmap{}}{} \cC{E_1}{M}  \hookrightarrow {N_{E_2}} : s, \tau$ is a valid judgment, then the proof gives $M$ a typing $\typenv \byinov{\secmap{}}{} {M} \hookrightarrow N: \bar s, \bar{\tau}$ for some $N$, $\bar s$ and $\bar{\tau}$ such that $s \klpreceq \bar s$ and for which there exists $\cE{E_2}$ such that ${N_{E_2}} = \cC{E_1}{N}$.  
In this case, if $\typenv \byinov{\secmap{}}{} {M'} \hookrightarrow N': \bar s', \bar{\tau}'$ with $\bar s \klmeet A \klpreceq \bar s'$ for some $A$ and $\bar{\tau} \klpreceqtyp \bar{\tau}'$, then $\typenv \byinov{\secmap{}}{} {\cC{E_1}{M'}} \hookrightarrow \cC{E_2}{N'}: s', \tau'$ for some $s'$, $\tau'$ such that $s \klmeet A \klpreceq s'$ and $\tau \klpreceqtyp \tau'$.
\item If $\typenv \byinov{\secmap{}}{} {M_{E_1}} \hookrightarrow \cC{E_2}{N} : s, \tau$ is a valid judgment, then the proof gives $N$ a typing $\typenv \byinov{\secmap{}}{} {M} \hookrightarrow N: \bar s, \bar{\tau}$ for some $M$, $\bar s$ and $\bar{\tau}$ such that $s \klpreceq \bar s$ and for which there exists $\cE{E_1}$ such that ${M_{E_1}} = \cC{E_1}{M}$.  
In this case, if $\typenv \byinov{\secmap{}}{} {M'} \hookrightarrow N': \bar s', \bar{\tau}'$ with $\bar s \klmeet A \klpreceq \bar s'$ for some $A$ and $\bar{\tau} \klpreceqtyp \bar{\tau}'$, then $\typenv \byinov{\secmap{}}{} {\cC{E_1}{M'}} \hookrightarrow \cC{E_2}{N'}: s', \tau'$ for some $s'$, $\tau'$ such that $s \klmeet A \klpreceq s'$ and $\tau \klpreceqtyp \tau'$.
\end{enumerate}
\end{lem}

\begin{proof}  \text{}
  
\begin{enumerate} 

\item By induction on the structure of $\cE{E_1}$.  The proof is symmetric to the second case.
  
\item By induction on the structure of $\cE{E_2}$.
\begin{description}
\item [$\boldsymbol{\cC{E_2}{N} = N}$]

This case is direct.
\item [$\boldsymbol{{\cC{E_2}{N} = \cond {\cC{\hat{E}_2}{N}} {\hat N_t} {\hat N_f}}}$]

By \condtypI, ${M_{E_1}} = \cond {M_{\hat{E}_1}} {\hat M_t} {\hat M_f}$, and $\typenv \byinov{\secmap{}}{} {M_{\hat{E}_1}} \hookrightarrow {\cC{\hat{E}_2}{N}}: \hat s, \bool$, and also $\typenv \byinov{\secmap{}}{} {\hat M_t} \hookrightarrow {\hat N_t} : \hat{s}_t, {\hat{\tau}_t}$ and $\typenv \byinov{\secmap{}}{} {\hat M_f} \hookrightarrow {\hat N_f} : \hat{s}_f, \hat{\tau}_f$ with $s = \hat s \klmeet \hat{s}_t \klmeet \hat{s}_f$, $\hat{\tau}_t \kleqtyp \hat{\tau}_f$ and $\tau = {\hat{\tau}_t} \klmeet {\hat{\tau}_f}$.  By induction hypothesis, the proof gives $N$ a typing $\typenv \byinov{\secmap{}}{} {M} \hookrightarrow N: {\bar s}, {\bar \tau}$, for some ${M}$, ${\bar s}$, ${\bar \tau}$ such that $\hat s \klpreceq {\bar s}$, and for which there exists $\cE{\hat{E}_1}$ such that ${M_{\hat{E}_1}} = \cC{\hat{E}_1}{M}$.

Also by induction hypothesis, $\typenv \byinov{\secmap{}}{} {\cC{\hat{E_1}}{M'}} \hookrightarrow \cC{\hat{E}_2}{N'} : \hat s', \bool$, for some $\hat s'$ such that $\hat s \klmeet A \klpreceq \hat s'$.  Again by rule \condtypI we have that $\typenv \byinov{\secmap{}}{} \cond {\cC{\hat{E_1}}{M'}} {\hat M_t} {\hat M_f}$ $\hookrightarrow$ $\cond {\cC{\hat{E}_2}{N'}} {\hat N_t} {\hat N_f}  : s', \tau'$ with $s' = \hat s' \klmeet \hat{s}_t \klmeet \hat{s}_f$ and $\tau' = {\hat{\tau}_t} \klmeet {\hat{\tau}_f}$.  Notice that $s \klmeet A = \hat s \klmeet \hat{s}_t \klmeet \hat{s}_f \klmeet A \klpreceq \hat s' \klmeet \hat{s}_t \klmeet \hat{s}_f  = s'$.

\item [$\boldsymbol{\cC{E_2}{M} = \flow {\hat F} {\cC{\hat{E}_2}{M}}}$]

By \flowtypI, ${M_{E_1}} = \flow {\hat F} {M_{\hat{E}_1}}$, and $\typenv \byinov{\secmap{}}{} {M_{\hat{E}_1}} \hookrightarrow \cC{\hat{E}_2}{N} : \hat s, \tau$ and $s = \hat s \klmeet \hat F$.  By induction hypothesis, the proof gives $N$ a typing $\typenv \byinov{\secmap{}}{} {M} \hookrightarrow N : \bar s, \bar{\tau}$, for some $M$, $\bar s$, $\bar{\tau}$ such that $\hat s \klpreceq \bar s$, and for which there exists $\cE{\hat{E}_1}$ such that ${M_{\hat{E}_1}} = \cC{\hat{E}_1}{M}$.

Also by induction hypothesis, $\typenv \byinov{\secmap{}}{} {\cC{\hat{E_1}}{M'}} \hookrightarrow {\cC{\hat{E}_2}{N'}} : \hat s', \tau$, for some $\hat s'$ such that $\hat s \klmeet A \klpreceq \hat s'$.
Then, again by \flowtypI, we have $\typenv \byinov{\secmap{}}{} \flow {\hat F} {\cC{\hat{E_1}}{M'}} \hookrightarrow \flow {\bar F}{\cC{\hat{E}_2}{N'}}  : s', \tau$ with $s' = \hat s' \klmeet \hat F$.  Notice that $s \klmeet A = \hat s \klmeet \hat{F} \klmeet A \klpreceq \hat s' \klmeet \hat{F}  = s'$.

\end{description}
The proofs for the cases $\cC{E_2}{M} = \extrf{\assign{\cC{\hat{E}_2}{M}}N}$, $\cC{E_2}{M} = \extrf{\assign V {\cC{\hat{E}_2}{M}}}$, $\cC{E_2}{M} = \extrf{\deref{\cC{\hat{E}_2}{M}}}$, $\cC{E_2}{M} = \extrf{{\app {\cC{\hat{E}_2}{M}} N}}$, $\cC{E_2}{M} = \extrf{{\app V {\cC{\hat{E}_2}{M}}}}$, $\cC{E_2}{M} = \extrf{\seq{\cC{\hat{E}_2}{M}}N}$ and $\cC{E_2}{M} = \extrf{\rfrl{l}{\theta}{\cC{\hat{E}_2}{M}}}$,  are all analogous to the proof for the case where $\cC{E_2}{M} =$ $\cond {\cC{\hat{E}_2}{M}} {N_t} {N_f}$.  \qedhere
\end{enumerate}
\end{proof}

The following proposition ensures that the annotation processing is preserved by the annotated semantics.  This is formulated by stating that after reduction, programs are still well annotated.  %
More precisely, the following result states that if a program is the result of an annotation process, a certain declassification effect and type, then after one computation step it is still the result of annotating a program, and is given a not more permissive declassification effect and type.

\begin{prop}[Subject Reduction, or Preservation of Annotations -- Proposition~\ref{prop-preservanot}] \label{app-prop-preservanot} 
Consider a thread $M^{m}$ such that ${\typenv \byinov{\secmap{}}{} {M} \hookrightarrow {N}: s, {\tau}}$ and suppose that
$W \semvdash \iconft{T}{\{{N}^{m}\}} S$ $\xarr{\var {F}}{d}$ $\iconft{T'} {\{{N'}^{m}\} \cup P} {S'}$, for a memory $S$ that is $(W,{\secmap{}},{\typenv})$-compatible. Then there exist $M'$, $s'$, $\tau'$ such that $s \klmeet {W(T(m))} \klpreceq s'$, and $\tau \klpreceqtyp \tau'$, and $\typenv\byinov{\secmap{}}{} M' \hookrightarrow {N'}:s',\tau'$, and $S'$ is also $(W,{\secmap{}},{\typenv})$-compatible.
Furthermore, if $P = \{N''^n\}$ for some expression $N''$ and thread name $n$, then there exist $M''$, $s''$ such that $W(T'(n)) \klpreceq s''$ and $\typenv\byinov{\secmap{}}{} M'' \hookrightarrow N'':s'',\unit$.
\end{prop} 
\begin{proof}
Suppose that $N=\cC{\bar{E}}{\bar{N}}$ and $W \semvdash \iconft{T}{\{{\bar{N}}^m\}}{S}$ $\xarr{\bar{F}}{d}$ $\iconft{\bar T'}{\{\bar{N}'^m\} \cup P'}{\bar{S}'}$.  We start by observing that this implies $F = \bar{F} \klmeet \extrf{\cE{\bar{E}}}$, $M'=\cC{\bar{E}}{\bar{M}'}$, $P = P'$, ${\bar{T}'} = {T'}$ and ${\bar{S}'} = {S'}$.  
We can assume, without loss of generality, that $\bar{N}$ is the smallest in the sense that there is no $\cE{\hat{E}},\hat{N}$ such that $\cE{\hat{E}}\neq[]$ and $\cC{\hat{E}}{\hat{N}}=\bar{N}$ for which we can write $W \semvdash \iconft{T}{\{\hat{N}^m\}}{S}$ $\xarr{\hat{F}}{d}$ $\iconft{T'}{\{\hat{N}'^m\} \cup P}{S'}$.  %

By Replacement (Lemma~\ref{app-prop-decleffect-confinement-repllemma}), we have $\typenv \byinov{\secmap{}}{} {{\bar{M}}} \hookrightarrow \bar{N}: \bar s, \bar{\tau}$, for some $\bar M$, $\bar s$' $\bar \tau$ such that $s \klpreceq \bar s$, in the proof of $\typenv \byinov{\secmap{}}{} {M} \hookrightarrow {N}: s, {\tau}$.
We proceed by case analysis on the transition $W \semvdash \iconft{T}{\{{\bar{N}}^m\}}{S}$ $\xarr{\bar{F}}{d}$ $\iconft{T'}{\{\bar{N}'^m\} \cup P}{S'}$, and prove that:
\begin{itemize}
\item There exist $\bar M'$, $\bar s'$ and $\bar \tau'$ such that $\typenv\byinov{\secmap{}}{} {\bar M'} \hookrightarrow {\bar N'}: {\bar s'}, {\bar \tau'}$, and $\bar s \klmeet W(T(m)) \klpreceq \bar s'$ and $\bar \tau \klpreceqtyp \bar \tau'$.  Furthermore, for every reference $a \in \dom{S'}$ implies ${\typenv \byint{\kltop}{\secmap{}}{} V \hookrightarrow S'(a) : \secmap{2}(a)}$ for some value $V$.
\item If~$P=\{N''^n\}$~for some expression $N''$ and thread name $n$, then there exist $M''$ and $\bar s''$ such that $\typenv\byinov{\secmap{}}{} {M''} \hookrightarrow {N''}: {\bar s''}, {\unit}$, and $W(T'(n)) \klpreceq \bar s'$.  (Note that in this case $S=S'$.)
\end{itemize}

By case analysis on the structure of $\bar N$:
\begin{description}
\item [$\boldsymbol{\bar{N} = \app {\lam x {\hat{N}}} {V_2}}$]

Here we have $\bar{N}' = \substi x {V_2} {\hat{N}}$, $S=S'$ and $P=\emptyset$.  By rule \apptypI, there exist $\hat M$, $V_1$, $\hat{s_1}$, $\hat{s_2}$, $\hat{s_3}$, $\hat{\tau}$, $\hat{\sigma}$ and $\hat{\tau}''$ such that $\typenv \byinov{\secmap{}}{} {\lam x {\hat{M}}} \hookrightarrow  {\lam x {\hat{N}}}: \hat{s_1}, \hat{\tau} \xarr{}{\hat{s_3}} \hat{\sigma}$ and $\typenv \byinov{\secmap{}}{} {{V_1}} \hookrightarrow {{V_2}} : \hat{s_2}, \hat{\tau}''$ with $\bar s = \hat{s_1} \klmeet \hat{s_2} \klmeet \hat{s_3}$, $\hat \tau \klpreceqtyp \hat \tau''$ and $\bar{\tau} = \hat{\sigma}$.
By \abstypI, then $\typenv,x:\hat{\tau} \byinov{\secmap{}}{} {{\hat{M}}} \hookrightarrow {{\hat{N}}} : \hat{s_3}, \hat{\sigma}$. %
Therefore, by Lemma~\ref{app-prop-decleffect-confinement-subslemma}, we get $\typenv \byinov{\secmap{}}{} {\substi x {V_1} {\hat{M}}} \hookrightarrow {\substi x {V_2} {\hat{N}}} : \hat{s_3}', {\hat \sigma'}$ with $\hat{s_3} \klpreceq \hat{s_3}'$ and $\hat{\sigma} \klpreceqtyp \bar{\tau}'$.  We take $\bar s' = \hat{s_3}$ and $\bar \tau' = \hat \sigma'$.
\item [$\boldsymbol{\bar{N} = \flow {\hat F} {V_2}}$]

Here we have $\bar{N}' = V_2$, $S=S'$ and $P=\emptyset$.  By rule \flowtypI~and by Remark~\ref{app-rem-valTransf}, there exist $V_1$, $\hat{s}$, such that $\typenv \byinov{\secmap{}}{} V_1 \hookrightarrow V_2 : \kltop, {\bar \tau}$.  We take $\bar s' = \kltop$.
\item [$\boldsymbol{\bar{N} = \allowed {\hat F} {N_t}{N_f}}$ and $\boldsymbol{W(T(m)) \klpreceq \hat F}$]

Here we have $\bar{N}' = N_t$, $S=S'$ and $P=\emptyset$.  By \allowtypI, there exist $M_t$, $N_t$, $\hat{s}_t$, $\hat{s}_f$, $\hat{\tau}_t$ and $\hat{\tau}_f$ such that $\typenv \byinov{\secmap{}}{} {M_t} \hookrightarrow {N_t} : \hat{s}_t, \hat{\tau}_t$, and $\typenv \byinov{\secmap{}}{} {M_f} \hookrightarrow {N_f} : \hat{s}_f, \hat{\tau}_f$, where $\bar s = \hat{s}_t \klpseudominus \hat F \klmeet \hat{s}_f$, $\tau_t \kleqtyp \tau_f$ and $\bar{\tau} = \hat{\tau}_t \klmeet \hat{\tau}_f$.  We take $\bar s' = \hat{s}_t$ and $\bar \tau' = \hat{\tau}_t$.  Notice that $\bar s \klmeet W(T(m)) = \hat{s}_t \klpseudominus \hat F \klmeet \hat{s}_f \klmeet W(T(m)) \klpreceq \hat{s}_t$ and $\bar \tau \klpreceqtyp \bar \tau'$.  
\item [$\boldsymbol{\bar{N} = \threadanot k {\hat N} d {\dot s}}$]

Here we have $\bar{N}' = \nil$, $S=S'$, $P=\{\hat N^n\}$ for some thread name $n$, $T'(n)=d$ and $W(d) \klpreceq {\dot s}$.  By \migtypI, we have that there exists $\hat M$ such that $\typenv \byinov{\secmap{}}{} {\hat M} \hookrightarrow {\hat N} : \dot s,\unit$ and $\bar \tau = \unit$, and by \niltypI~we have that $\typenv \byinov{\secmap{}}{} {\nil} \hookrightarrow {\nil} : \kltop, \unit$.
\end{description}

By Replacement (Lemma~\ref{app-prop-decleffect-confinement-repllemma}), we can finally conclude that 
$\typenv\byinov{\secmap{}}{} M'_E \hookrightarrow {\cC{\bar{E}}{\bar{N}'}} : s', {\tau}'$ for some $M'_E$, $s'$, $\tau'$ such that $s \klmeet W(T(m)) \klpreceq s'$ and $\tau \klpreceqtyp \tau'$.
\end{proof}

\begin{prop}[Safety (weakened) -- Proposition~\ref{prop-decleffect-safety}] \label{app-prop-decleffect-safety}
Consider a closed thread $M^{m}$ such that ${\emptyset \byinov{\secmap{}}{} M \hookrightarrow {\hat M}: s,\tau}$.  Then, for any allowed-policy mapping $W$, memory $S$ that is $(W,{\secmap{}},{\emptyset})$-compatible and position-tracker $T$, either the program $\hat M$ is a value, or:
\begin{enumerate}
\item the program $\hat M$ is a value, or \label{case1'}
\item $W \vdash \iconft{T}{\{{\hat M}^{m}\}} S$ $\xarr{\var {F'}}{T(m)}$ $\iconft{T'} {\{{\hat M}'^{m}\} \cup P} {S'}$, for some $F'$, ${\hat M}'$, $P$, $S'$ and $T'$, or
  \label{case2'}
\item $M=\cC{{E}}{\threadnat {l} {N} {d}}$, for some ${E}$, $l$, $d$ and $\hat s$ such that $\hat M=\cC{{E}}{\threadanot {l} {\hat N} {d} {\hat s}}$ and ${\emptyset \byinov{\secmap{}}{} N \hookrightarrow {\hat N}: {\hat s},\tau}$ but ${ W(d) \not\klpreceq \hat{s}}$.
  \label{case3'}
\end{enumerate}
\end{prop}
\begin{proof} By induction on the derivation of ${\typenv \byinov{\secmap{}}{} M \hookrightarrow {\hat M}: s,\tau}$.  If %
$\hat M \in Val'$ then case~\ref{case1'} holds.
\begin{description}
\item [\flowtypI] If $M = \flow {F} {N}$ and $\hat M = \flow {F} {\hat N}$, then ${\typenv \byinov{\secmap{}}{} N \hookrightarrow {\hat N}: s', \tau}$.  By induction hypothesis, then one of the following cases holds for $\hat N$:
  \begin{description}
  \item [Case~\ref{case1'}]  Then case~\ref{case2'} holds for $\hat M$, with ${\hat M}'={\hat N}$, $d=T(m)$ and $F'=F$.
  \item [Case~\ref{case2'}]  Suppose that there exists $F'$, ${\hat N}'$, $P$, $S'$ and $T'$ such that $W \vdash \iconft{T}{\{{\hat N}^{m}\}} S$ $\xarr{\var {F'}}{d}$ $\iconft{T'} {\{{\hat N}'^{m}\} \cup P} {S'}$.  Then $W \vdash \iconft{T}{\{{\hat M}^{m}\}} S$ $\xarr{\var {F' \cup F}}{d}$ $\iconft{T'} {\{{\flow F {{\hat N}'}}^{m}\} \cup P} {S'}$, so case~\ref{case2'} holds for $\hat M$.
  \item [Case~\ref{case3'}]  Suppose that there exists an evaluation context ${\bar E}$, a security level $l$, a domain name $d$ and a security effect $\hat s$ such that $N=\cC{{E}}{\threadnat {l} {N'} {d}}$, $\hat N=\cC{{E}}{\threadanot {l} {{\hat N}'} {d} {\hat s}}$ and ${\typenv \byinov{\secmap{}}{} {N'} \hookrightarrow {{\hat N}'}: {\hat s},\tau}$ but ${ W(d) \not\klpreceq \hat s}$.  Then, case~\ref{case3'} also holds for $\hat M$ for the evaluation context $E={\flow F {\cC{\bar E}{}}}$. 
  \end{description}
\item [$\allowtyp$] If $M = \allowed {F} {N_t} {N_f}$ and $M = \allowed {F} {{\hat N}_t} {{\hat N}_f}$, then we have that $W \vdash \iconft{T}{\{{\hat M}^{m}\}} S$ $\xarr{\emptyset}{T(m)}$ $\iconft{T} {\{{\hat M}'^{m}\}} {S}$, with ${\hat M}'={\hat N}_t$ of ${\hat M}'={\hat N}_f$ depending on whether $W(T(m)) \klpreceq F$ or $W(T(m)) \klpreceq F$ (respectively).  Therefore, case~\ref{case2'} holds for $\hat M$.
\item [$\migtyp$] If $M = \threadnat l N d$ and $\hat M = \threadanot l {\hat N} d {\hat s}$, then $\typenv \byinov{\secmap{}}{} {{N}} \hookrightarrow {\hat N} : {\hat s}, {\unit}$, and either:
  \begin{description}
  \item [$\boldsymbol{W(d) \klpreceq \hat{s}}$]  Then, $W \vdash \iconft{T}{\{M^{m}\}} S$ $\xarr{\emptyset}{T(m)}$ $\iconft{T} {\{{\nil}^m,N^{n}\}} {S}$.  Therefore, case~\ref{case2'} holds for $M$.
  \item [$\boldsymbol{W(d) \not\klpreceq \hat{s}}$]  Then, case~\ref{case3'} holds for $M$ for the evaluation context $E={\cE{\bar E}}$.  \qedhere
  \end{description}
\end{description}
\end{proof}

\subsubsection{Preservation of the semantics}  \label{app-deceffect-preservation}

\begin{prop}[Proposition~\ref{prop-preservsem}]
If ${\typenv \byinov{\secmap{}}{} {M} \hookrightarrow {N}: s, \tau}$, then for all allowed-policy mappings $W$ and thread names $m \in \Nam$ we have that $\{M^m\} \realbiseq{W,\secmap{},{\typenv}} \{N^m\}$.
\end{prop}
\begin{proof}
We prove that, for all allowed-policy mappings $W$, the set
\[
B = \{\confd{\{M^m\}}{\{N^m\}} ~|~ m \in \Nam \textit{ and } \exists s,\tau ~.~ {\typenv \byinov{\secmap{}}{} {M} \hookrightarrow {N}: s, \tau}\}
\]
is a $({W},{\secmap{}},{\typenv})$-simulation according to Definition~\ref{def-realbisim}.

Let us assume that $W \semvdash \iconft{T}{\{N^m\}}{{\hat S}}$ $\xarr{F}{d}$ $\iconft{T'}{\{N'^m\} \cup {P}_2'}{{\hat S}'}$.  Consider $\cE{\bar{E}_2},\bar N$ such that $N=\cC{\bar{E}_2}{\bar{N}}$ and $W \semvdash \iconft{T}{\{{\bar{N}}^m\}}{{\hat S}}$ $\xarr{\bar{F}}{d}$ $\iconft{\bar T'}{\{\bar{N}'^m\} \cup {\bar P}_2'}{\bar{{\hat S}}'}$ and $N'=\cC{\bar{E}_2}{\bar{N}'}$.  We start by observing that $F = \bar{F} \klmeet \extrf{\cE{\bar{E}_2}}$, $N'=\cC{\bar{E}_2}{\bar{N}'}$, $P_2' = {\bar P}_2'$, ${\bar{T}'} = {T'}$ and ${\bar{{\hat S}}'} = {{\hat S}'}$.  
We can assume, without loss of generality, that $\bar{N}$ is the smallest in the sense that there is no $\cE{\hat{E}_2},\hat{N}$ such that $\cE{\hat{E}_2}\neq[]$ and $\cC{\hat{E}_2}{\hat{N}}=\bar{N}$ for which we can write $W \semvdash \iconft{T}{\{\hat{N}^m\}}{{\hat S}}$ $\xarr{\hat{F}}{d}$ $\iconft{T'}{\{\hat{N}'^m\} \cup P_2'}{{\hat S}'}$.  %

By Replacement (Lemma~\ref{app-prop-decleffect-confinement-repllemma}), we have $\typenv \byinov{\secmap{}}{} {{\bar{M}}} \hookrightarrow \bar{N}: \bar s, \bar{\tau}$, for some $\bar M$, $\bar s$, $\bar \tau$ such that $s \klpreceq \bar s$, in the proof of $\typenv \byinov{\secmap{}}{} {M} \hookrightarrow {N}: s, {\tau}$.  Furthermore, there exists $\cE{\bar E_1}$ such that $M = \cC{\bar E_1}{\bar M}$.
We proceed by case analysis on the transition $W \semvdash \iconft{T}{\{{\bar{N}}^m\}}{{\hat S}}$ $\xarr{\bar{F}}{d}$ $\iconft{T'}{\{\bar{N}'^m\} \cup {P}_2'}{{\hat S}'}$, and prove that:
\begin{itemize}
\item if $\annot{S} = \hat S$, there exist $\bar M'$, $S'$ s.t. $W \semvdash \iconft{T}{\{{\bar{M}}^m\}}{S}$ $\xarr{\bar{F}}{d}$ $\iconft{T'}{\{\bar{M}'^m\} \cup {P}_1'}{S'}$ and $\annot{S'} = \hat S'$, and for which there exist $\bar s'$ and $\bar \tau'$ such that $\typenv\byinov{\secmap{}}{} {\bar M'} \hookrightarrow {\bar N'}: {\bar s'}, {\bar \tau'}$.
\item If~$P_2'=\{N''^n\}$~for some expression $N''$ and thread name $n$, then there exist $M''$ and $s''$ such that $P_1'=\{M''^n\}$ and $\typenv\byinov{\secmap{}}{} {M''} \hookrightarrow {N''}: {s''}, {\unit}$. %
\end{itemize}

By case analysis on the structure of $\bar N$:
\begin{description}
\item [$\boldsymbol{\bar{N} = \flow {\hat F} {V_2}}$]

Here we have $\bar{N}' = V_2$, $S=S'$, $T=T'$ and $P_2'=\emptyset$.  By rule \flowtypI~and Remark~\ref{app-rem-valTransf}, there exist $V_1$, $\hat{s}$, such that $\bar M = \flow {\hat F} {V_1}$ and $\typenv \byinov{\secmap{}}{} V_1 \hookrightarrow V_2 : \hat s, {\bar \tau}$.  We then have $W \vdash \iconft {T}{\{{\bar M}^m\}}{S} \xarr{F}{d} \iconft{T}{\{{V_1}^m\}}{S}$, so we take $\bar M' = {V_1}$. %
\item [$\boldsymbol{\bar{N} = \allowed {\hat F} {N_t}{N_f}}$ and $\boldsymbol{W(T(m)) \klpreceq \hat F}$]

Here we have that $\bar{N}' = N_t$, $S=S'$, $T=T'$ and $P_2'=\emptyset$.  By \allowtypI, there exist $M_t$, $M_f$, $\hat{s}_t$, $\hat{s}_f$, $\hat{\tau}_t$ and $\hat{\tau}_f$ such that $\bar M = \allowed {\hat F} {M_t}{M_f}$ and $\typenv \byinov{\secmap{}}{} {M_t} \hookrightarrow {N_t} : \hat{s}_t, \hat{\tau}_t$, and $\typenv \byinov{\secmap{}}{} {M_f} \hookrightarrow {N_f} : \hat{s}_f, \hat{\tau}_f$, where $\bar s = \hat{s}_t \klpseudominus \hat F \klmeet \hat{s}_f$, $\tau_t \kleqtyp \tau_f$ and $\bar{\tau} = \hat{\tau}_t \klmeet \hat{\tau}_f$.  We take $\bar s' = \hat{s}_t$ and $\bar \tau' = \hat{\tau}_t$.   We then have $W \vdash \iconft {T}{\{{\bar M}^m\}}{S} \xarr{F}{d} \iconft{T}{\{{M_t}^m\}}{S}$, so we take $\bar M' = {M_t}$. %
\item [$\boldsymbol{\bar{N} = \threadanot k {\hat N} d {\dot s}}$]

Here we have $\bar{N}' = \nil$, $S=S'$, $P=\{\hat N^n\}$ for some thread name $n$, $T'(n)=d$ and $W(d) \klpreceq {\dot s}$.  By \migtypI, we have that there exists $\hat M$ such that $\bar M = \threadnat k {\hat M} d$ and $\typenv \byinov{\secmap{}}{} {\hat M} \hookrightarrow {\hat N} : \dot s,\unit$ and $\bar \tau = \unit$, and by \niltypI~we have that $\typenv \byinov{\secmap{}}{} {\nil} \hookrightarrow {\nil} : \kltop, \unit$.  %
Therefore, $W \vdash \iconft {T}{\{{\bar M}^m\}}{S} \xarr{F}{d} \iconft{T}{\{{\nil}^m,{\hat M}^n\}}{S}$, so we take $\bar M' = {\nil}$ and $M''=\hat M$. %
\end{description}

We then have that $W \semvdash \iconft{T}{\{M^m\}}{S} \xarr{F}{d} \iconft{T'}{\{M'^m\} \cup P_1'}{{S'}}$ with $M'=\cC{\bar E_1}{\bar M'}$.  
By Replacement (Lemma~\ref{app-prop-decleffect-confinement-repllemma}), we can conclude that 
$\typenv\byinov{\secmap{}}{} \cC{\bar E_1}{\bar M'} \hookrightarrow {\cC{\bar E_2}{\bar{N}'}} : s', {\tau}'$ for some $s'$, $\tau'$. %
Therefore, $\{M'^m\} B \{N'^m\}$, and if $P_2'=\{N''^n\}$, then $P_1'=\{M''^n\}$ and $\{M'^m,M''^n\} B \{N'^m,N''^n\}$.
\end{proof}

\subsubsection{Soundness}

\begin{thm}[Soundness of Enforcement Mechanism~III -- Theorem~\ref{prop-soundness-declassif}] \label{app-prop-soundness-declassif}
  Consider an allowed-policy mapping $W$, reference labeling $\secmap{}$, typing environment $\typenv$, and a thread configuration $\confd P T$ such that for all $M^{m} \in P$ there exist $\hat M$, $s$ and $\tau$ such that ${\typenv \byinov{\secmap{}}{} M \hookrightarrow \hat{M}: s,\tau}$ and ${W(T(m))} \klpreceq s$.  Then ${\confd {\hat P} T}$, formed by annotating the threads in $\confd P T$, is $(W,\secmap{},\typenv)$-confined.
\end{thm}
\begin{proof} 
Consider the following set:
\myexample{
C = \{ \confd{P}{T} ~|~ \forall {\hat M}^m \in P,~ \exists M,s,\tau ~.~ {\typenv \byinov{\secmap{}}{} M \hookrightarrow \hat M: s,\tau} \text{ and } W(T(m)) \klpreceq s\}   %
}
We show that $C$ is a set of $(W,\secmap{},\typenv)$-confined thread configurations.
As in the proof of Theorem~\ref{app-prop-confinementI-soundness}, if for a given $(W,\secmap{},\typenv)$-compatible store $S$ we have that there exist $P',T',S'$ such that $W \vdash \iconft{T}{P} S$ $\xarr{F}{d}$ $\iconft{T'} {P'} {S'}$, then, there is a thread $M^m$ such that $P = \{M^m\} \cup \bar P$ and $W \vdash \iconft{T}{\{M^{m}\}} S$ $\xarr{F}{d}$ $\iconft{T'} {\{M^{m}\} \cup \bar{P}'} {S'}$, with $P' = {\{M^{m}\} \cup \bar{P}'} \cup \bar P$ and $T(m)=d$. %
By induction on the inference of ${\typenv \byinov{\secmap{}}{} M \hookrightarrow \hat M: s,\tau}$, we prove that $W(d) \klpreceq F$ and there exists $M'$, $s'$ and $\tau'$ such that ${\typenv \byinov{\secmap{}}{} M' \hookrightarrow \hat M': s',\tau'}$ with $W(T'(m)) \klpreceq s'$.
Furthermore, if $\bar P =N^n$ for some expression $N$ and thread name $N$, then there exists $M'$, $s'$, $\tau'$ such that ${\typenv \byinov{\secmap{}}{} M' \hookrightarrow \hat M': s',\tau'}$  with $W(T'(m)) \klpreceq s$, and there exists $N$, $s''$ such that ${\typenv \byinov{\secmap{}}{} N \hookrightarrow \hat N: s'',\unit}$ with $W(T'(n)) \klpreceq s''$.

Notice that since we necessarily have $T'(m) = T(m) = d$, then typability of the threads that result from the transition step is guaranteed by Subject Reduction (Proposition~\ref{prop-preservanot}).  Also, by the same result, we have compliance of the new declassification effect $s'$ to the domain's allowed flow policy, for:  If $s \klmeet W(T(m)) \klpreceq s'$, and since $W(d) = W(T(m)) \klpreceq s$, then $W(T(m)) \klpreceq s'$.  Furthermore, if $\bar P = N^n$, then %
$W(T'(n)) \klpreceq s''$.

It remains to prove the conditions regarding the compliance of the declared flow policies to the current domain's allowed flow policy.
Assuming that ${\typenv \byinov{\secmap{}}{} M \hookrightarrow N:s, \tau}$ and $W(T(m)) \klpreceq s$ and by case analysis on the last rule in the corresponding typing proof:
\begin{description}
\item [\flowtypI]

Here $M = \flow {\bar F} {\bar{M}}$, $N = \flow {\bar F} {\bar{N}}$, and ${\typenv \byinov{\secmap{}}{} {\bar M} \hookrightarrow {\bar N} : \bar s, \tau}$, with 
$s = \bar s \klmeet \bar F$.  Then, we have that $W(d) \klpreceq \bar s$.  %
There are two cases to consider:
  \begin{description}
  \item [$\bar N$ can compute] Then %
$W \semvdash \iconft{T}{\{{\bar N}^m\}} S$ $\xarr{{\bar F}'}{d}$ $\iconft{T'} {\{{{\bar N}'^m}\} \cup P} {S'}$, with $F = \bar F \klmeet \bar F'$.  %
By induction hypothesis, then $W(T(m)) \klpreceq {\bar F}'$. %
Since $W(T(m))=W(d) \klpreceq {\bar F}$, then $W(T(m)) \klpreceq F$.  %
  \item [$\bar N \in \Val$] Then we have $F = \kltop$, so $W(T(m)) \klpreceq F$ holds vacuously. %
  \end{description}

\item [\allowtypI]
Here $N = {\allowed {\bar A} {\bar{N_t}} {\bar{N_f}}}$, and we have %
$F = \kltop$.  Therefore $W(T(m)) \klpreceq F$ holds vacuously.

\item [\migtypI]

In this case we have $M = \threadnat l {\bar{M}} {\bar d}$ and $N = \threadanot l {\bar{N}} {\bar d} {\bar s}$, with $\typenv \byinov{\secmap{}}{}$ ${\bar M} \hookrightarrow {\bar N} :$ $\bar s, \unit$.  Then we have $F = \kltop$, so $W(T(m)) \klpreceq F$ holds vacuously.
 Since $N$ can reduce, then $W(\bar d) = W(T'(n)) \klpreceq \bar s$.

\end{description}

The cases for \reftypI, \dereftypI, \assigntypI, \seqtypI, \apptypI~and \condtypI~are similar to \flowtypI, since for the cases where the sub-expressions are not all values, these sub-expressions are typed with a declassification effect $s'$ that is at least as restrictive as that $s'$ of the concluding judgment, i.e. $s \klpreceq s'$, and so since $W(d) \klpreceq s'$ then the induction hypothesis can be applied.  Furthermore, for the cases where the sub-expression is ready to reduce, the flow policy that decorates the transition is the same as $F$.  When the sub-expressions are all values (more precisely, those that are to be evaluated in the evaluation context provided by each of these constructs), then $F = \kltop$.
The case for \rectypI~is similar to \allowtypI, since the constructs are not evaluation contexts, and therefore the transition is decorated with the top flow policy $F = \kltop$.
\end{proof}

\subsubsection{Comparison of Enforcement Mechanisms~II and~III}

\begin{prop}[Proposition~\ref{prop-checkinform-comparison}] \label{app-prop-checkinform-comparison} \label{app-prop-equivalencetypsys}
  If for an allowed policy $A$ and type $\tau$ we have that ${\typenv \byund{A}{\secmap{}}{} M : \tau}$, then there exist $N$, $s$ and $\tau'$ such that ${\typenv \byinov{\secmap{}}{} M \hookrightarrow N: s, \tau'}$ with $A \klpreceq s$ and $\tau \klpreceq \tau'$.
\end{prop} 
\begin{proof}
By induction on the inference of $\typenv\byund{A}{\secmap{}}{} M : {\tau}$, and by case analysis on the last rule used in this typing proof. 
\begin{description} 
{
\item [\niltyp]

Here $M = \nil$ and $\tau = \unit$.  We conclude using \niltypI, $N=\nil$, $s=\kltop$ and $\tau' = \unit$.

\item [\vartyp]

Here $M=x$ and $\tau=\typenv(x)$.  We conclude using \vartypI, $N=x$, $s=\kltop$ and $\tau' = \typenv(x)$.

\item [\abstyp]

Here $M = \lam x {\bar{M}}$ and ${\typenv,{x:\bar{\theta}} \byund{\bar A}{\secmap{}}{} \bar{M} : \bar \sigma}$ where ${\tau} = \bar{\theta} \xarr{}{\bar A} \bar{\sigma}$.  
By induction hypothesis, there exist $\bar N$, $\bar s$ and $\bar \sigma'$ such that $\typenv,{x:\bar{\theta}} \byinov{\secmap{}}{} {\bar{M}} \hookrightarrow {\bar{N}}: \bar s, \bar{\sigma}'$ and $\bar A \klpreceq \bar s$ with $\bar\sigma \klpreceq \bar\sigma'$.
We conclude using \abstypI, $N = \lam x {\bar{N}}$, $s = \kltop$, and $\tau' = \bar{\theta} \xarr{}{\bar s} \bar{\sigma}'$.

\item [\rectyp] 

Here $M = \fix x {\bar{X_1}}$, and $\typenv,{x:\tau} \byund{A}{\secmap{}}{} {\bar{X_1}} : \tau$.  
By induction hypothesis, there exist $\bar {X_2}$, $\bar s$ and $\bar \tau '$ such that $\typenv,{x:\tau} \byinov{\secmap{}}{} {\bar{X_1}}\hookrightarrow  \bar {X_2}: \bar s, \bar \tau'$ and $A \klpreceq \bar s$ with $\tau \klpreceq \bar \tau'$.  
We conclude using \rectypI, $N = {\fix x {\bar{X_2}}}$, $s = \bar s$, and $\tau' = \bar \tau'$.  %

\item [\condtyp]

Here $M = \cond {\bar{M}} {{\bar M}_t} {{\bar M}_f}$ and we have $\typenv \byund{A}{\secmap{}}{} {\bar{M}} : \bool$, $\typenv \byund{A}{\secmap{}}{} \bar {M_t} : \bar{\tau}$ and $\typenv \byund{A}{\secmap{}}{} \bar {M_f} : \bar{\tau}$ where $\bar \tau = \tau$.  
By induction hypothesis, there exist $\bar {N}$, $\bar {N_t}$, $\bar {N_f}$, $\bar s$, ${\bar s}_t$, ${\bar s}_f$, $\bar \tau_t$ and $\bar \tau_f$ such that $\typenv \byinov{\secmap{}}{} {\bar M} \hookrightarrow \bar N: \bar s, \bool$, $\typenv \byinov{\secmap{}}{} {{\bar M}_t} \hookrightarrow {\bar {N_t}}: {\bar s}_t,\bar {\tau_t}$ and $\typenv \byinov{\secmap{}}{} {{\bar M}_f} \hookrightarrow {\bar {N_f}} : \bar{s_f}, \bar {\tau_f}$ and $A \klpreceq \bar s$, $A \klpreceq {\bar s}_t$, $A \klpreceq {\bar s}_f$ with $\bar \tau \klpreceq \bar \tau_t$ and $\bar \tau \klpreceq \bar \tau_f$.
Therefore, $A \klpreceq \bar s \klmeet {{\bar s}_t}' \klmeet {{\bar s}_f}'$, and we conclude using rule \condtypI, $N=\cond {\bar{N}} {{\bar N}_t} {{\bar N}_f}$, $s = \bar s \klmeet {{\bar s}_t}' \klmeet {{\bar s}_f}'$ and $\tau' = \bar{\tau_t} \klmeet \bar{\tau_f}$.

\item [\reftyp]

Here $M = \rfr {\theta} {\bar M}$ and we have $\typenv \byund{A}{\secmap{}}{} \bar {M} : {\theta}$ where $\tau = \rfrt{{\theta}}{}$.  
By induction hypothesis, there exist $\bar {N}$, $\bar {s}$ and $\theta'$ such that $\typenv \byinov{\secmap{}}{} {\bar {M}} \hookrightarrow {\bar {N}}: \bar {s},{\theta}'$ and $A \klpreceq \bar {s}$ with ${\theta} \klpreceq \theta'$.
We conclude using rule \reftypI, $N=\rfr {}{\bar N}$, $s = \bar{s}$ and $\tau' = \rfrt{\theta}{}$.

\item [\dereftyp]

Here $M = \deref {\bar M}$ and we have $\typenv \byund{A}{\secmap{}}{} \bar {M} : \bar\tau$ with $\bar\tau = \rfrt{{\tau}}{}$.
By induction hypothesis, there exist $\bar {N}$, $\bar {s}$ and $\bar\tau'$ such that $\typenv \byinov{\secmap{}}{} {\bar {M}} \hookrightarrow {\bar {N}}: \bar {s},{\bar{\tau}'}$ and $A \klpreceq \bar {s}$ with ${\bar \tau} \klpreceq \bar\tau'$.  
Notice that $\bar \tau' = \rfrt{\tau}{}$.  %
We conclude using rule \dereftypI, $N=\deref {\bar N}$, $s = \bar{s}$ and $\tau' = \tau$.  %

\item [\assigntyp]

Here $M = \assign {\bar M_1} {\bar M_2}$ and we have $\typenv \byund{A}{\secmap{}}{} \bar {M_1} : \rfrt{\bar \theta}{}$ and $\typenv \byund{A}{\secmap{}}{} \bar {M_2} : \bar{\theta}$ where $\tau = \unit$.  
By induction hypothesis, there exist $\bar {N_1}$, $\bar {N_2}$, $\bar {s_1}$, $\bar {s_2}$, $\bar \tau'$, and $\bar \theta'$ such that $\typenv \byinov{\secmap{}}{} {\bar {M_1}} \hookrightarrow {\bar {N_1}}: \bar {s_1},\bar {\tau}'$ and $\typenv \byinov{\secmap{}}{} {\bar {M_2}} \hookrightarrow {\bar {N_2}} : \bar{s_2}, \bar {\theta}'$ and $A \klpreceq \bar {s_1}'$, $A \klpreceq \bar {s_2}'$ with $\rfrt{\bar \theta}{} \klpreceq \bar \tau'$ and $\bar{\theta} \klpreceq \bar{\theta}'$. 
Notice that $\bar \tau' = \rfrt{\bar \theta}{}$.  %
Therefore, $A \klpreceq \bar{s_1}' \klmeet \bar{s_2}'$, and we conclude using rule \apptypI, $N=\app {\bar N_1} {\bar N_2}$, $s = \bar{s_1}' \klmeet \bar{s_2}'$ and $\tau' = \unit$.  %

\item [\apptyp]

Here $M = \app {\bar M_1} {\bar M_2}$ and we have $\typenv \byund{A}{\secmap{}}{} \bar {M_1} : \bar{\theta} \xarr{}{A} \bar{\sigma}$ and $\typenv \byund{A}{\secmap{}}{} \bar {M_2} : \bar{\theta}$ where $\tau = \bar{\sigma}$.  
By induction hypothesis, there exist $\bar {N_1}$, $\bar {N_2}$, $\bar {s_1}$, $\bar {s_2}$, $\bar \tau'$, and $\bar \theta'$ such that $\typenv \byinov{\secmap{}}{} {\bar {M_1}} \hookrightarrow {\bar {N_1}}: \bar {s_1},\bar {\tau}'$ and $\typenv \byinov{\secmap{}}{} {\bar {M_2}} \hookrightarrow {\bar {N_2}} : \bar{s_2}, \bar {\theta}'$ and $A \klpreceq \bar {s_1}'$, $A \klpreceq \bar {s_2}'$ with $\bar{\theta} \xarr{}{A} \bar{\sigma} \klpreceq \bar \tau'$ and $\bar{\theta} \klpreceq \bar{\theta}'$.  Notice that $\bar \tau' = \bar{\theta} \xarr{}{A''} \bar{\sigma}''$, for some $A''$ and $\bar{\sigma}''$ such that $A \klpreceq A''$ and $\bar \sigma \klpreceq \bar{\sigma}''$.
Therefore, $A \klpreceq \bar{s_1}' \klmeet \bar{s_2}' \klmeet A''$, and we conclude using rule \apptypI, $N=\app {\bar N_1} {\bar N_2}$, $s = \bar{s_1}' \klmeet \bar{s_2}' \klmeet A''$ and $\tau' = \bar{\sigma}''$.    %
}
\item [\migtyp]

Here $M = \threadnat l {\bar M} {d}$ and we have that $\typenv \byund{\klbot}{\secmap{}}{} {\bar{M}} : \bar \tau$, with $\tau = \bar \tau =\unit$.  
By induction hypothesis,  there exist $\bar {N}$, $\bar s$, $\bar \tau'$ such that $\typenv \byinov{\secmap{}}{} {\bar{M}} \hookrightarrow {\bar N}: \bar s, \bar \tau'$ and $\klbot \klpreceq \bar s$ with $\bar \tau \klpreceq \bar \tau'$.  
We conclude using rule \migtypI, $N= \threadanot l {\bar{N}} d {\bar s}$, $s = \kltop$ and $\tau' = \unit$.

\item [\flowtyp]

Here $M = \flow {\bar F} {\bar{M}}$ and $\typenv \byund{A}{\secmap{}}{} {\bar{M}} : \bar \tau$, with  $A \klpreceq {\bar F}$ and $\tau = \bar \tau$.  By induction hypothesis, there exist $\bar {N}$, $\bar s$, $\bar \tau'$ such that $\typenv \byinov{\secmap{}}{} {\bar{M}} \hookrightarrow \bar N: \bar s, \bar \tau'$  and $A \klpreceq \bar s$ with $\bar \tau \klpreceq \bar \tau'$.  
Therefore $A \klpreceq \bar s \klmeet \bar F$ and we conclude using rule \flowtypI, $N= {\flow {\bar{F}} {\bar{N}}}$, $s = \bar s \klmeet \bar F$ and $\tau' = \bar\tau'$.

\item [\allowtyp]

Here $M = \allowed {\bar{F}} {{\bar M}_t} {{\bar M}_f}$ and we have $\typenv \byund{A \klmeet \bar F}{\secmap{}}{} \bar {M_t} : \bar{\tau}$ and $\typenv \byund{A}{\secmap{}}{} \bar {M_f} : \bar{\tau}$ where $\bar \tau = \tau$.  
By induction hypothesis, there exist $\bar {N_t}$, $\bar {N_f}$, ${\bar s}_t$, ${\bar s}_f$ and $\bar \tau'$ such that $\typenv \byinov{\secmap{}}{} {{\bar M}_t} \hookrightarrow {\bar {N_t}}: {\bar s}_t,\bar {\tau}'$ and $\typenv \byinov{\secmap{}}{} {{\bar M}_f} \hookrightarrow {\bar {N_f}} : \bar{s_f}, \bar {\tau}'$ and $A \klmeet \bar F \klpreceq {\bar s}_t'$, $A \klpreceq {\bar s}_f'$ with $\bar \tau \klpreceq \bar \tau'$.
Therefore, $A \klpreceq {{\bar s}_t}' \klpseudominus \bar F  \klmeet {{\bar s}_f}'$, and we conclude using rule \allowtypI, $N=\allowed {\bar{F}} {{\bar N}_t} {{\bar N}_f}$, $s = {{\bar s}_t}' \klpseudominus \bar F \klmeet {{\bar s}_f}'$ and $\tau' = \bar{\tau}'$.

\end{description}
{The proofs for the cases $\loctyp$, $\boolttyp$ and $\boolftyp$ are analogous to the one for $\niltyp$, while the proof for $\seqtyp$ is analogous to the one for $\condtyp$.}
\end{proof}

\end{document}